\documentclass{vldb}
\usepackage{graphicx}
\usepackage{balance}  

\usepackage{amsmath}
\usepackage{mathrsfs}
\usepackage[lined,boxed,commentsnumbered,ruled]{algorithm2e}
\usepackage{algorithmic}
\usepackage{color}
\usepackage{subfigure}
\usepackage{makecell}
\usepackage{bm}
\usepackage{url}
\usepackage{hyperref}
\usepackage{mathtools}
\usepackage{balance}

\usepackage{nopageno}

\newtheorem{proposition}{Proposition}
\newtheorem{definition}{Definition}
\newtheorem{example}{Example}
\newtheorem{theorem}{Theorem}

\newcommand{\titlepa}[1]{{\fontsize{#1}{#1}\selectfont \text{$(\mathbf{p},\alpha)$}}}

\newcommand{\mc}[1]{#1} 

\newcommand{\todo}[1]{[\textcolor{red}{{\sc Todo:} #1}]} 

\newcommand{\nop}[1]{}

\DeclarePairedDelimiter\floor{\lfloor}{\rfloor}

\vldbTitle{Finding Theme Communities from Database Networks}
\vldbAuthors{Lingyang Chu, Zhefeng Wang, Jian Pei, Yanyan Zhang, Yu Yang, Enhong Chen}
\vldbDOI{https://doi.org/10.14778/3339490.3339492}
\vldbVolume{12}
\vldbNumber{10}
\vldbYear{2019}

\begin{document}
\sloppy

\title{Finding Theme Communities from Database Networks}



%
%
%
%

\numberofauthors{3} 

\author{
%
%
\alignauthor
Lingyang Chu\\
       \affaddr{Simon Fraser University}\\
       \affaddr{Burnaby, Canada}\\
       \email{lca117@sfu.ca}
\alignauthor
Zhefeng Wang\\
       \affaddr{Huawei Technologies}\\
       \affaddr{China}
       \email{wangzhefeng@huawei.com}
\alignauthor
Jian Pei\\
       \affaddr{Simon Fraser University}\\
       \affaddr{Burnaby, Canada}\\
       \email{jpei@cs.sfu.ca}
\and  
\alignauthor 
Yanyan Zhang\\
       \affaddr{Simon Fraser University}\\
       \affaddr{Burnaby, Canada}\\
       \email{yanyanz@sfu.ca}
\alignauthor 
Yu Yang\\
       \affaddr{Simon Fraser University}\\
       \affaddr{Burnaby, Canada}\\
       \email{yya119@sfu.ca}
\alignauthor 
Enhong Chen\\
       \affaddr{Univ. of Science and Tech. of China (Hefei, China)}\\
       \email{cheneh@ustc.edu.cn}
}

\maketitle

\begin{abstract}
Given a database network where each vertex is  associated with a transaction database, we are interested in finding theme communities. 
Here, a theme community is a cohesive subgraph such that a common pattern is frequent in all transaction databases associated with the vertices in the subgraph.
Finding all theme communities from a database network enjoys many novel applications.  However, it is challenging since even counting the number of all theme communities in a database network is \#P-hard.
Inspired by the observation that a theme community shrinks when the length of the pattern increases, we investigate several properties of theme communities and develop TCFI, a scalable algorithm that uses these properties to effectively prune the patterns that cannot form any theme community.
We also design TC-Tree, a scalable algorithm that decomposes and indexes theme communities efficiently.  
Retrieving a ranked list of theme communities from a TC-Tree of hundreds of millions of theme communities takes less than 1 second.
Extensive experiments and a case study demonstrate the effectiveness and scalability of TCFI and TC-Tree in discovering and querying meaningful theme communities from large database networks.
\end{abstract}

\maketitle

\section{Introduction}
\label{sec:intro}

Finding communities from large networks is a fundamental data mining problem that enjoys various applications, such as targeted advertisement in e-commerce networks~\cite{assal2010multi,huang2013commerce}, friend recommendation in social networks~\cite{huang2015method,mislove2010you} and research group discovery in co-author networks~\cite{yang2015defining}.

\nop{
By modelling real world networks as \emph{simple networks} that contain only vertices and edges, conventional methods, such as graph partitioning~\cite{shi2000normalized,ng2002spectral}, dense subgraph mining~\cite{chu2015alid}, and truss detection~\cite{cohen2008trusses,wang2012truss}, detect communities by finding sets of vertices that have dense edge connections~\cite{fortunato2010community}.
}

Conventional community detection methods, such as graph partitioning~\cite{shi2000normalized,ng2002spectral}, dense subgraph mining~\cite{chu2015alid,pavan2007dominant} and truss detection~\cite{cohen2008trusses,wang2012truss}, model real world networks as \emph{simple networks} that contain only graph structures.
To enhance the model of simple network, \emph{vertex attributed networks} further profile the attributes of each vertex by a set of items~\cite{steinhaeuser2008community,zhou2009graph}.
Community detection in vertex attributed networks aims to find communities such that all vertices in the same community contain the same set of items and are densely connected~\cite{berlingerio2013abacus,prado2013mining}.
Due to the homogeneity of vertex attributes, these communities are usually more meaningful and accurate~\cite{hu2013utilizing,moser2009mining}.

\nop{
These methods detect communities by finding sets of vertices that are densely connected~\cite{fortunato2010community}.
}

\nop{To characterize vertex properties, \emph{vertex attributed network} models network structure and further profiles each vertex by a set of attributes~\cite{gunnemann2011db,steinhaeuser2008community,zhou2009graph}.}

However, in most real world networks, the items of a vertex are not equally important and often do not co-occur all together. 
The valuable vertex information, such as item co-occurrence and the frequency of co-occurring items, is much beyond the limited descriptive power of the single set of items associated with the vertex.
To tackle this problem, in this paper, we propose to model real world networks as \emph{DataBase Networks} (DBN), where every vertex is associated with a transaction database named \emph{vertex database}. 

DBN is a natural descriptive model for many real world networks.
For example, in e-commerce networks where each vertex represents a customer, every set of items purchased together by the customer is recorded as a transaction, and a vertex database stores all transactions of the customer.
In location-based social networks where each vertex represents a user, the set of locations that the user checks in during a period (e.g., a day, a week or a month) can be recorded as a transaction, and the user's all transactions form a vertex database.
The co-author network can also be enhanced by associating with each author a vertex database, where each transaction stores the keywords in an article published by the author.

The vertex databases of DBN accurately describe the co-occurrences and frequencies of different sets of items. A set of co-occurring items is called a \emph{pattern}~\cite{agrawal1994fast,han2000mining}, and a pattern with high frequency in a vertex database dominates the properties of the vertex.
Therefore, in DBNs, it is more interesting to consider item co-occurrences and pattern frequencies, and find communities such that all vertices in the same community share the same dominant pattern and are densely connected. We call such communities \emph{theme communities}, where each \emph{theme} is the dominant pattern of a community.

For example, in a co-author network, authors are vertices and two authors are linked if they collaborated before.  Each author is associated with a transaction database where each transaction is the set of keywords in an article published by the author. In this DBN, a pattern is a set of keywords that describes a research topic, and a theme community represents a group of closely collaborating authors who frequently publish papers in the same research topic.

\nop{
\todo{Instead of using two examples, you may want to clearly explain one.  What are the edges and what are the transactions in a node? What is a pattern? What is a theme community?}\mc{For example, in a location-based social network, a pattern is a set of locations, and a theme community represents a group of best friends who frequently visit the same set of locations.
In co-author networks, a pattern is a set of keywords that specifies a research topic, and a theme community represents a group of closely collaborating authors who frequently publish papers in the same research topic.
}
}
\nop{
a theme community in this database network represent a group of friends who frequently visit the same set of locations.

\todo{It may help an author to understand the idea of theme communities better if we can give a simple yet concrete example here.  The above description is a little bit too abstract and dry.}
}

Can we adapt existing methods straightforwardly to find theme communities in DBNs? Unfortunately, the answer is no due to the following challenges.

First, a vertex database may contain an exponential number of patterns.
Since the conventional methods that work on simple networks can only detect communities of one pattern at a time, it is computationally intractable to call those methods for each of an exponential number of patterns.

Second, the existing methods that work on vertex attributed networks only consider the case where each vertex is associated with a single set of items. In such a case, all items of a vertex occur together and thus the pattern frequencies for all patterns of the same vertex are trivially the same.
Therefore, these methods cannot distinguish the different frequencies of different patterns in DBNs. 

\nop{
Therefore, it is computationally intractable to enumerate theme communities by conventional community detection methods that work in simple networks, because these methods only detect the theme communities of one pattern at a time.
}
\nop{
In a brute force manner, we can apply conventional community detection methods that work in simple networks to detect the theme communities of one pattern at a time. However, since 
Since the community detection methods that work in simple networks only detect the theme communities of one pattern at a time, we have to perform community detection for each of the exponential number of patterns, which is computationally intractable.
}

\nop{\todo{The motivation is unclear.  Can you come up with a concrete application scenario where a data warehouse of theme communities is useful?} }

Last but not least, theme community finding can be a fast query answering service -- different users can easily use the service to efficiently explore a DBN and quickly retrieve theme communities of their own interest in real time. 
Providing this service requires enumerating and indexing all theme communities in a DBN, which is challenging because a large DBN usually contains a huge number of arbitrarily overlapping theme communities, and even counting the number of theme communities is \#P-hard.

\nop{
Efficiently enumerating and indexing all theme communities are essential for this service.
However, this is a challenging task because a large database network usually contains a huge number of arbitrarily overlapping theme communities, and even counting the number of theme communities is \#P-hard.
}

\nop{
Last but not least, efficiently enumerating and indexing all theme communities are essential for fast query answering. However, this is a challenging task because a large database network usually contains a huge number of arbitrarily overlapping theme communities, and even counting the number of theme communities is \#P-hard.
}
\nop{
Efficiently enumerating all overlapping theme communities and indexing them for fast query answering are challenging. 
}

\nop{
However, a user usually does not have the exact theme or community in mind when he starts to find theme communities of his interest. 

However, 

Therefore, e}

\nop{
Simple networks that consist of only vertices and edges are widely adopted to represent various real life networks.

Large networks are one of the most important form of data in the era of big data.

Simple networks that consist of only vertices and edges are widely used to represent various real life networks, such as e-commerce networks, social networks and collaboration networks.
In simple networks, a community naturally exists as a set of vertices that have strong and dense edge connections between each other~\cite{huang2017attribute}.
Community detection in simple networks is a fundamental problem that has been extensively studied for decades.

}

In this paper, we tackle the problem of finding theme communities from DBNs and make the following contributions.

First, we introduce the notion of DBN to model real networks in a natural and expressive manner. A DBN contains rich information about item co-occurrencess, pattern frequencies and graph/subgraph structures.

\nop{This presents novel opportunities to find meaningful theme communities.}

Second, we motivate the novel problem of finding theme communities from DBNs and prove that even counting the number of theme communities in a DBN is \#P-hard.

Third, we propose to find theme communities by enumerating \emph{maximal $(\mathbf{p},\alpha)$-trusses} in DBNs.
A $(\mathbf{p},\alpha)$-truss is a subgraph of a DBN such that the vertex databases of all vertices in the subgraph contain pattern $\mathbf{p}$, and the cohesion of every edge in the subgraph passes a non-negative threshold $\alpha$. Here, the cohesion of an edge incorporates the rich information about item co-occurrencess and pattern frequencies comprehensively, and is closely related to the structure of the $(\mathbf{p},\alpha)$-truss.

We propose a greedy algorithm and two effective pruning methods to enumerate maximal $(\mathbf{p},\alpha)$-trusses.
In our experiments, the pruning methods reduce the time cost of the greedy algorithm by over two orders of magnitude without any sacrifice in detection accuracy. \nop{Compared with what?}

\nop{
We investigate and apply several useful properties of maximal $(\mathbf{p},\alpha)$-truss to design two effective detection methods.
 In experiments, the detection methods reduce the time cost by more than two orders of magnitudes without losing the detection accuracy of theme communities.
}
Fourth, we advocate the construction of a data warehouse of maximal $(\mathbf{p},\alpha)$-trusses. 
To facilitate indexing and query answering in the data warehouse, we show that a maximal $(\mathbf{p},\alpha)$-truss can be efficiently decomposed and stored in a linked list. 
We use the decomposition to design an efficient indexing tree, and develop a query answering method that takes less than 1 second to retrieve a ranked list of theme communities from the indexing tree storing hundreds of millions of theme communities. Moreover, to ensure usability and avoid user disappointment due to queries that retrieve no valid theme community, we develop a query recommendation method that efficiently explores the indexing tree to recommend to the user a ranked list of new queries, which are highly similar to the user query, and can retrieve meaningful theme communities with large cohesiveness.
\nop{and can also recommend to a user a ranked list of meaningful new queries. \todo{What are those new queries about?  These are not mentioned before.}}

\nop{
More often than not, the user-provided query pattern $\mathbf{q}$ could be inappropriate, that is, the cohesiveness and size of the retrieved theme communities may be small or the answer to the query could be $\mathcal{R}_\mathbf{q}=\emptyset$.
In this case, we explore the TC-Tree to recommend the user a rank list of new query patterns, denoted by $\mathcal{U}=\mathbf{q}_1, \ldots, \mathbf{q}_k$. 
}

Last, we report extensive experimental results that demonstrate the accuracy and efficiency of the proposed methods in the enumeration, indexing and query answering of theme communities.
A case study shows that, by exploring the indexing tree that stores hundreds of millions of theme communities in a large co-author network, 
the query recommendation method quickly discovers meaningful research topics from an exponential number of combinations of user-interested keywords, and the query answering method efficiently retrieves the communities of closely collaborating scholars who frequently publish papers in those research topics discovered.

\nop{
 finding theme communities discovers meaningful hot research topics, as well as the corresponding groups of closely collaborating scholars who frequently publish papers in the same research topics.
}

\nop{
all scholars in the group closely collaborate with each other and every scholar in the group frequently publish papers in the same topic.}
\nop{
who share the same research interest. \todo{This claim is a little bit weak.  If one is only interested in finding common research interest, she does not have to find theme communities.}
}

The rest of the paper is organized as follows. We review related works in Section~\ref{Sec:rw} and formulate the theme community finding problem in Section~\ref{sec:prob}. We present a baseline method and a maximal $(\mathbf{p},\alpha)$-truss detection method in Section~\ref{sec:baseline}. We develop our major theme community finding algorithms in Section~\ref{sec:algo}, and the indexing and querying answering algorithms in Section~\ref{sec:index}. We report a systematic empirical study in Section~\ref{sec:exp} and conclude the paper in Section~\ref{sec:con}.

\nop{
Third, we formulate the task of finding theme communities in database network as finding maximal $(\mathbf{p},\alpha)$-trusses.
A \emph{maximal $(\mathbf{p},\alpha)$-truss} is a cohesive subgraph in a \emph{theme network}, which is induced from a database network by retaining all vertices containing the same pattern. 

Since each pattern corresponds to a unique theme network, $(\mathbf{p},\alpha)$-truss naturally allows arbitrary overlaps between theme communities with different themes.

Second, $(\mathbf{p},\alpha)$-truss smoothly incorporates item co-occurrence, pattern frequency and graph cohesiveness by the proposed edge \emph{cohesion}. This enables us to simultaneously discover both theme and community without knowing any of them in advance.
}

\section{Related Works}
\label{Sec:rw}
To the best of our knowledge, systematically finding theme communities from DBNs is novel and has not been formulated or tackled in literature.
Broadly, it is related to frequent pattern mining, truss detection and vertex attributed network clustering.


Frequent pattern mining is to find frequent patterns from a transaction database.  Some typical methods include
Apriori~\cite{agrawal1994fast} and 
FP-Growth~\cite{han2000mining}. 
Since frequent pattern mining methods do not handle densely connected graph structures, they cannot find communities in DBNs.

\nop{
\subsection{Frequent Pattern Mining}
Frequent pattern mining aims to find frequent patterns from a transactional database.
To list a few well known methods, Agrawal~\emph{et~al.}~\cite{agrawal1994fast} proposed Apriori to find frequent patterns by candidate generation;
Han~\emph{et~al.}~\cite{han2000mining} proposed the FP-Growth algorithm that avoids candidate generation by FP-Tree;
Yan~\emph{et~al.}~\cite{yan2002gspan} proposed gSpan to find the subgraphs that are frequently contained by a graph database.
The graph database is a set of graphs, which is substantially different from the DBN.

Most frequent pattern mining methods do not take a network as input. Therefore, they cannot find theme communities from a network of vertex databases.
}

Truss detection aims to detect $k$-trusses from unweighted simple networks.
Cohen~\cite{cohen2008trusses} defined $k$-truss as a subgraph $S$ where each edge is contained in at least $k-2$ triangles.
As demonstrated by many studies~\cite{cohen2008barycentric,cohen2008trusses,cohen2009graph}, $k$-truss naturally models cohesive communities in social networks and is elegantly related to some other graph structures, such as $k$-core~\cite{seidman1983network} and $k$-clique~\cite{luce1950connectivity}.

\nop{ and $k$-plex~\cite{seidman1978graph}}

The elegance of $k$-truss attracts much research attention. 
Wang~\emph{et~al.}~\cite{wang2012truss} proposed two memory efficient methods to find $k$-trusses for all possible values of $k$ in unweighted simple networks. 
Huang~\emph{et~al.}~\cite{huang2016truss} proposed $(k,\gamma)$-truss to extend the concept of $k$-truss from deterministic networks to probabilistic networks.
Huang~\emph{et~al.}~\cite{huang2014querying} designed an online community search method to query $k$-trusses by vertices. 

The above methods do not consider attributes of vertices, thus the detected communities may be hard to interpret due to the heterogeneity of vertex attributes~\cite{huang2017attribute}.
One may also wonder whether we can enumerate theme communities by finding $k$-trusses in each of the simple networks induced by a pattern. However, this is impractical because a DBN may contain an exponential number of patterns.

\newcommand{\parawidth}{56mm}
\begin{figure*}[t]
\centering
\subfigure[Database network]{\includegraphics[width=\parawidth]{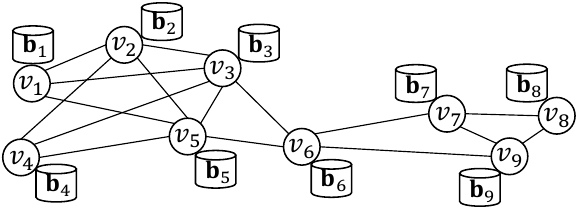}}
\subfigure[Theme network $G_{\mathbf{p}}$]{\includegraphics[width=\parawidth]{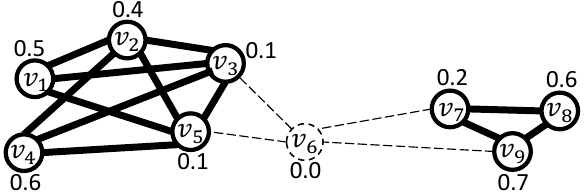}}
\subfigure[Theme network $G_{\mathbf{q}}$]{\includegraphics[width=\parawidth]{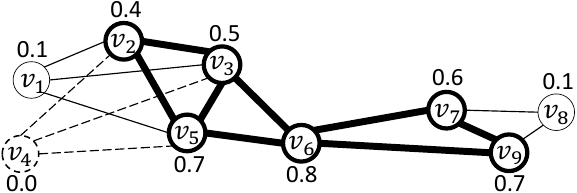}}

\vspace{-2mm}
\caption{A toy example of DBN, theme network and theme community. The pattern frequencies are labeled beside each vertex. The theme communities marked bold in (b) are valid when $\alpha\in [0, 0.2)$. The theme community marked in bold in (c) is valid when $\alpha\in[0.2, 0.4)$.}
\label{Fig:toy_example}
\end{figure*}

\nop{
A vertex attributed network is a special database network such that each vertex database contains only one transaction.
Many methods were proposed to find communities such that all vertices in the same community contain the same pattern.
}

Vertex attributed network clustering methods aim to find communities such that all vertices in the same community contain the same pattern and are densely connected.
ABACUS~\cite{berlingerio2013abacus} finds multi-dimensional communities by mining frequent itemsets.
CoPaM~\cite{moser2009mining} uses pruning methods to find maximal cohesive communities.
Prado~\emph{et~al.}~\cite{prado2013mining} designed interestingness measures to find cohesive communities.
Moosavi~\emph{et~al.}~\cite{moosavi2016community} used frequent pattern mining to find cohesive groups of users sharing similar features.
Huang~\emph{et~al.}~\cite{huang2017attribute} formulated the ATC problem, which finds a community containing a query vertex.
The ATC problem is NP-hard~\cite{huang2017attribute} and is substantially different from our problem, because our goal is to enumerate all theme communities, and even counting the number of theme communities in a DBN is \#P-hard.
There are also effective methods \nop{\todo{For what?} based on }that detect communities in vertex attributed networks by
graph weighting~\cite{steinhaeuser2008community, cruz2011semantic}, 
structural embedding~\cite{combe2012combining, dang2012community}, 
statistical inference~\cite{balasubramanyan2011block,yang2013community} and
subspace clustering~\cite{gunnemann2011db,wang2016semantic}.

\nop{Comparing with a vertex attributed network, the major strength of a database network is the pattern frequency that accurately measures the significance of a pattern.}

Since the above methods cannot distinguish different item co-occurrences and the corresponding pattern frequencies in DBNs, they cannot be directly applied to enumerate theme communities in DBNs.
One may also wonder whether we can transform a DBN into a vertex attributed network by treating every transaction as a set-valued item or by taking the union of all transactions per vertex database. Unfortunately, these transformations are ineffective because they lose the valuable information about item co-occurrencess and pattern frequencies.

\nop{
Besides, we cannot simply transform a vertex database into a single set of items by taking the union of all transactions or regarding each transaction as a set-valued item, because these waste the valuable information about item co-occurrences and pattern frequencies.
}

\nop{
since pattern frequency is not well defined in the set of items associated with each vertex in a vertex attributed network, 

Comparing with a VAN, the major strength of a database network is the rich and useful information about pattern frequencies.
Since the VAN clustering methods ignore pattern frequency, 
Since the database network contains an exponential number of patterns and we don't know which pattern forms a theme community, enumerating all theme communities in the database network requires to perform community detection for each of the exponential number of patterns, which is computationally intractable.
}

\section{Problem Definition}
\label{sec:prob}

In this section, we first introduce the notions of DBN, theme network and theme community, and then formalize the theme community finding problem.

\subsection{Database Network and Theme Network}
Let $S=\{s_1, \ldots, s_m\}$ be a set of items. An \emph{itemset} $\mathbf{x}$ is a subset of $S$.
A \emph{transaction} $\mathbf{t}$ is an itemset. Transaction $\mathbf{t}$ is said to \emph{contain} itemset $\mathbf{x}$ if $\mathbf{x}\subseteq \mathbf{t}$. The \emph{length} of transaction $\mathbf{t}$, denoted by $|\mathbf{t}|$, is the number of items in $\mathbf{t}$.
A \emph{vertex database} $\mathbf{b}=\{\mathbf{t}_1, \ldots, \mathbf{t}_h\}$ ($h\geq 1$) is a multi-set of transactions, that is, an itemset may appear multiple times as transactions in a vertex database.

A \emph{database network} (DBN) is an undirected graph denoted by $G=(V,E,B,S)$, where each vertex is associated with a vertex database.
Specifically, $V=\{v_1, \ldots, v_n\}$ is a set of vertices; $E=\{e_{ij}=(v_i, v_j) \mid v_i,v_j \in V\}$ is a set of edges; $B=\{\mathbf{b}_1, \ldots, \mathbf{b}_n\}$ is a set of vertex databases, where $\mathbf{b}_i$ is the vertex database associated with vertex $v_i$; and $S=\{s_1, \ldots, s_m\}$ is the set of items that constitute all vertex databases in $B$. That is, $S=\cup_{\mathbf{b}_i\in B} \cup_{\mathbf{t}\in\mathbf{b}_i} \mathbf{t}$.

Figure~\ref{Fig:toy_example}(a) gives a toy DBN, where the details of vertex databases are omitted due to the limit of space.


A \emph{pattern} is an itemset $\mathbf{p}\subseteq S$~\cite{agrawal1994fast,han2000mining}.
The \emph{length} of $\mathbf{p}$, denoted by $|\mathbf{p}|$, is the number of items in $\mathbf{p}$.
The \emph{absolute frequency} of $\mathbf{p}$ in vertex database $\mathbf{b}_i$ is the number of transactions in $\mathbf{b}_i$ that contain $\mathbf{p}$. 
The \emph{relative frequency} of $\mathbf{p}$ in $\mathbf{b}_i$ is the proportion of such transactions in $\mathbf{b}_i$~\cite{agrawal1994fast,han2000mining}.
Our method works well for both absolute frequency and relative frequency. 
For the sake of clarity, we use relative frequency by default in the rest of the paper, and write the relative frequency of $\mathbf{p}$ in $\mathbf{b}_i$ as $f_i(\mathbf{p})$.

\nop{
Since our method works with both absolute frequency and relative frequency, we indistinguishably denote them by $f_i(\mathbf{p})$ and call $f_i(\mathbf{p})$ the frequency of $\mathbf{p}$ on vertex $v_i$. 
For clearance, 
}

\nop{$f_i(\mathbf{p})$ is also called the frequency of $\mathbf{p}$ on vertex $v_i$.}

Given a pattern $\mathbf{p}$, the \emph{theme network} $G_\mathbf{p}$ is a subgraph induced from $G$ by the set of vertices satisfying $f_i(\mathbf{p}) > 0$, denoted by $G_\mathbf{p}=(V_\mathbf{p}, E_\mathbf{p})$, where 
$V_\mathbf{p}=\{v_i \in V  \mid f_i(\mathbf{p}) > 0\}$ is the set of vertices and 
$E_\mathbf{p}=\{e_{ij}\in E \mid v_i, v_j \in V_\mathbf{p}\}$ is the set of edges.
A subgraph $C_\mathbf{p}$ of $G_\mathbf{p}$ is written as $C_\mathbf{p}\subseteq G_\mathbf{p}$.

\nop{
A \emph{pattern} $p$ is an itemset that is also regarded as the \emph{theme} of the theme network. .
The \emph{frequency} of pattern $p$ on vertex $v_i$ is defined as the proportion of transactions in $d_i$ that contains $p$~\cite{agrawal1994fast,han2000mining}. .
}

\nop{The theme network $G_p$ is induced from database network $G$ by the set of vertices satisfying $f_i(p) > 0$.}

\nop{
\item $F_p=\{f_i(p) \mid v_i \in V_p\}$ is the set of vertex weights, where $f_i(p)$ is the weight of $v_i$.
A higher weight $f_i(p)$ indicates that pattern $p$ is more frequent on $v_i$.
A higher weight  indicates that pattern $p$ is more significant on $v_i$, thus $v_i$ gets a larger weight in the theme network $G_p$.
Later, when applying standard graph operations (e.g., union, intersection, etc) on $G_p$, we treat $G_p$ as a standard unweighted graph and ignore $F_p$.}

Figures~\ref{Fig:toy_example}(b)-(c) show two theme networks induced by different patterns $\mathbf{p}$ and $\mathbf{q}$, respectively. The edges and vertices in dashed lines are not contained in the theme networks.
\nop{\mcout{What are $\mathbf{p}_1$ and $\mathbf{p}_2$, respectively, in this example?}}
\nop{
since $f_6(p_1)=0$, $v_6$ and the edges connected to $v_6$ are removed when inducing $G_{p_1}$. The rest of the vertices and edges consist the theme network induced by $p_1$. For the theme network induced by $p_2$ (see Figure~\ref{Fig:toy_example}(c)), we remove $v_4$ and its connected edges, since $f_4(p_2)=0$.
}

We can induce a theme network by each pattern $\mathbf{p} \subseteq S$. A DBN $G$ can induce at most $2^{|S|}$ theme networks, where $G$ itself is the theme network of $\mathbf{p}=\emptyset$.

\nop{
\newcommand{\NotationTWidth}{65mm}
\begin{table}[t]
\centering\small
\caption{Frequently used notations.}
\label{Table:notations}
\begin{tabular}{|c|p{\NotationTWidth}|}
\hline
\makecell[c]{Notation}       &      \makecell[c]{Description} \\ \hline

$G$
&   The database network. \\ \hline

$G_\mathbf{p}$
&   The theme network induced by pattern $\mathbf{p}$. \\ \hline

$S$
&   The complete set of items in $G$. \\ \hline

$|\cdot|$
&   The set volume operator. \\ \hline

$C^*_{\mathbf{p},\alpha}$
&   The maximal $(\mathbf{p},\alpha)$-truss in theme network $G_\mathbf{p}$ with respect to threshold $\alpha$. \\ \hline

$\mathbb{T}_{\mathbf{p}, \alpha}$
&   The set of theme communities in $C^*_{\mathbf{p},\alpha}$. \\ \hline

$\mathcal{L}_\mathbf{p}$
&   The linked list storing the decomposed results of maximal $(\mathbf{p},\alpha)$-truss $C^*_{\mathbf{p},\alpha=0}$. \\ \hline

\nop{
$f_i(\mathbf{p})$
&   The frequency of pattern $\mathbf{p}$ on vertex $v_i$. \\ \hline
}

\nop{
$eco_{ij}(C_\mathbf{p})$
&   See Definition~\ref{Def:edge_cohesion}. \\ \hline
}

$\epsilon$
&   The threshold of pattern frequency for TCS. \\ \hline

\nop{
$s_{n_i}$
&   The item stored in node $n_i$ in the TC-Tree. \\ \hline
}

\nop{
$\mathcal{L}_{\mathbf{p}_i}$
&   The linked list stored in node $n_i$ in the TC-Tree. \\ \hline
}

\end{tabular}
\end{table}
}

\subsection{Theme Community}
\label{Sec:theme_community}

\nop{\todo{A reader has to have quite good patience to read and understand this subsection to appreciate the solid notion of theme community.  Is it possible that we can deliver the ideas and advantages quickly at the beginning of this subsection?}}

A \emph{theme community} is a subgraph of a theme network such that the vertices form a cohesively connected subgraph. 
We define a theme community by extending the well-defined $k$-truss~\cite{cohen2008trusses} from simple networks to DBN, \mc{such that it naturally models communities that are highly cohesive in both graph structure and pattern frequency}, and is elegantly related to some well-established graph structures such as $k$-core~\cite{seidman1983network} and $k$-clique~\cite{luce1950connectivity}. Moreover, theme communities in different theme networks may arbitrarily overlap with each other, which reflects the application scenarios where a vertex may participate in communities of different themes.

\nop{In this subsection, we first revisit the definition of $k$-truss~\cite{cohen2008trusses} in simple networks, then we extend the basic concepts of $k$-truss to define edge cohesion, $(\mathbf{p}, \alpha)$-truss, maximal $(\mathbf{p}, \alpha)$-truss and theme community in DBNs.
}

\nop{
In this subsection, we define theme community by smoothly extending the well-defined $k$-truss~\cite{cohen2008trusses} in simple networks to database networks.
}

The intuition of $k$-truss is that, if every edge of a community is contained in more triangles, then the vertices in the community are more densely connected~\cite{cohen2008trusses}.
Here, a \emph{triangle}, denoted by $\triangle_{ijk}= \{v_i, v_j, v_k\}$, is a clique containing vertices $v_i$, $v_j$ and $v_k$. 
If a triangle $\triangle_{ijk}$ is a subgraph of a graph $T$, we say $\triangle_{ijk}$ is in $T$.

\nop{
A triangle $\triangle_{ijk}$ is said to be in a subgraph $T$ if $\triangle_{ijk}$ is a subgraph of $T$.
}

In simple networks, Cohen~\cite{cohen2008trusses} first measured the cohesion of an edge $e_{ij}$ in a subgraph $T$ by the number of triangles in $T$ that contain $e_{ij}$; then he defines $k$-truss as a subgraph $T$ such that the cohesion of every edge in $T$ is at least $k-2$. 
Essentially, every triangle $\triangle_{ijk}$ in a subgraph $T$ contributes a weight of $1.0$ to the cohesion of each edge in $\triangle_{ijk}$, and the cohesion of an edge $e_{ij}$ in $T$ is the sum of the weights contributed by all triangles in $T$ that contain $e_{ij}$.

Next, we illustrate how to integrate the rich information about item co-occurrencess and pattern frequencies with the edge cohesion in DBNs.

In a DBN, to incorporate item co-occurrences and the corresponding pattern frequencies, the contribution of a triangle $\triangle_{ijk}$ to the cohesion of each edge in $\triangle_{ijk}$ should be relevant to two factors: 1) a pattern $\mathbf{p}\subseteq S$, that is, a set of co-occurring items; and 2) the frequencies of $\mathbf{p}$ in the vertex databases $\mathbf{b}_i$, $\mathbf{b}_j$ and $\mathbf{b}_k$.

\nop{
the frequencies of a pattern $\mathbf{p}$ are different in different vertex databases. 
Therefore, for any pattern $\mathbf{p}$, the contribution of triangle $\triangle_{ijk}$ in a subgraph $T$ to the cohesion of edge $e_{ij}$ in $T$ should be relevant to the frequencies of $\mathbf{p}$ in the vertex databases $\mathbf{b}_i$, $\mathbf{b}_j$ and $\mathbf{b}_k$.
}

Intuitively, if $\mathbf{p}$ has higher frequencies in $\mathbf{b}_i$, $\mathbf{b}_j$ and $\mathbf{b}_k$, then all vertices in the triangle $\triangle_{ijk}$ are dominated by the pattern $\mathbf{p}$, thus $\triangle_{ijk}$ should contribute a heavier weight to the cohesion of each edge in $\triangle_{ijk}$.
If the frequency of $\mathbf{p}$ is zero in any of $\mathbf{b}_i$, $\mathbf{b}_j$ and $\mathbf{b}_k$, then at least one vertex in $\triangle_{ijk}$ is not associated with $\mathbf{p}$, therefore $\triangle_{ijk}$ should not contribute any weight to the cohesion of any edge of $\triangle_{ijk}$.

Following the above intuition, for a pattern $\mathbf{p}$, we define the contribution of a triangle $\triangle_{ijk}$ to the cohesion of each edge in $\triangle_{ijk}$ as
\begin{equation}\nonumber
\label{Eqn:contribute}
\min{(f_i(\mathbf{p}), f_j(\mathbf{p}), f_k(\mathbf{p}))}.
\end{equation}

\nop{the minimum of $f_i(\mathbf{p}), f_j(\mathbf{p})$ and $f_k(\mathbf{p})$, that is,}

For any pattern $\mathbf{p}$, a triangle $\triangle_{ijk}$ contributes a none-zero weight to the cohesion of its edges if and only if $\triangle_{ijk}$ is in the theme network $G_\mathbf{p}$.
The reason is that if $\triangle_{ijk}$ is not in $G_\mathbf{p}$, then the frequency of $\mathbf{p}$ is zero in at least one of the vertex databases $\mathbf{b}_i$, $\mathbf{b}_j$ and $\mathbf{b}_k$.
In this case, $\triangle_{ijk}$ does not contribute any weight to the cohesion of its edges.
Accordingly, we define the \emph{edge cohesion} in DBNs as follows.

\begin{definition}[Edge Cohesion]
\label{Def:edge_cohesion}
Consider a pattern $\mathbf{p}$ and the theme network $G_\mathbf{p}$, for a subgraph $C_\mathbf{p}\subseteq G_\mathbf{p}$ and an edge $e_{ij}$ in $C_\mathbf{p}$, the \textbf{edge cohesion} of $e_{ij}$ in $C_\mathbf{p}$ is 
\begin{equation}
\label{Eqn:cohesion}
eco_{ij}(C_\mathbf{p})=\sum\limits_{\triangle_{ijk}\subseteq C_\mathbf{p}} \min{(f_i(\mathbf{p}), f_j(\mathbf{p}), f_k(\mathbf{p}))} \nonumber
\end{equation}
\end{definition}

\begin{example}
In Figure~\ref{Fig:toy_example}(b), for subgraph $C_{\mathbf{p}}$ induced by the set of vertices $\{v_1, v_2, v_3, v_4, v_5\}$, edge $e_{12}$ is contained in $\triangle_{123}$ and $\triangle_{125}$, thus the edge cohesion of $e_{12}$ is $eco_{12}(C_{\mathbf{p}})=\min(f_1(\mathbf{p}),$ $ f_2(\mathbf{p}), f_3(\mathbf{p})) + \min(f_1(\mathbf{p}), f_2(\mathbf{p}), $ $f_5(\mathbf{p}))=0.2$.
\end{example}

The edge cohesion in simple networks~\cite{cohen2008trusses} is a special case of the edge cohesion in Definition~\ref{Eqn:cohesion}.
Because, when $f_i(\mathbf{p}) = 1$ for the vertex database of every vertex $v_i$ in $C_\mathbf{p}$, $eco_{ij}(C_\mathbf{p})$ is exactly the number of triangles in $C_\mathbf{p}$ that contain $e_{ij}$.

\nop{
If $f_i(\mathbf{p}) = 1$ for every vertex $v_i$ in the subgraph $C_\mathbf{p}$, then $eco_{ij}(C_\mathbf{p})$ is exactly the number of triangles in $C_\mathbf{p}$ that contain $e_{ij}$.
}

\nop{
Apparently, a large edge cohesion $eco_{ij}(C_\mathbf{p})$ indicates that vertices $v_i, v_j$ either have high frequencies of $\mathbf{p}$ or have a large number of common neighbouring vertices whose vertex databases contain theme $\mathbf{p}$.
According to the criteria of good theme community, if every edge in $C_\mathbf{p}$ has a large edge cohesion, $C_\mathbf{p}$ will be a good theme community. Based on this insight, 
}

\nop{
Now, we propose \emph{$(\mathbf{p},\alpha)$-truss}, a subgraph such that the cohesion of every edge in the subgraph is larger than a threshold.
}

Based on the edge cohesion in DBNs, we smoothly extend the notion of $k$-truss~\cite{cohen2008trusses} in simple networks to the notion of \emph{$(\mathbf{p},\alpha)$-truss} in DBNs as follows.

\nop{
Based on the above edge cohesion in database networks, we can smoothly extend the $k$-truss in simple networks, and define the \emph{$(\mathbf{p},\alpha)$-truss} in database networks as a subgraph $T\subseteq G_\mathbf{p}$ such that the cohesion of every edge in $T$ is larger than a non-negative threshold $\alpha$.
}

\nop{
The formal definition of $(\mathbf{p},\alpha)$-truss is as follows.
}

\begin{definition}[$(\mathbf{p},\alpha)$-Truss]
\label{Def:pattern_truss}
Given a pattern $\mathbf{p}$ and a minimum cohesion threshold $\alpha \geq 0$, a \textbf{$(\mathbf{p},\alpha)$-truss}, denoted by $C_{\mathbf{p},\alpha}$, is a subgraph of the theme network $G_\mathbf{p}$ such that the edge cohesion $eco_{ij}(C_{\mathbf{p},\alpha})$ of every edge in $C_{\mathbf{p},\alpha}$ is \textbf{larger than} $\alpha$.
\end{definition}

\mc{
A $(\mathbf{p},\alpha)$-truss $C_{\mathbf{p},\alpha}$ is elegantly related to some well-established graph structures, such as $k$-truss~\cite{cohen2008trusses}, $k$-core~\cite{seidman1983network} and $k$-clique~\cite{luce1950connectivity}, when $f_i(\mathbf{p}) = 1$ for the vertex database of every vertex $v_i$ in $C_{\mathbf{p},\alpha}$.
First, if $\alpha = k-3$, $C_{\mathbf{p},\alpha}$ becomes a $k$-truss.
Second, if $\alpha = k-2$ and $C_{\mathbf{p},\alpha}$ is a maximal connected subgraph in $G_\mathbf{p}$, it is a $k$-core.
Last, if $\alpha=k-3$ and $C_{\mathbf{p},\alpha}$ contains $k$ vertices, it is a $k$-clique.
}

A $(\mathbf{p},\alpha)$-truss is not necessarily a connected subgraph, and the union of multiple $(\mathbf{p},\alpha)$-trusses is still a $(\mathbf{p},\alpha)$-truss.
For example, consider the edges in bold in Figure~\ref{Fig:toy_example}(b).  When $\alpha \in [0, 0.2)$, the two subgraphs induced by $\{v_1, v_2, v_3, v_4, v_5\}$ and $\{v_7, v_8, v_9\}$ are both $(\mathbf{p},\alpha)$-trusses. The union of the two subgraphs is still a $(\mathbf{p},\alpha)$-truss, but it is not a connected subgraph.

\nop{
not a connected subgraph, however, it is still a $(\mathbf{p},\alpha)$-truss.
}
\nop{
is also a $(\mathbf{p},\alpha)$-truss that is not a connected subgraph.
}
\nop{
It is easy to see that, for a given $\alpha$, the union of multiple $(\mathbf{p},\alpha)$-trusses is still a $(\mathbf{p},\alpha)$-truss.
}
\nop{Take Figure~\ref{Fig:toy_example}(c) as an example, when $\alpha \in [0.2, 0.4)$, $\{v_2, v_3, v_5, v_6\}$, $\{v_3, v_5, v_6, v_7, v_9\}$ and the union of them (i.e., $\{v_2, v_3, v_5, v_6, v_7, v_9\}$) are all valid $(\mathbf{p},\alpha)$-trusses.}

\begin{definition}[Maximal $(\mathbf{p},\alpha)$-Truss]
\label{Def:maximal_pattern_truss}
A \textbf{maximal $(\mathbf{p},\alpha)$-truss}, denoted by $C^*_{\mathbf{p},\alpha}$, is a $(\mathbf{p},\alpha)$-truss in $G_\mathbf{p}$ such that any proper superset of $C^*_{\mathbf{p},\alpha}$ is not a $(\mathbf{p},\alpha)$-truss in $G_\mathbf{p}$.
\end{definition}

Since the union of multiple $(\mathbf{p},\alpha)$-trusses is still a $(\mathbf{p},\alpha)$-truss, a maximal $(\mathbf{p},\alpha)$-truss is the union of all $(\mathbf{p},\alpha)$-trusses in $G_{\mathbf{p}}$.  Apparently, a maximal $(\mathbf{p},\alpha)$-truss is still not necessarily a connected subgraph.

\nop{
We denote by $C^*_{\mathbf{p},\alpha}=(V^*_{\mathbf{p},\alpha}, E^*_{\mathbf{p},\alpha})$ the maximal $(\mathbf{p},\alpha)$-truss in $G_\mathbf{p}$. 
}

Now we are ready to define theme community.
\begin{definition}[Theme Community]
\label{Def:theme_community}
Every maximal connected subgraph in a maximal $(\mathbf{p},\alpha)$-truss is a \textbf{theme community}.
\end{definition}

\nop{
A \emph{maximal $(\mathbf{p},\alpha)$-truss} is a $(\mathbf{p},\alpha)$-truss that is not a subgraph of any other $(\mathbf{p},\alpha)$-truss. We denote maximal $(\mathbf{p},\alpha)$-truss by $C^*_p(\alpha)=\{V^*_p(\alpha), E^*_p(\alpha)\}$.

We list the following useful facts about $(\mathbf{p},\alpha)$-truss: 
\begin{itemize}
\item Pattern truss degenerates to $k$-truss~\cite{cohen2008trusses} if $\forall v_i \in V_p, f_i(p) = 1$ and $\alpha = k-2$. 
\item A $(\mathbf{p},\alpha)$-truss is not necessarily a connected subgraph. For example, in Figure~\ref{Fig:toy_example}(b), $\{v_1, v_2, v_3, v_4,$ $ v_5, v_7, v_8, v_9\}$ is a valid $(\mathbf{p},\alpha)$-truss when $\alpha \in [0, 0.2)$.
\item The union of two $(\mathbf{p},\alpha)$-trusses is also a $(\mathbf{p},\alpha)$-truss. This is because adding new vertices and edges to an existing $(\mathbf{p},\alpha)$-truss does not decrease the cohesion of any original edge of the $(\mathbf{p},\alpha)$-truss. For example, in Figure~\ref{Fig:toy_example}(c), $\{v_2, v_3, v_5, v_6\}$, $\{v_3, v_5, v_6, v_7, v_9\}$ and the union of them (i.e., $\{v_2, v_3, v_5, v_6, v_7, v_9\}$) are all valid $(\mathbf{p},\alpha)$-trusses when $\alpha \in [0.2, 0.4)$.
\item For a fixed threshold $\alpha$, a theme network can only have one unique maximal $(\mathbf{p},\alpha)$-truss, which is the union of all valid $(\mathbf{p},\alpha)$-trusses in the theme network.
\end{itemize}
}

\begin{example}
In Figure~\ref{Fig:toy_example}(b), when $\alpha \in [0, 0.2)$, $\{v_1, v_2, v_3, v_4, v_5\}$ and $\{v_7, v_8, v_9\}$ are two theme communities in $G_{\mathbf{p}}$.
In Figure~\ref{Fig:toy_example}(c), when $\alpha \in [0.2, 0.4)$, $\{v_2, v_3, v_5, v_6, v_7, v_9\}$ is a theme community in $G_{\mathbf{q}}$, and partially overlaps with the two theme communities in $G_{\mathbf{p}}$.
\end{example}

\nop{
 shows another theme community in theme network $G_{\mathbf{p}_2}$ when $\alpha \in [0.2, 0.4)$. Such community is $\{v_2, v_3, v_5, v_6, v_7, v_9\}$, which partially overlaps with both the theme communities in $G_{\mathbf{p}_1}$.}

\nop{
\begin{table}[t]
\caption{\mc{Three types of good theme communities.}}
\centering
\label{Table:thcm_types}
\begin{tabular}{|c|c|c|}
\hline
Types  	& Theme Coherence & Edge Connection \\ \hline
Type-I       &  Medium                 & High  			\\ \hline
Type-II      &  High                      & Medium  		\\ \hline
Type-III     &  High                      & High 	 		\\ \hline
\end{tabular}
\end{table}
}

We write the set of all theme communities in the maximal $(\mathbf{p}, \alpha)$-truss as $\mathbb{T}_{\mathbf{p}, \alpha}=\{T^1_{\mathbf{p}, \alpha},\ldots, T^m_{\mathbf{p}, \alpha}\}$, where $T^i_{\mathbf{p}, \alpha}$ is the $i$-th theme community in $\mathbb{T}_{\mathbf{p}, \alpha}$ $(i\in\{1, \ldots, m\})$.
Every theme community $T^i_{\mathbf{p}, \alpha}\in\mathbb{T}_{\mathbf{p}, \alpha}$ is also a $(\mathbf{p},\alpha)$-truss. 
We define the cohesiveness of a theme community as follows.

\begin{definition}[Cohesiveness of Theme Community]
\label{Def:cohe_theme_community}
The \textbf{cohesiveness} of a theme community is the minimum cohesion of its edges.
\end{definition}

There are several important benefits from modeling theme communities using maximal $(\mathbf{p},\alpha)$-trusses. 
First, there exist polynomial time algorithms to find maximal $(\mathbf{p},\alpha)$-trusses.
Second, maximal $(\mathbf{p},\alpha)$-trusses of different theme networks may overlap with each other, which reflects the application scenarios where a vertex may participate in communities of different themes.
Last, as to be proved in Sections~\ref{Sec:pompt} and~\ref{Sec:mpt_dec}, maximal $(\mathbf{p},\alpha)$-trusses have many desirable properties that enable us to design efficient mining and indexing algorithms for theme community finding.

\nop{One interesting observation is that even though the left community has a lower average vertex frequency than the right community, it is equally valid as the right community under the same threshold $\alpha$. This is because the left community has a denser edge connection that contains more triangles, and the edge cohesion comprehensively considers both edge connections and vertex frequencies.}

\subsection{Problem Definition and Complexity}
\label{Sec:tcfp}
\nop{We define the theme community finding problem as follows.}
\begin{definition}[Theme Community Finding]
\label{Def:theme_comm_finding}
Given a DBN $G$ and a minimum cohesion threshold $\alpha$, the \textbf{problem of theme community finding} is to enumerate all theme communities in $G$.
\end{definition}

\nop{
This is a challenging problem, since a database network can induce up to $2^{|S|}$ theme networks and each theme network may contain many theme communities. 
As a result, enumerating all theme communities in a database network is computationally intractable.
}

\nop{
Next, we prove in Theorem~\ref{theo:hardness} that even counting the number of the theme communities in $G$ is \#P-hard. 
}

\nop{The proof of Theorem~\ref{theo:hardness} is in Section~\ref{Apd:hardness} of the appendix.}

\begin{theorem}[Complexity]
\label{theo:hardness}
Given a DBN $G$ and a minimum cohesion threshold $\alpha$, the problem of counting the number of theme communities in $G$ is \#P-hard.
\end{theorem}

\begin{proof}
We prove by a reduction from the \emph{Frequent Pattern Counting} (FPC) problem, which is \#P-complete~\cite{gunopulos2003discovering}.

Given a vertex database $\mathbf{b}$ and a minimum support threshold $\alpha\geq 0$, an instance of the FPC problem is to count the number of patterns $\mathbf{p}$ in $\mathbf{b}$ such that $f(\mathbf{p}) > \alpha$.
Here, $f(\mathbf{p})$ is the frequency of $\mathbf{p}$ in $\mathbf{b}$.

We construct a DBN $G=(V,E,B,S)$, where $V=\{v_1, v_2, v_3\}$. $E=\{(v_1, v_2), (v_2, v_3), (v_3, v_1)\}$ forms a triangle; $B=\{\mathbf{b}_1, \mathbf{b}_2, \mathbf{b}_3 \mid \mathbf{b}_1=\mathbf{b}_2=\mathbf{b}_3=\mathbf{b}\}$; and $S$ is the set of items appearing in $\mathbf{b}$. Apparently, $G$ can be constructed in $O(|\mathbf{b}|)$ time.

For any pattern $\mathbf{p}\subseteq S$, since $\mathbf{b}_1=\mathbf{b}_2=\mathbf{b}_3=\mathbf{b}$, it follows $f_1(\mathbf{p})=f_2(\mathbf{p})=f_3(\mathbf{p})=f(\mathbf{p})$. According to Definition~\ref{Def:edge_cohesion}, $eco_{12}(G_\mathbf{p})=eco_{13}(G_\mathbf{p})=eco_{23}(G_\mathbf{p})=f(\mathbf{p})$. 
By Definition~\ref{Def:theme_community}, $G_\mathbf{p}$ is a theme community in $G$ if and only if $f(\mathbf{p})> \alpha$.
Therefore, for any threshold $\alpha\geq 0$, the number of theme communities in $G$ is equal to the number of patterns in $\mathbf{b}$ satisfying $f(\mathbf{p})>\alpha$, which is exactly the answer to the FPC problem.
\end{proof}

The theme community finding problem is challenging because a DBN $G$ can induce up to $2^{|S|}-1$ theme networks, and each theme network may contain many theme communities.
Denote by $n$ the number of vertices in $G$. Since the smallest theme community is a triangle that is not connected to any other vertices, the maximum number of theme communities in a theme network is $\floor*{\frac{n}{3}}$ for a fixed threshold $\alpha\geq 0$.
Therefore, the maximum number of theme communities in $G$ is $\floor*{\frac{n}{3}} (2^{|S|}-1)$.
However, this theoretical worst case seldom occurs in real-world datasets, since a large proportion of patterns do not induce any theme community. This is because a longer pattern $\mathbf{p}$ is less likely to appear in a vertex database, and the sparse structure of $G$ makes it even harder for a small number of vertices containing $\mathbf{p}$ to form a theme community.
As demonstrated by our experiments in section~\ref{sec:exp}, the proposed methods efficiently enumerate and index theme communities in large real-world DBNs.

\nop{
longer patterns are less likely to be contained in a vertex database

1) longer patterns are less likely to be contained in a vertex database, thus . 2)  3) the structure of $G$ is sparse. 
}

\nop{
Denote by $L$ the length of the longest pattern $\mathbf{p}^{longest}$ that induces a theme community in $G$. Then, $\floor*{\frac{n}{3}} (2^L-1)$ is the upper bound of the number of theme communities in $G$. Since $L \leq |S|$, this bound is $\floor*{\frac{n}{3}} (2^{|S|}-1)$ in the worst case.
}

\nop{
$G$ contains at most $2^L - 1$ theme networks that contain  

. The maximum number of theme networks in $G$ is $2^L - 1$. Therefore,
}

\nop{
\mc{Given a database network $G$ and a threshold $\alpha$, if an algorithm can enumerate all theme communities in $G$, it is also able to count them. 
From this perspective, the theme community finding problem (Definition~\ref{Def:theme_comm_finding}) is at least as hard as the problem of counting the number of theme communities in $G$.}
}


Since extracting theme communities (i.e., maximal connected subgraphs) from a maximal $(\mathbf{p},\alpha)$-truss is straightforward, the core of the theme community finding problem is to identify the maximal $(\mathbf{p},\alpha)$-trusses of all theme networks. 
In the rest of the paper, we develop an exact algorithm to find maximal $(\mathbf{p},\alpha)$-truss and investigate various techniques to speed up the search.

\nop{
However, as proved in Section~\ref{Sec:pompt}, the size of maximal $(\mathbf{p},\alpha)$-truss decreases monotonously when the length of pattern increases. 
With carefully designed algorithms, we can safely prune a pattern $\mathbf{p}$ if all its sub-patterns $\{\mathbf{h} \mid \mathbf{h}\subset \mathbf{p}, \mathbf{h}\neq \emptyset\}$ cannot induce $(\mathbf{p},\alpha)$-truss.
Such prunning leads to high detection efficiency without accuracy loss.
}

\section{Baseline}
\label{sec:baseline}

In this section, we first introduce \emph{Maximal $(\mathbf{p},\alpha)$-Truss Detector} (MTD), which finds the maximal $(\mathbf{p},\alpha)$-truss of a given theme network $G_\mathbf{p}$. Then, we present a baseline for theme community finding.

\subsection{Maximal \titlepa{11}-Truss Detector}
\label{sec:mtdsub}

Given $G_\mathbf{p}$ and $\alpha$, an edge in $G_\mathbf{p}$ is called an \emph{unqualified edge} if the edge cohesion is not larger than $\alpha$.
The key idea of MTD is to iteratively remove all unqualified edges so that the remaining edges and connected vertices constitute the maximal $(\mathbf{p},\alpha)$-truss.
By Definition~\ref{Def:edge_cohesion}, the cohesion of an edge $e_{ij}$ depends on the structure of subgraph $C_\mathbf{p} \subseteq G_\mathbf{p}$ that contains $e_{ij}$, thus removing an edge from $C_\mathbf{p}$ reduces the cohesion of some remaining edges.
As a result, MTD iteratively updates the cohesion of all remaining edges after removing each unqualified edge. The iteration continues until there is no unqualified edge to remove.

As shown in Algorithm~\ref{Alg:mptd}, MTD consists of two phases. Phase 1 (Lines 1-8) computes the initial cohesion of each edge and pushes unqualified edges into queue $Q$. Phase 2 (Lines 9-18) iteratively removes the unqualified edges in $Q$ from $E_\mathbf{p}$. Since removing $e_{ij}$ also breaks $\triangle_{ijk}$, we update $eco_{ik}(G_\mathbf{p})$ and $eco_{jk}(G_\mathbf{p})$ in Lines 12-13. If consequently $e_{ik}$ or $e_{jk}$ become unqualified, they are pushed into $Q$ (Lines 14-15).
Finally, the surviving edges and connected vertices are returned as the maximal $(\mathbf{p},\alpha)$-truss.

\nop{
The correctness of MTD is based on the fact that MTD does not remove any edge from a valid maximal $(\mathbf{p},\alpha)$-truss $C^*_p(\alpha)$. 
This is because MTD only removes unqualified edges and all edges in $C^*_p(\alpha)$ are qualified edges according to the definition of $C^*_p(\alpha)$. 
}

We show the correctness of MTD as follows. 
If $C^*_{\mathbf{p},\alpha}= \emptyset$, then all edges in $E_\mathbf{p}$ are removed as unqualified edges and MTD returns $\emptyset$. 
If $C^*_{\mathbf{p},\alpha}\neq \emptyset$, then only the edges in $E_\mathbf{p}$ that are not contained in $C^*_{\mathbf{p},\alpha}$ are removed as unqualified edges, and MTD returns exactly $C^*_{\mathbf{p},\alpha}$.

\nop{
all edges in $E_\mathbf{p}\setminus E_\mathbf{p}^*(\alpha)$ are removed as unqualified edges and MTD returns exactly $C^*_{\mathbf{p},\alpha}$.
}

\nop{we decompose the theme network $G_p$ into a maximal $(\mathbf{p},\alpha)$-truss $C^*_p(\alpha)$ and the remaining subgraph denoted by $\overline{C^*_p(\alpha)}=G_p \setminus C^*_p(\alpha)$.
Since adding new edges or vertices to $C^*_p(\alpha)$ does not decrease the cohesion of the existing edges in $C^*_p(\alpha)$, the edges and vertices in $\overline{C^*_p(\alpha)}$ will not decrease the cohesion of any edge in $C^*_p(\alpha)$. Furthermore, $C^*_p(\alpha)$ being the maximal $(\mathbf{p},\alpha)$-truss equally means all edges in $\overline{C^*_p(\alpha)}$ are unqualified, thus $\overline{C^*_p(\alpha)}$ will be removed by MTD and the maximal $(\mathbf{p},\alpha)$-truss $C^*_p(\alpha)$ is found. Additionally, if theme network $G_p$ does not contain any maximal $(\mathbf{p},\alpha)$-truss, we have $C^*_p(\alpha)=\emptyset$ and MTD will remove all edges from $G_p$.}

The time complexity of Algorithm~\ref{Alg:mptd} is dominated by the complexity of triangle enumeration for each edge $e_{ij}$ in $E_\mathbf{p}$. 
This requires checking all neighboring vertices of $v_i$ and $v_j$, which costs $\mathcal{O}(d(v_i) + d(v_j))$ time, where $d(v_i)$ and $d(v_j)$ are the degrees of $v_i$ and $v_j$, respectively.
Since all edges in $E_\mathbf{p}$ are checked, the cost for Lines 1-8 in Algorithm~\ref{Alg:mptd} is $\mathcal{O}(\sum_{e_{ij}\in E_\mathbf{p}} (d(v_i) + d(v_j))) = \mathcal{O}(\sum_{v_i\in V_\mathbf{p}} d^2(v_i))$.
The cost of Lines 9-18 is also $\mathcal{O}(\sum_{v_i\in V_\mathbf{p}} d^2(v_i))$. The worst case happens when all edges are removed. Therefore, the time complexity of MTD is $\mathcal{O}(\sum_{v_i\in V_\mathbf{p}} d^2(v_i))$.  In many real networks, most vertices have very small degrees.  Thus, MTD can efficiently find the maximal $(\mathbf{p},\alpha)$-truss of a sparse theme network.

\subsection{Theme Community Scanner: A Baseline}
\label{sec:tcs}

\nop{
detect the maximal $(\mathbf{p},\alpha)$-truss on each theme network using MTD, and improve the detection efficiency by pre-filtering out the patterns whose maximum frequencies in all vertex databases cannot reach a minimum frequency threshold $\epsilon$. 
}

Since a DBN $G$ may induce up to $2^{|S|}$ theme networks, running MTD on all theme networks is impractical.
Intuitively, patterns with low frequencies may be less likely to induce a theme community. Thus, a simple idea is to first filter out the patterns whose maximum frequencies in all vertex databases fail a minimum frequency threshold $\epsilon$, then apply MTD to detect the maximal $(\mathbf{p},\alpha)$-truss on each theme network induced by the remaining patterns.

Following the above idea, we introduce a baseline method, called \emph{Theme Community Scanner} (TCS).
Given a frequency threshold $\epsilon$, TCS first obtains the set of candidate patterns $\mathcal{P}=\{\mathbf{p} \mid \exists v_i\in V, f_i(\mathbf{p}) \geq \epsilon\}$ by enumerating all patterns in each vertex database. Then, for each candidate pattern $\mathbf{p}\in \mathcal{P}$, we induce theme network $G_\mathbf{p}$ and find the maximal $(\mathbf{p},\alpha)$-truss by MTD. The final result is a set of maximal $(\mathbf{p},\alpha)$-trusses, denoted by $\mathbb{C}(\alpha)=\{C^*_{\mathbf{p},\alpha} \mid C^*_{\mathbf{p},\alpha}\neq\emptyset, \mathbf{p}\in \mathcal{P}\}$.

The filtering step of TCS improves the detection efficiency of theme communities, however, it may miss some theme communities, since a pattern $\mathbf{p}$ with relatively small frequencies on all vertex databases can still form a good theme community, if a large number of vertices containing $\mathbf{p}$ form a densely connected subgraph.
As a result, TCS trades accuracy for efficiency, which, however, is not as effective as one may expect according to our experiments in Section~\ref{Sec:eop}.

\begin{algorithm}[t]
\caption{Maximal $(\mathbf{p},\alpha)$-Truss Detector}
\label{Alg:mptd}
\KwIn{A theme network $G_\mathbf{p}$ and a user input $\alpha$.}
\KwOut{The maximal $(\mathbf{p},\alpha)$-truss $C^*_{\mathbf{p},\alpha}$ in $G_\mathbf{p}$.}
\BlankLine
\begin{algorithmic}[1]
    \STATE Initialize: $Q\leftarrow \emptyset$.
    \FOR{each $e_{ij}\in E_\mathbf{p}$}
    	\STATE $eco_{ij}(G_\mathbf{p})\leftarrow 0$.
    	\FOR{each $v_k\in \triangle_{ijk}$}
		\STATE $eco_{ij}(G_\mathbf{p})\leftarrow eco_{ij}(G_\mathbf{p}) + \min(f_i(\mathbf{p}), f_j(\mathbf{p}), f_k(\mathbf{p}))$.
	\ENDFOR
	\STATE \textbf{if} $eco_{ij}(G_\mathbf{p}) \leq \alpha$ \textbf{then} $Q.push(e_{ij})$.
    \ENDFOR
    
    \WHILE{$Q\neq\emptyset$}
    	\STATE $e_{ij}\leftarrow Q.pop()$.
	\FOR{each $v_k\in \triangle_{ijk}$}
		\STATE $eco_{ik}(G_\mathbf{p})\leftarrow eco_{ik}(G_\mathbf{p}){-}\min(f_i(\mathbf{p}), f_j(\mathbf{p}), f_k(\mathbf{p})).$
		\STATE $eco_{jk}(G_\mathbf{p})\leftarrow eco_{jk}(G_\mathbf{p}){-}\min(f_i(\mathbf{p}), f_j(\mathbf{p}), f_k(\mathbf{p})).$
		\STATE \textbf{if} $eco_{ik}(G_\mathbf{p})\leq \alpha$ \textbf{then} $Q.push(e_{ik})$.
		\STATE \textbf{if} $eco_{jk}(G_\mathbf{p})\leq \alpha$ \textbf{then} $Q.push(e_{jk})$.
	\ENDFOR
	\STATE Remove $e_{ij}$ from $G_{\mathbf{p}}$.
    \ENDWHILE
\BlankLine
\RETURN $C^*_{\mathbf{p},\alpha} = G_{\mathbf{p}}$.
\end{algorithmic}
\end{algorithm}

\section{Theme Community Finding}
\label{sec:algo}

In this section, we first explore several fundamental properties of maximal $(\mathbf{p},\alpha)$-truss, then apply them to develop two fast and exact theme community finding methods.

\nop{\emph{Theme Community Finder Apriori (TCFA)} and \emph{Theme Community Finder Intersection (TCFI)}.}

\subsection{Properties of Maximal \titlepa{11}-Truss}
\label{Sec:pompt}

\begin{theorem}[Graph Anti-monotonicity]
\label{Prop:gam}
If patterns $\mathbf{p}_1 \subseteq \mathbf{p}_2$, then maximal $(\mathbf{p}, \alpha)$-trusses $C^*_{\mathbf{p}_2,\alpha}\subseteq C^*_{\mathbf{p}_1,\alpha}$.
\end{theorem}

\begin{proof}
Since $C^*_{\mathbf{p}_1,\alpha}$ is the union of all $(\mathbf{p}_1,\alpha)$-trusses, we prove $C^*_{\mathbf{p}_2,\alpha}\subseteq C^*_{\mathbf{p}_1,\alpha}$ by proving $C^*_{\mathbf{p}_2,\alpha}$ is also a $(\mathbf{p}_1,\alpha)$-truss.

First, we prove that there exists a subgraph $H_{\mathbf{p}_1}$ in $G_{\mathbf{p}_1}$ such that $H_{\mathbf{p}_1} = C^*_{\mathbf{p}_2,\alpha}$, that is, $H_{\mathbf{p}_1}$ and $C^*_{\mathbf{p}_2,\alpha}$ have exactly the same sets of vertices and edges.
Since $\mathbf{p}_1\subseteq \mathbf{p}_2$, it follows the anti-monotonicity~\cite{agrawal1994fast,han2000mining} that $\forall v_i \in V, f_i(\mathbf{p}_1) \geq f_i(\mathbf{p}_2)$. Thus, $G_{\mathbf{p}_2}\subseteq G_{\mathbf{p}_1}$. 
Since $C^*_{\mathbf{p}_2,\alpha}\subseteq G_{\mathbf{p}_2}$, $C^*_{\mathbf{p}_2,\alpha}\subseteq G_{\mathbf{p}_1}$.
Therefore, $H_{\mathbf{p}_1}$ exists.

Next, we prove $H_{\mathbf{p}_1}$ and $C^*_{\mathbf{p}_2,\alpha}$ are both $(\mathbf{p}_1,\alpha)$-trusses.
Since $\forall v_i \in V, f_i(\mathbf{p}_1) \geq f_i(\mathbf{p}_2)$, the following inequality holds for every triangle $\triangle_{ijk}$ in $H_{\mathbf{p}_1}$.
\begin{equation}\nonumber
\label{eq:mono}
\begin{split}
	\min(f_i(\mathbf{p}_1), f_j(\mathbf{p}_1), f_k(\mathbf{p}_1)) \geq\min(f_i(\mathbf{p}_2), f_j(\mathbf{p}_2), f_k(\mathbf{p}_2))
\end{split}
\end{equation}

Since $H_{\mathbf{p}_1}=C^*_{\mathbf{p}_2,\alpha}$, it follows the above inequality that
$eco_{ij}(H_{\mathbf{p}_1})\geq eco_{ij}(C^*_{\mathbf{p}_2,\alpha})$ for every edge $e_{ij}$ in $H_{\mathbf{p}_1}$.

Since $C^*_{\mathbf{p}_2,\alpha}$ is the maximal $(\mathbf{p}_2,\alpha)$-truss, $eco_{ij}(C^*_{\mathbf{p}_2,\alpha})>\alpha$ for every edge $e_{ij}$ in $C^*_{\mathbf{p}_2,\alpha}$.

Now we can conclude from the above that $eco_{ij}(H_{\mathbf{p}_1})\geq eco_{ij}(C^*_{\mathbf{p}_2,\alpha}) > \alpha$ for every edge $e_{ij}$ in $H_{\mathbf{p}_1}$.
This means both $H_{\mathbf{p}_1}$ and $C^*_{\mathbf{p}_2,\alpha}$ are $(\mathbf{p}_1,\alpha)$-trusses. 
The theorem follows.
\end{proof}

\begin{proposition}[Pattern Anti-monotonicity]
\label{Prop:pam}
For patterns $\mathbf{p}_1\subseteq \mathbf{p}_2$ and a cohesion threshold $\alpha$, 
\begin{enumerate}
	\item If $C^*_{\mathbf{p}_2,\alpha}\neq \emptyset$, then $C^*_{\mathbf{p}_1,\alpha}\neq\emptyset$.
	\item If $C^*_{\mathbf{p}_1,\alpha}= \emptyset$, then $C^*_{\mathbf{p}_2,\alpha}= \emptyset$.
\end{enumerate}
\end{proposition}
\begin{proof}
According to Theorem~\ref{Prop:gam}, since $\mathbf{p}_1\subseteq \mathbf{p}_2$, $C^*_{\mathbf{p}_2,\alpha}\subseteq C^*_{\mathbf{p}_1,\alpha}$.
The proposition follows immediately.
\end{proof}

\begin{proposition}[Graph Intersection Property]
\label{Lem:gip}
If $\mathbf{p}_1 \subseteq \mathbf{p}_3$ and $\mathbf{p}_2 \subseteq \mathbf{p}_3$, then $C^*_{\mathbf{p}_3,\alpha}\subseteq C^*_{\mathbf{p}_1,\alpha} \cap C^*_{\mathbf{p}_2,\alpha}$.
\end{proposition}

\begin{proof}
By Theorem~\ref{Prop:gam}, since $\mathbf{p}_1 \subseteq \mathbf{p}_3$, $C^*_{\mathbf{p}_3,\alpha}\subseteq C^*_{\mathbf{p}_1,\alpha}$. 
Similarly, $C^*_{\mathbf{p}_3,\alpha}\subseteq C^*_{\mathbf{p}_2,\alpha}$. The proposition follows.
\end{proof}

\nop{
The proofs of Propositions~\ref{Prop:pam}, \ref{Prop:gam} and \ref{Lem:gip} can be found in Sections~\ref{Apd:pam}, \ref{Apd:gam} and \ref{Apd:gip} of the appendix, respectively.
}

\begin{algorithm}[t]
\caption{Generate Apriori Candidate Patterns}
\label{Alg:apriori}
\KwIn{The length-$(k-1)$ qualified patterns $\mathcal{P}^{k-1}$.}
\KwOut{The set of length-$k$ candidate patterns $\mathcal{M}^k$.}
\BlankLine
\begin{algorithmic}[1]
        \STATE Initialize: $\mathcal{M}^k\leftarrow \emptyset$.
    	\FOR{$\{\mathbf{p}, \mathbf{q}\} \subset \mathcal{P}^{k-1} \land |\mathbf{p} \cup \mathbf{q}| = k$}
		\STATE $\mathbf{h}\leftarrow \mathbf{p} \cup \mathbf{q}$.
		\STATE \textbf{if} all length-$(k-1)$ sub-patterns of $\mathbf{h}$ are qualified \textbf{then} $\mathcal{M}^k\leftarrow \mathcal{M}^k \cup \mathbf{h}$.
		
	\ENDFOR
\BlankLine
\RETURN $\mathcal{M}^k$.
\end{algorithmic}
\end{algorithm}

\begin{algorithm}[t]
\caption{Theme Community Finder Apriori}
\label{Alg:tcfa}
\KwIn{A DBN $G$ and a user input $\alpha$.}
\KwOut{The set of maximal $(\mathbf{p},\alpha)$-trusses $\mathbb{C}$ in $G$.}
\BlankLine
\begin{algorithmic}[1]
    \STATE Initialize: $\mathcal{P}^1$, $\mathbb{C}\leftarrow \mathbb{C}^1$, $k\leftarrow2$.
    \WHILE{$\mathcal{P}^{k-1}\neq \emptyset$}
    	\STATE Call Algorithm~\ref{Alg:apriori}: $\mathcal{M}^k\leftarrow \mathcal{P}^{k-1}$.
	\STATE $\mathcal{P}^k\leftarrow \emptyset$, $\mathbb{C}^k\leftarrow \emptyset$.
	 \FOR{each length-$k$ pattern $\mathbf{h} \in \mathcal{M}^k$}
    		\STATE Induce $G_{\mathbf{h}}$ from $G$. 
		\STATE Compute $C^*_{\mathbf{h},\alpha}$ using $G_{\mathbf{h}}$ by Algorithm~\ref{Alg:mptd}.
		\STATE \textbf{if} $C^*_{\mathbf{h},\alpha} \neq \emptyset$ \textbf{then} $\mathbb{C}^k\leftarrow \mathbb{C}^k \cup C^*_{\mathbf{h},\alpha}$, $\mathcal{P}^k\leftarrow \mathcal{P}^k\cup \mathbf{h}$.
    	\ENDFOR
	\STATE $\mathbb{C}\leftarrow \mathbb{C} \cup \mathbb{C}^k$ and $k\leftarrow k+1$.
    \ENDWHILE
\BlankLine
\RETURN $\mathbb{C}$.
\end{algorithmic}
\end{algorithm}

\subsection{Theme Community Finder Apriori}
\label{Sec:tcfa}
In this subsection, we introduce algorithm \emph{Theme Community Finder Apriori} (TCFA) to solve the theme community finding problem.
The key idea of TCFA is to improve theme community finding efficiency by early pruning unqualified patterns in an Apriori-like manner~\cite{agrawal1994fast}.

A pattern $\mathbf{p}$ is said to be \emph{unqualified} if $C^*_{\mathbf{p},\alpha}=\emptyset$, and to be \emph{qualified} if $C^*_{\mathbf{p},\alpha}\neq\emptyset$.
For two patterns $\mathbf{p}_1$ and $\mathbf{p}_2$, if $\mathbf{p}_1 \subseteq \mathbf{p}_2$, $\mathbf{p}_1$ is called a \emph{sub-pattern} of $\mathbf{p}_2$.

According to the second item in Proposition~\ref{Prop:pam}, for two patterns $\mathbf{p}_1$ and $\mathbf{p}_2$, if $\mathbf{p}_1 \subseteq \mathbf{p}_2$ and $\mathbf{p}_1$ is unqualified, then $\mathbf{p}_2$ is unqualified, either, thus $\mathbf{p}_2$ can be immediately pruned. Therefore, we can prune a length-$k$ pattern if any of its length-$(k-1)$ sub-patterns is unqualified.

Algorithm~\ref{Alg:apriori} shows how we generate the set of length-$k$ candidate patterns by retaining only the length-$k$ patterns whose all length-$(k-1)$ sub-patterns are qualified.
\nop{
Specifically, line 3 generates a possible length-$k$ candidate pattern $\mathbf{p}^k$ by merging two length-$(k-1)$ qualified patterns $\{\mathbf{p}^{k-1}, \mathbf{q}^{k-1}\}$. 
Line 4 retains $\mathbf{p}^k$ as a candidate pattern if all its length-$(k-1)$ sub-patterns are qualified. 
}

Algorithm~\ref{Alg:tcfa} introduces the details of TCFA. Line 1 computes the set of length-1 qualified patterns $\mathcal{P}^1=\{\mathbf{p} \subset S \mid C^*_{\mathbf{p},\alpha}\neq\emptyset, |\mathbf{p}| = 1\}$ and the corresponding set of maximal $(\mathbf{p},\alpha)$-trusses $\mathbb{C}^1 = \{C^*_{\mathbf{p},\alpha} \mid \mathbf{p}\in \mathcal{P}^1 \}$. 
This requires to run MTD on each theme network induced by a single item in $S$. 
\nop{Since the theme networks induced by different items are independent, we can run this process in parallel. Our implementation use multiple threads for this step.}
Line 3 calls Algorithm~\ref{Alg:apriori} to generate the set of length-$k$ candidate patterns $\mathcal{M}^k$.
Lines 5-9 remove the unqualified candidate patterns in $\mathcal{M}^k$ by discarding every candidate pattern that cannot form a non-empty maximal $(\mathbf{p},\alpha)$-truss.
In this way, we iteratively generate the set of length-$k$ qualified patterns $\mathcal{P}^k$ from $\mathcal{P}^{k-1}$ until no qualified patterns can be found. 
Last, the exact set of maximal $(\mathbf{p},\alpha)$-trusses $\mathbb{C}$ is returned.

\nop{Specifically, step 3 calls Algorithm~\ref{Alg:apriori} to generate candidate pattern set $L^k$; steps 4-9 obtain $P^k$ by retaining only the qualified patterns in $L^k$.}

\nop{
Algorithm~\ref{Alg:apriori} shows how we generate the length-$k$ candidate patterns and prune unqualified patterns in an Apriori-like manner.
According to the second property of Proposition~\ref{Prop:pam}, for two patterns $p_1$ and $p_2$, if $p_1 \subset p_2$ and $p_1$ is unqualified, then $p_2$ is unqualified and can be immediately pruned without running Algorithm~\ref{Alg:mptd} on $G_{p_2}$. 
As a result, we perform efficient pattern pruning by retaining only the length-$k$ candidate patterns whose length-$(k-1)$ sub-patterns are all qualified patterns.
Specifically, step 3 generates a possible length-$k$ candidate pattern $p^k$ by merging two length-$(k-1)$ qualified patterns $\{p^{k-1}, q^{k-1}\}$. Step 4 retains $p^k$ as a candidate pattern if all its length-$(k-1)$ sub-patterns are qualified. 
}

Comparing with the baseline TCS in Section~\ref{sec:tcs}, TCFA achieves a good efficiency improvement by effectively pruning a large number of unqualified patterns using the Apriori-like method.
However, due to the limitation of Apriori~\cite{agrawal1994fast}, the set of candidate patterns $\mathcal{M}^k$ is often very large and still contains many unqualified candidate patterns. 
Consequently, Lines 5-9 of Algorithm~\ref{Alg:tcfa} become the bottleneck of TCFA.
We solve this problem next.

\nop{
can be identified and removed by running Algorithm~\ref{Alg:mptd} on the corresponding theme networks, 

Thus, checking the qualification of such unqualified patterns in lines 5-9 of Algorithm~\ref{Alg:tcfa} becomes the bottleneck of computational efficiency. 
}

\subsection{Theme Community Finder Intersection}
\label{Sec:tcfi}
The \emph{Theme Community Finder Intersection} (TCFI) method significantly improves the efficiency of TCFA by pruning unqualified patterns in $\mathcal{M}^k$ using Proposition~\ref{Lem:gip}.

Consider pattern $\mathbf{h}$ of length $k$ and patterns $\mathbf{p}$ and $\mathbf{q}$ both of length $k-1$.
According to Proposition~\ref{Lem:gip}, if $\mathbf{h} = \mathbf{p} \cup \mathbf{q}$, then $C^*_{\mathbf{h},\alpha}\subseteq C^*_{\mathbf{p},\alpha} \cap C^*_{\mathbf{q},\alpha}$. 
Therefore, let $Z=C^*_{\mathbf{p},\alpha} \cap C^*_{\mathbf{q},\alpha}$, if $Z=\emptyset$, then $C^*_{\mathbf{h},\alpha}=\emptyset$.  That is, we can prune $\mathbf{h}$ immediately.
If $Z\neq\emptyset$, we can induce theme network $G_{\mathbf{h}}$ from $Z$
and find $C^*_{\mathbf{h},\alpha}$ within $G_{\mathbf{h}}$ by MTD.

Accordingly, TCFI improves TCFA by modifying only Line 6 of Algorithm~\ref{Alg:tcfa}.
Instead of inducing $G_{\mathbf{h}}$ from $G$, TCFI induces $G_{\mathbf{h}}$ from $Z$ when $Z\neq \emptyset$. 
Here, $Z=C^*_{\mathbf{p},\alpha} \cap C^*_{\mathbf{q},\alpha}$ where $\mathbf{p}$ and $\mathbf{q}$ are qualified patterns in $\mathcal{P}^{k-1}$ such that $\mathbf{h} = \mathbf{p} \cup \mathbf{q}$.

TCFI dramatically improves the detection efficiency.
First, TCFI prunes a large number of candidate patterns in $\mathcal{M}^k$ by efficiently checking whether $Z=\emptyset$.
Second, when $Z\neq\emptyset$, inducing $G_{\mathbf{h}}$ from $Z$ is more efficient than inducing $G_{\mathbf{h}}$ from $G$, since $Z$ is often much smaller than $G$. 
Third, $G_{\mathbf{h}}$ induced from $Z$ is often much smaller than $G_{\mathbf{h}}$ induced from $G$, which significantly reduces the time cost of running MTD on $G_{\mathbf{h}}$. 
Last, according to Theorem~\ref{Prop:gam}, the size of a maximal $(\mathbf{p},\alpha)$-truss decreases when the length of the pattern increases. Thus, when a pattern grows longer, the size of $Z$ decreases rapidly, which significantly improves the pruning effectiveness of TCFI.


\section{Theme Community Indexing}
\label{sec:index}
\nop{
When a user inputs a new threshold $\alpha$, TCS, TCFA and TCFI have to recompute from scratch. 
Can we save the re-computation cost by decomposing and indexing all maximal $(\mathbf{p},\alpha)$-trusses?
In this section, we propose the \emph{Theme Community Tree} (TC-Tree) for fast query answering.
}

\mc{
In practice, different users may be interested in theme communities in different maximal $(\mathbf{p}, \alpha)$-trusses.
Unfortunately, for every new threshold $\alpha$, TCS, TCFA and TCFI have to recompute from scratch.
Can we save the re-computation cost by providing a fast query answering service that allows users to efficiently explore a DBN and quickly retrieve theme communities of their own interest?
In this section, we propose \emph{Theme Community Tree} (TC-Tree) to provide fast query answering service by decomposing and indexing all maximal $(\mathbf{p},\alpha)$-trusses in a DBN.
}
\nop{
TCS, TCFA and TCFI find a large number of maximal $(\mathbf{p},\alpha)$-trusses

The maximal $(\mathbf{p},\alpha)$-trusses detected by TCFA and TCFI contain a large number of theme communities with different themes and cohesiveness.
}

We first introduce how to decompose maximal $(\mathbf{p},\alpha)$-truss.
Then, we illustrate how to build TC-Tree with decomposed maximal $(\mathbf{p},\alpha)$-trusses. 
Last, we present a query answering method that can efficiently retrieve a ranked list of theme communities to answer a user query, and can also recommend a ranked list of meaningful new queries to address user disappointment when a query does not return any community.

\subsection{Maximal \titlepa{11}-Truss Decomposition}
\label{Sec:mpt_dec}
In this subsection, we introduce how to decompose a maximal $(\mathbf{p},\alpha)$-truss into multiple disjoint sets of edges.

\nop{
\begin{property}
\label{Obs:anti_alpha}
$\forall p\subseteq S$, if $\alpha_2 > \alpha_1 \geq 0$, then $C^*_{p}(\alpha_2)\subseteq C^*_{p}(\alpha_1)$.  
\end{property}

Property~\ref{Obs:anti_alpha} shows that the size of maximal $(\mathbf{p},\alpha)$-truss reduces when the threshold increases from $\alpha_1$ to $\alpha_2$.
According to Definition~\ref{Def:pattern_truss}, the cohesion of all edges in $C^*_{p}(\alpha_1)$ are larger than $\alpha_1$. Since $\alpha_2 > \alpha_1$, there can be some edges in $C^*_{p}(\alpha_1)$ whose cohesion is smaller than $\alpha_2$.
Since $C^*_{p}(\alpha_2)$ is formed by removing such edges from $C^*_{p}(\alpha_1)$, we have $C^*_{p}(\alpha_2)\subseteq C^*_{p}(\alpha_1)$.
}

\begin{theorem}
\label{Obs:discrete}
Given a maximal $(\mathbf{p},\alpha_1)$-truss $C^*_{\mathbf{p},\alpha_1}$ with minimum edge cohesion $\beta_{\mathbf{p},\alpha_1}$, for any cohesion threshold $\alpha_2 \geq \beta_{\mathbf{p},\alpha_1}$, $C^*_{\mathbf{p},\alpha_2}\subset C^*_{\mathbf{p},\alpha_1}$.
\end{theorem}


\begin{proof}
First, we prove $\alpha_2 > \alpha_1$. 
By Definition~\ref{Def:pattern_truss}, for any edge $e_{ij}$ in $C^*_{\mathbf{p},\alpha_1}$, $eco_{ij}(C^*_{\mathbf{p},\alpha_1}) > \alpha_1$. 
Since $\beta_{\mathbf{p},\alpha_1}$ is the minimum edge cohesion of $C^*_{\mathbf{p},\alpha_1}$, $\beta_{\mathbf{p},\alpha_1} > \alpha_1$. 
Since $\alpha_2 \geq \beta_{\mathbf{p},\alpha_1}$, $\alpha_2 > \alpha_1$.

Second, we prove $C^*_{\mathbf{p},\alpha_2}\subseteq C^*_{\mathbf{p},\alpha_1}$. 
Since $\alpha_2 > \alpha_1$, it follows Definition~\ref{Def:pattern_truss} that, for any edge $e_{ij}$ in $C^*_{\mathbf{p},\alpha_2}$, $eco_{ij}(C^*_{\mathbf{p},\alpha_2}) > \alpha_2> \alpha_1$.
This means $C^*_{\mathbf{p},\alpha_2}$ is also a $(\mathbf{p},\alpha_1)$-truss.
Since $C^*_{\mathbf{p},\alpha_1}$ is \nop{\todo{a maximal or the maximum?}}the maximal $(\mathbf{p},\alpha_1)$-truss, $C^*_{\mathbf{p},\alpha_2}\subseteq C^*_{\mathbf{p},\alpha_1}$.

Last, we prove $C^*_{\mathbf{p},\alpha_2}\neq C^*_{\mathbf{p},\alpha_1}$. 
Let $e^*_{ij}$ be the edge that has the minimum edge cohesion $\beta_{\mathbf{p},\alpha_1}$ in $ C^*_{\mathbf{p},\alpha_1}$. 
Since $\alpha_2\geq \beta_{\mathbf{p},\alpha_1}$, $e^*_{ij}$ is not an edge of $C^*_{\mathbf{p},\alpha_2}$. 
Thus, $C^*_{\mathbf{p},\alpha_2}\neq C^*_{\mathbf{p},\alpha_1}$.
Recall that $C^*_{\mathbf{p},\alpha_2}\subseteq C^*_{\mathbf{p},\alpha_1}$, the theorem follows.
\end{proof}

Theorem~\ref{Obs:discrete} indicates that the size of $C^*_{\mathbf{p},\alpha_1}$ is smaller than $C^*_{\mathbf{p},\alpha_2}$ only when $\alpha_2 \geq \beta_{\mathbf{p},\alpha_1}$.
Thus, we can iteratively compute a sequence of ascending cohesion thresholds $\mathcal{A}_\mathbf{p}=\alpha_0, \alpha_1, \ldots, \alpha_h$, where $\alpha_0 = 0$, $\alpha_k = \beta_{\mathbf{p},\alpha_{k-1}}$ for $k\in\{1, \ldots, h\}$, and $\alpha_h$ is the largest $\alpha$ in $G_\mathbf{p}$ such that $C^*_{\mathbf{p},\alpha}=\emptyset$ for all $\alpha\geq \alpha_h$.

We use $\mathcal{A}_\mathbf{p}$ to decompose a maximal $(\mathbf{p}, \alpha)$-truss as follows. First, we call MTD 
to compute $C^*_{\mathbf{p},\alpha_0}$, which is the largest maximal $(\mathbf{p},\alpha)$-truss in $G_\mathbf{p}$. 
Then, for $\alpha_1, \ldots, \alpha_h$, we decompose $C^*_{\mathbf{p},\alpha_0}$ into a sequence of sets of edges $R_{\mathbf{p},\alpha_1}, \ldots, R_{\mathbf{p},\alpha_h}$, where $R_{\mathbf{p},\alpha_k}$ is the set of edges that are contained in $C^*_{\mathbf{p},\alpha_{k-1}}$ but not in $C^*_{\mathbf{p},\alpha_k}$.

The decomposition results are stored in a linked list $\mathcal{L}_\mathbf{p}=\mathcal{L}_{\mathbf{p},\alpha_1}, \ldots, \mathcal{L}_{\mathbf{p},\alpha_h}$, where the $k$-th node stores $\mathcal{L}_{\mathbf{p},\alpha_k} = (\alpha_k, R_{\mathbf{p},\alpha_k})$. Since $\mathcal{L}_\mathbf{p}$ stores the same number of edges as $C^*_{\mathbf{p},\alpha_0}$, it does not incur much extra memory cost. 

\nop{This decomposition iterates until all edges in $C^*_{\mathbf{p},\alpha_0}$ are removed.}

\nop{That is, $R_{\mathbf{p},\alpha_k}=E^*_{\mathbf{p},\alpha_{k-1}} \setminus E^*_{\mathbf{p},\alpha_k}$, where $E^*_{\mathbf{p},\alpha_{k-1}}$ and $E^*_{\mathbf{p},\alpha_k}$ are the set of edges of $C^*_{\mathbf{p},\alpha_{k-1}}$ and $C^*_{\mathbf{p},\alpha_k}$, respectively.}

\nop{
$C^*_{p}(\alpha_i)$ reduces to $C^*_{p}(\alpha_{i+1})$ when $\alpha_{i+1}\in [\beta_i, \beta_{i+1})$, where $\beta_i = \psi_p(\alpha_i)$.
}

\nop{
Denote by $E^*_{\mathbf{p}, \alpha}$ the set of edges of the maximal $(\mathbf{p}, \alpha)$-truss $C^*_{\mathbf{p}, \alpha}$. 
We can use $\mathcal{L}_\mathbf{p}$ to efficiently compute $E^*_{\mathbf{p}, \alpha}$ as $E^*_{\mathbf{p}, \alpha}=\bigcup_{\alpha_k>\alpha} R_{\mathbf{p},\alpha_k}$, and directly induce $C^*_{\mathbf{p}, \alpha}$ from  $E^*_{\mathbf{p}, \alpha}$.
Then, we can obtain the set of theme communities $\mathbb{T}_{\mathbf{p}, \alpha}=\{T^1_{\mathbf{p}, \alpha},\ldots, T^m_{\mathbf{p}, \alpha}\}$ by finding the set of maximal connected subgraphs in $C^*_{\mathbf{p},\alpha}$.
}

\nop{
$C^*_{\mathbf{p}, \alpha}$ by first obtaining $E^*_{\mathbf{p}, \alpha}=\bigcup_{\alpha_k>\alpha} R_{\mathbf{p},\alpha_k}$, then inducing $C^*_{\mathbf{p}, \alpha}$ from  $E^*_{\mathbf{p}, \alpha}$.
}

Using $\mathcal{L}_\mathbf{p}$, we can efficiently get the set of theme communities $\mathbb{T}_{\mathbf{p}, \alpha}=\{T^1_{\mathbf{p}, \alpha},\ldots, T^m_{\mathbf{p}, \alpha}\}$ in two steps. 
Denote by $E^*_{\mathbf{p},\alpha}$ the set of edges of $C^*_{\mathbf{p},\alpha}$, in the first step, we compute $E^*_{\mathbf{p},\alpha}=\bigcup_{\alpha_k>\alpha} R_{\mathbf{p},\alpha_k}$, and induce $C^*_{\mathbf{p},\alpha}$ from $E^*_{\mathbf{p},\alpha}$.
In the second step, we obtain $\mathbb{T}_{\mathbf{p}, \alpha}$ by finding the set of maximal connected subgraphs in $C^*_{\mathbf{p},\alpha}$.

\nop{
Using $\mathcal{L}_\mathbf{p}$, we can efficiently get the maximal $(\mathbf{p}, \alpha)$-truss $C^*_{\mathbf{p}, \alpha}$ in two steps: 
1) obtain the set of edges of $C^*_{\mathbf{p},\alpha}$ by $\bigcup\limits_{\alpha_k>\alpha} R_{\mathbf{p},\alpha_k}$; 2) induce set of vertices from the set of edges.

the set of theme communities $\mathbb{T}_{\mathbf{p}, \alpha}=\{T^1_{\mathbf{p}, \alpha},\ldots, T^m_{\mathbf{p}, \alpha}\}$ in two steps: 
1) we get the set of edges of $C^*_{\mathbf{p},\alpha}$ by $\bigcup\limits_{\alpha_k>\alpha} R_{\mathbf{p},\alpha_k}$
and inducing $V^*_{\mathbf{p},\alpha}$ from $E^*_{\mathbf{p},\alpha}$;
2) we obtain $\mathbb{T}_{\mathbf{p}, \alpha}$ by finding the set of maximal connected subgraphs in $C^*_{\mathbf{p},\alpha}$.
}

Next, we introduce how to compute the cohesiveness of every theme community in $\mathbb{T}_{\mathbf{p}, \alpha}$.

\nop{
$E^*_{\mathbf{p},\alpha}=\bigcup\limits_{\alpha_k>\alpha} R_{\mathbf{p},\alpha_k}$
}

\begin{theorem}
\label{thm:tc_cohesiveness}
Given a theme community $T^i_{\mathbf{p}, \alpha}\in\mathbb{T}_{\mathbf{p}, \alpha}$, denote by $\gamma$ the cohesiveness of  $T^i_{\mathbf{p}, \alpha}$, if $\alpha_k$ is the smallest cohesion threshold in $\mathcal{A}_\mathbf{p}$ such that $T^i_{\mathbf{p}, \alpha}\not\subseteq C^*_{\mathbf{p},\alpha_k}$, then $\gamma=\alpha_{k}$.
\end{theorem}

\begin{proof}
First, we prove $\gamma \leq \alpha_{k}$.
By Definition~\ref{Def:cohe_theme_community}, $\gamma$ is the minimum cohesion of all edges in $T^i_{\mathbf{p}, \alpha}$. 
Since $T^i_{\mathbf{p}, \alpha}\not\subseteq C^*_{\mathbf{p},\alpha_k}$, it follows Definition~\ref{Def:pattern_truss} that $\gamma \leq \alpha_{k}$.

Second, we prove $\gamma \geq \alpha_{k}$. Since $\alpha_k$ is the smallest cohesion threshold in $\mathcal{A}_\mathbf{p}$ such that $T^i_{\mathbf{p}, \alpha}\not\subseteq C^*_{\mathbf{p},\alpha_k}$, we have $T^i_{\mathbf{p}, \alpha}\subseteq C^*_{\mathbf{p},\alpha_{k-1}}$, and thus $\gamma \geq \beta_{\mathbf{p}, \alpha_{k-1}}$. 
Since $\alpha_{k} = \beta_{\mathbf{p}, \alpha_{k-1}}$, $\gamma \geq \alpha_{k} $. The theorem follows.
\end{proof}

According to Theorem~\ref{thm:tc_cohesiveness}, $\mathcal{A}_\mathbf{p}$ is exactly the set of the cohesiveness of all theme communities in $G_\mathbf{p}$. 
To compute the cohesiveness of a theme community $T^i_{\mathbf{p}, \alpha}\in\mathbb{T}_{\mathbf{p}, \alpha}$, we simply use $\mathcal{L}_\mathbf{p}$ to find the smallest $\alpha_k\in\mathcal{A}_\mathbf{p}$ such that $T^i_{\mathbf{p}, \alpha}\not\subseteq C^*_{\mathbf{p},\alpha_k}$.

Next, we introduce how to use the decomposition property of maximal $(\mathbf{p}, \alpha)$-truss to build a TC-Tree.

\begin{figure}[t]
\centering
\includegraphics[width=85mm]{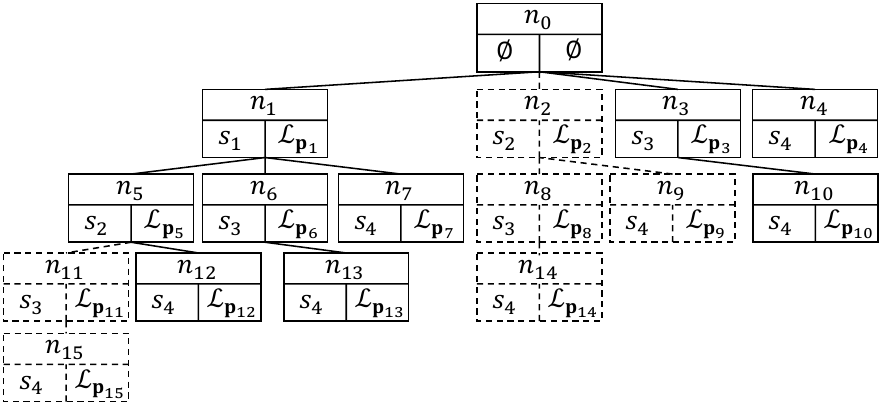}

\vspace{-2mm}
\caption{An example of SE-Tree and TC-Tree when $S=\{s_1, s_2, s_3, s_4\}$ and $\mathcal{L}_{\mathbf{p}_0}=\mathcal{L}_{\mathbf{p}_2}=\mathcal{L}_{\mathbf{p}_8}=\mathcal{L}_{\mathbf{p}_9}=\mathcal{L}_{\mathbf{p}_{11}}=\mathcal{L}_{\mathbf{p}_{14}}=\mathcal{L}_{\mathbf{p}_{15}}=\emptyset$. SE-Tree includes all nodes marked in solid and dashed lines. TC-tree contains only the nodes in solid line.}
\label{Fig:tct}
\end{figure}

\subsection{Theme Community Tree}
A TC-Tree, denoted by $\mathcal{T}$, is an extension of a \emph{set enumeration tree} (SE-Tree)~\cite{rymon1992search} and is carefully customized for efficient theme community indexing and query answering.

\nop{
A SE-Tree is a basic data structure that enumerates all the patterns in the power set of $S$. 
Each node of the SE-Tree stores an item in $S$. 
The $i$-th node of the SE-Tree is denoted by $n_i$, the item stored in $n_i$ is denoted by $s_{n_i}\in S$.
We map each item in $S$ to a unique rank using a \emph{ranking function} $\psi: S \rightarrow \mathbb{Z}^+$. 
For the $j$-th item $s_j$ in $S$, the ranking function maps $s_j$ to its rank $j$ in $S$, that is, $\psi(s_j) = j$.
\nop{We compare a pair of items $s_i$ and $s_j$ in $S$ by a \emph{pre-imposed order}, that is, \emph{if $j > i$ then $s_j > s_i$}.}
Then, for every node $n_i$ of the SE-Tree, we build a child of $n_i$ for each item $s_j\in S$ that satisfies $\psi(s_j) > \psi(s_{n_i})$.
In this way, every node $n_i$ of the SE-Tree uniquely represents a pattern in the power set of $S$, denoted by $\mathbf{p}_i\subseteq S$, which is the union of the items stored in all the nodes along the path from the root to $n_i$. 
As shown in Figure~\ref{Fig:tct}, the SE-Tree of $S=\{s_1, s_2, s_3, s_4\}$ has 16 nodes $\{n_0, n_1, n_2, \cdots, n_{15}\}$.
For node $n_{13}$, the path from the root to $n_{13}$ contains nodes $\{n_0, n_1, n_6, n_{13}\}$, where $s_{n_0}=\emptyset$, $s_{n_1}=s_1$, $s_{n_6}=s_3$ and $s_{n_{13}}=s_4$, thus $\mathbf{p}_{13}=\{s_1, s_3, s_4\}$.
}

A SE-Tree is a basic data structure that enumerates all the subsets of a set $S$. 
A total order $\prec$ on the items in $S$ is assumed. 
Thus, any subset of $S$ can be written as a sequence of items in the order of $\prec$.
\nop{ \todo{This is not right. $\prec$ is the order on the alphabet.  Subsets of $S$ has to be ordered using the dictionary order based on $\prec$ instead of $\prec$ itself.}}

Every node of a SE-Tree uniquely represents a subset of $S$.
The root node represents empty set $\emptyset$. 
For subsets $S_1$ and $S_2$ of $S$, the node representing $S_2$ is the child of the node representing $S_1$, if $S_1 \subset S_2$, $|S_2 \setminus S_1|=1$, and $S_1$ is a prefix of $S_2$ when $S_1$ and $S_2$ are written as sequences of items in order $\prec$. 
Each node of a SE-Tree only stores the item in $S$ that is appended to the parent node to extend the child from the parent.
In this way, the set of items represented by node $n_i$ is the union of the items stored in all the nodes along the path from the root to $n_i$.
Figure~\ref{Fig:tct} shows an example of the SE-tree of set $S=\{s_1, s_2, s_3, s_4\}$.  
For node $n_{13}$, the path from the root to $n_{13}$ contains nodes $n_0$-$n_1$-$n_6$-$n_{13}$, thus the set of items represented by $n_{13}$ is $\{s_1, s_3, s_4\}$.

A TC-Tree is an extension of a SE-Tree.
In a TC-Tree, each node $n_i$ represents a pattern $\mathbf{p}_i$, which is a subset of $S$. 
The item stored in $n_i$ is denoted by $s_{n_i}$. 
We also store the decomposed maximal $(\mathbf{p},\alpha)$-truss $\mathcal{L}_{\mathbf{p}_i}$ in $n_i$. 
To save memory, we omit the nodes $n_j$ $(j\geq 1)$ whose decomposed maximal $(\mathbf{p},\alpha)$-trusses are $\mathcal{L}_{\mathbf{p}_j}=\emptyset$. 

We can build a TC-Tree in a top-down manner efficiently.
If $\mathcal{L}_{\mathbf{p}_j}=\emptyset$, we can prune the entire subtree rooted at $n_j$ immediately.
This is because, for node $n_j$ and its descendant $n_d$, we have $\mathbf{p}_j\subset \mathbf{p}_d$. Since $\mathcal{L}_{\mathbf{p}_j}=\emptyset$, we can derive from Proposition~\ref{Prop:pam} that $\mathcal{L}_{\mathbf{p}_d}=\emptyset$. As a result, all descendants of $n_j$ can be immediately pruned.

Algorithm~\ref{Alg:tctb} gives the details of building a TC-Tree $\mathcal{T}$. 
Lines 2-5 generate the nodes at the first layer of $\mathcal{T}$.
Since the theme networks induced by different items in $S$ are independent, we can compute $\mathcal{L}_{\mathbf{p}_i}$ in parallel. Our implementation uses multiple threads for this step.
Lines 6-12 iteratively build the rest of the nodes of $\mathcal{T}$ in breadth first order.
Here, $n_f.siblings$ is the set of nodes that have the same parent as $n_f$. The children of $n_f$, denoted by $n_c$, are built in Lines 8-11.
\nop{The pre-imposed order of items (i.e., $n_b.s_b > n_f.s_f$) is applied in Line 9 to ensure that the TC-Tree is a subtree of SE-Tree.  \mc{What do you mean by ``sub-tree'' here?, Do you mean a sub-graph?}}
In Line 9, we apply Proposition~\ref{Lem:gip} to efficiently calculate $\mathcal{L}_{\mathbf{p}_c}$.
Since $\mathbf{p}_c=\mathbf{p}_f \cup \mathbf{p}_b$, we have $\mathbf{p}_f\subset \mathbf{p}_c$ and $\mathbf{p}_b\subset \mathbf{p}_c$. 
From Proposition~\ref{Lem:gip}, we know $C^*_{\mathbf{p}_c, \alpha_0}\subseteq C^*_{\mathbf{p}_f, \alpha_0} \cap C^*_{\mathbf{p}_b, \alpha_0}$.
Therefore, we can find $C^*_{\mathbf{p}_c, \alpha_0}$ within a small subgraph $C^*_{\mathbf{p}_f, \alpha_0} \cap C^*_{\mathbf{p}_b, \alpha_0}$ using MTD,
and then get $\mathcal{L}_{\mathbf{p}_c}$ by decomposing $C^*_{\mathbf{p}_c, \alpha_0}$.

In summary, every node of a TC-Tree stores the decomposed maximal $(\mathbf{p}, \alpha)$-truss $\mathcal{L}_\mathbf{p}$ of a unique pattern $\mathbf{p}\subseteq S$. 
Next, we introduce how to efficiently query a TC-Tree.

\subsection{Querying and Query Recommendation}

In this subsection, we first introduce how to query a TC-Tree $\mathcal{T}$ by a \textbf{query pattern} $\mathbf{q}$, then we illustrate how to recommend new queries based on a user-provided query pattern $\mathbf{q}$.

When querying a TC-Tree $\mathcal{T}$ by a query pattern $\mathbf{q}$, the \emph{answer} to the query, denoted by $\mathcal{R}_\mathbf{q}$, is a ranked list of  all the theme communities in $G_\mathbf{q}$, that is, $\cup_{\alpha_i\in\mathcal{A}_\mathbf{p}} \mathbb{T}_{\mathbf{q}, \alpha_i}$.

\nop{
all maximal $(\mathbf{q}, \alpha)$-trusses $C^*_{\mathbf{q}, \alpha}$ for every $\alpha\in\mathcal{A}_\mathbf{q}$, 
that is, .
}

We obtain $\mathcal{R}_\mathbf{q}$ in the following steps:
First, we find the node $n_i$ in $\mathcal{T}$ such that the pattern of $n_i$ is $\mathbf{p}_i=\mathbf{q}$. Second, we obtain the set of theme communities in $G_\mathbf{q}$ using the $\mathcal{L}_\mathbf{q}$ stored in $n_i$.  Last, we obtain $\mathcal{R}_\mathbf{q}$ by sorting the theme communities in the descending order of their cohesiveness.

\nop{

2) calculate $C^*_{\mathbf{q}, \alpha_\mathbf{q}}$ using the $\mathcal{L}_\mathbf{q}$ stored in $n_i$; 3) obtain the set of theme communities $\mathbb{T}_{\mathbf{q}, \alpha_\mathbf{q}}$ by finding the set of maximal connected subgraphs in $C^*_{\mathbf{q}, \alpha_\mathbf{q}}$; and 4) obtain $\mathcal{R}_{\mathbf{q}, \alpha_\mathbf{q}}$ by sorting the theme communities in $\mathbb{T}_{\mathbf{q}, \alpha_\mathbf{q}}$ in descending order of their cohesiveness.

}

\nop{
If there is no theme community that satisfies the user query $(\mathbf{q}, \alpha_\mathbf{q})$, then $C^*_{\mathbf{q}, \alpha_\mathbf{q}}=\emptyset$ and we return $\mathcal{R}_{\mathbf{q}, \alpha_\mathbf{q}}=\emptyset$.
}

\nop{
A straight forward way to query a TC-Tree $\mathcal{T}$ by $(\mathbf{q}, \alpha_\mathbf{q})$ is to first find the tree node $n_i$ such that $\mathbf{p_i}=\mathbf{q}$, then extract all theme communities from $\mathcal{L}_\mathbf{p_i}$ by Equation~\ref{Eqn:get_mpt}.
}

\nop{
The \emph{answer} to query $(\mathbf{q}, \alpha_\mathbf{q})$ is the set of maximal $(\mathbf{p},\alpha)$-trusses with respect to $\alpha_\mathbf{q}$ for any sub-pattern of $\mathbf{q}$, that is, $\mathbb{C}_\mathbf{q}(\alpha_\mathbf{q})=\{C^*_\mathbf{p}(\alpha_\mathbf{q}) \mid C^*_\mathbf{p}(\alpha_\mathbf{q})\neq \emptyset, \mathbf{p}\subseteq \mathbf{q}\}$.
With $\mathbb{C}_\mathbf{q}(\alpha_\mathbf{q})$, one can easily extract theme communities by finding the maximal connected subgraphs in all the retrieved maximal $(\mathbf{p},\alpha)$-trusses.
}

From time to time, a user-provided query pattern $\mathbf{q}$ may not lead to any answer. The cohesiveness and size of the retrieved theme communities may be too small or the answer to the query can even be $\mathcal{R}_\mathbf{q}=\emptyset$.
In such a case, we can explore the TC-Tree to recommend a ranked list of new query patterns, denoted by $\mathcal{U}=\mathbf{q}_1, \ldots, \mathbf{q}_k$.

Those recommended query patterns in $\mathcal{U}$ should satisfy three conditions. 
First, querying $\mathcal{T}$ by any query pattern $\mathbf{q}_i\in\mathcal{U}$ leads to a non-empty set of theme communities, that is, $\mathcal{L}_{\mathbf{q}_i}\neq\emptyset$. 
Second, the recommended query patterns should be contained in $\mathbf{q}$, because querying $\mathcal{T}$ by a pattern containing $\mathbf{q}$ only retrieves those theme communities with smaller cohesiveness and size.
Third, every query pattern $\mathbf{q}_i\in\mathcal{U}$ should be the most similar to the original query pattern $\mathbf{q}$, that is, among all query patterns that satisfy the previous two conditions, $\mathbf{q}_i$ has the minimum size of set difference $|\mathbf{q}\setminus\mathbf{q}_i|$. 

According to the third condition, the set differences of $\mathbf{q}$ and every query pattern in $\mathcal{U}$ must have the same size.
Since a query pattern $\mathbf{q}$ contains at most $|\mathbf{q}|\choose m$ patterns that have the same set difference size $m$ with respect to $\mathbf{q}$, the number of query patterns in $\mathcal{U}$ is at most $|\mathbf{q}|\choose \floor*{|\mathbf{q}|/2}$. In practice, since many query patterns contained in $\mathbf{q}$ retrieve no theme community from $\mathcal{T}$, the actual volume of $\mathcal{U}$ is very small.

Now, we illustrate how to rank all the recommended query patterns in $\mathcal{U}$. For each query pattern $\mathbf{q}_i\in\mathcal{U}$, we first access the last entry of $\mathcal{L}_{\mathbf{q}_i}$ to obtain $\alpha^*_{\mathbf{q}_i}$, which is the maximum cohesiveness of all theme communities in $G_{\mathbf{q}_i}$.
Then, we rank the new query patterns in $\mathcal{U}$ in the descending order of $\alpha^*_{\mathbf{q}_i}$, so that the user can first explore the query patterns that retrieve theme communities with large cohesiveness.

\nop{
 if $\mathcal{L}_\mathbf{q}\neq\emptyset$, then we obtain $\alpha^*_\mathbf{q}=\max \mathcal{A}_\mathbf{q}$ by visiting the last entry of $\mathcal{L}_\mathbf{q}$, and recommend the user a valid query $(\mathbf{q}, \alpha^*_\mathbf{q})$ that gives the non-trivial range of cohesion threshold for the theme communities in $G_\mathbf{q}$.
If $\mathcal{L}_\mathbf{q}=\emptyset$, then
}

\nop{
In this case, we can further explore the TC-Tree to recommend the user a rank list of potential queries $\mathcal{U}=(\mathbf{q}_1, \alpha^*_{\mathbf{q}_1}), (\mathbf{q}_2, \alpha^*_{\mathbf{q}_2}), \ldots, (\mathbf{q}_k, \alpha^*_{\mathbf{q}_k})$ that satisfies the following conditions:
}

As shown in Algorithm~\ref{Alg:qr}, the query recommendation method simply traverses the TC-Tree in the breadth first manner and collects the set of new queries that satisfy the above conditions for the new query patterns in $\mathcal{U}$. 

\nop{
The efficiency of Algorithm~\ref{Alg:qtct} comes from three factors.
First, in Line 4, if $s_{n_c}\not\in \mathbf{q}$, then $\mathbf{p}_c\not\subset \mathbf{q}$ and the patterns of all descendants of $n_c$ are not sub-patterns of $\mathbf{q}$. Therefore, we can prune the entire subtree rooted at $n_c$. 
Second, in Line 6, if $C^*_{\mathbf{p}_c}(\alpha_\mathbf{q})=\emptyset$, we can prune the entire subtree rooted at $n_c$, because, according to Proposition~\ref{Prop:pam}, no descendants of $n_c$ can induce a maximal $(\mathbf{p},\alpha)$-truss with respect to $\alpha_\mathbf{q}$.
Last, in Line 5, getting $C^*_{\mathbf{p}_c}(\alpha_\mathbf{q})$ from $\mathcal{L}_{\mathbf{p}_c}$ is efficient using Equation~\ref{Eqn:get_mpt}.
}

\nop{
then $C^*_{p_{g}}(\alpha_\mathbf{q})=\emptyset$ for all children $n_{g}$ of $n_c$. This is an immediate result of the second property of Proposition~\ref{Prop:pam}, since $p_c\subset p_g$. Therefore, when $C^*_{p_c}(\alpha_\mathbf{q})=\emptyset$, . 
}

In summary, TC-Tree enables fast user query answering and efficient query recommendation. 
As demonstrated by the case study in Section~\ref{Sec:cs} and the experiments in Section~\ref{Sec:eotci}, TC-Tree is efficient to build, easy to query, and scales well to index a large number of theme communities using practical size of memory.

\section{Experiments}
\label{sec:exp}

In this section, we first present a case study to demonstrate how theme community finding is useful.  Then, we comprehensively evaluate the performance of Theme Community Scanner (TCS), Theme Community Finder Apriori (TCFA), Theme community Finder Intersection (TCFI) and Theme Community Tree (TC-Tree). 
Last, we compare the theme community detection performance of TCFI and that of two vertex attributed network methods, CESNA~\cite{yang2013community} and SCI~\cite{wang2016semantic}.

\nop{\todo{What does this mean?} \mc{All these methods are compatible with both the absolute frequency and the relative frequency of patterns.}}

For each of our proposed methods, the final output is the set of theme communities obtained by finding the maximal connected subgraphs in the detected maximal $(\mathbf{p}, \alpha)$-trusses.
We use the absolute frequency for TC-Tree.
However, since the absolute frequency is unnormalized, which makes it difficult to set the frequency threshold $\epsilon$ for TCS, we adopt the relative frequency for TCS, and also use the relative frequency for TCFA and TCFI to fairly compare with TCS.
Since TC-Tree is an indexing method, it is not directly comparable with TCS, TCFA and TCFI.

We implement TCS, TCFA and TCFI in Java. 
In order to efficiently index the theme communities in large DBNs, we implement TC-Tree in C++ and parallelize the steps in Lines 2-5 of Algorithm~\ref{Alg:tctb} with 4 threads using OpenMP.
The source code of CESNA and SCI was provided by their authors.
All experiments are performed on a Windows 7 PC with Core-i7 CPU, 32GB RAM and a 5400 rpm hard drive.

\nop{
We compare the performance of TCS, TCFA and TCFI in Sections~\ref{Sec:eop} and~\ref{Sec:eotcf}, and evaluate the indexing performance of TC-Tree in Section~\ref{Sec:eotci}. Last, we present some interesting case studies in Section~\ref{Sec:cs}.
}

\nop{
The performance of TCS, TCFA and TCFI are evaluated on the following aspects.  First, ``Time Cost'' measures the total runtime of each method.
Second, the ``Number of Theme Communities (NTC)'' is the total number of all theme communities detected in $G$.
The evaluation metrics of TC-Tree are discussed in Section~\ref{Sec:eotci}
}
\nop{
Second, ``Number of Patterns (NP)'', ``Number of Vertices (NV)'' and ``Number of Edges (NE)'' are the total numbers of patterns, vertices and edges in all detected maximal $(\mathbf{p},\alpha)$-trusses, respectively. 
When counting NV, a vertex is counted $k$ times if it is contained by $k$ different maximal $(\mathbf{p},\alpha)$-trusses. 
For NE, an edge is counted $k$ times if it is contained by $k$ different maximal $(\mathbf{p},\alpha)$-trusses.
NP is equal to the number of maximal $(\mathbf{p},\alpha)$-trusses, since each maximal $(\mathbf{p},\alpha)$-truss uniquely corresponds to a pattern.
The evaluation metrics of TC-Tree are discussed in Section~\ref{Sec:eotci}
}

The following 3 data sets are used.

The \textbf{Brightkite (BK)}
data set is a public check-in data set produced by the location-based social networking website \url{BrightKite.com}~\cite{BKGW_data}. 
It includes a friendship network of 58,228 users and 4,491,143 user check-ins; every user check-in contains the check-in time and location.
We construct a DBN using this data set by taking the user friendship network as the network of the DBN. 
To create the vertex database for a user, we treat each check-in location as an item, and cut the check-in history of a user into periods of 2 days.  The set of check-in locations within a period is a transaction.
A theme community in this DBN is a group of friends who frequently visit the same set of places.

The \textbf{Gowalla (GW)} 
data set is a public data set produced by the location-based social networking website \url{Gowalla.com}~\cite{BKGW_data}. 
It includes a friendship network of 196,591 users and 6,442,890 user check-ins that contain the check-in time and location. 
We transform this data set into a DBN in the same way as BK.  

The \textbf{AMINER}
data set is built from the Citation network v2 (CNV2) data set~\cite{Aminer_data}.
CNV2 contains 1,397,240 papers. We transform it into a DBN in the following two steps. First, we treat each author as a vertex and build an edge between a pair of authors who co-author at least one paper. 
Second, to build the vertex database for an author, we treat each keyword in the abstract of a paper as an item, and all the keywords in the abstract of a paper are turned into a transaction. An author vertex is associated with a vertex database of all papers the author publishes.
In this DBN, a theme community represents a group of authors who collaborate closely and share the same research interest that is described by the same set of keywords.


\begin{algorithm}[t]
\caption{Build Theme Community Tree}
\label{Alg:tctb}
\KwIn{A DBN $G$.}
\KwOut{The TC-Tree $\mathcal{T}$ with root node $n_0$.}
\BlankLine
\begin{algorithmic}[1]
    \STATE Initialization: $Q\leftarrow\emptyset$, $s_{n_0}\leftarrow\emptyset$, $\mathcal{L}_{\mathbf{p}_0}\leftarrow\emptyset$.
    
    \FOR{each item $s_i\in S$}
    	\STATE $s_{n_i}\leftarrow s_i$, $\mathbf{p}_i\leftarrow s_i$ and compute $\mathcal{L}_{\mathbf{p}_i}$.
	\STATE \textbf{if} $\mathcal{L}_{\mathbf{p}_i}\neq \emptyset$ \textbf{then} $n_0.addChild(n_i)$ and $Q.push(n_i)$. 
    \ENDFOR
    
    \WHILE{$Q\neq\emptyset$}
    	\STATE $n_{f}\leftarrow Q.pop()$.
	\FOR{each node $n_b\in n_f.siblings$}
		\STATE \textbf{if} $s_{n_f} \prec s_{n_b}$ \textbf{then} $s_{n_c} \leftarrow s_{n_b}$, $\mathbf{p}_c \leftarrow \mathbf{p}_f \cup \mathbf{p}_b$, and compute $\mathcal{L}_{\mathbf{p}_c}$.
		\STATE \textbf{if} $\mathcal{L}_{\mathbf{p}_c}\neq \emptyset$ \textbf{then} $n_f.addChild(n_c)$ and $Q.push(n_c)$.
	\ENDFOR
    \ENDWHILE
\BlankLine
\RETURN The TC-Tree $\mathcal{T}$ with root node $n_0$.
\end{algorithmic}
\end{algorithm}

\begin{algorithm}[t]
\caption{Query Recommendation}
\label{Alg:qr}
\KwIn{A TC-Tree $\mathcal{T}$ and a query $\mathbf{q}$.}
\KwOut{A ranked list $\mathcal{U}=\mathbf{q}_1, \ldots, \mathbf{q}_k$.}
\BlankLine
\begin{algorithmic}[1]
    \STATE Initialization: $Q\leftarrow n_0$, $\mathcal{U}\leftarrow\emptyset$.
    \WHILE{$Q\neq\emptyset$}
    	\STATE $n_{f}\leftarrow Q.pop()$.
	\FOR{each node $n_c\in n_f.children \land s_{n_c}\in \mathbf{q}$}
			\STATE \textbf{if} $\mathcal{L}_{\mathbf{p}_c}\neq \emptyset$ \textbf{then} $\mathcal{U}\leftarrow \mathcal{U} \cup \mathbf{p}_c$; $Q.push(n_c)$.
	\ENDFOR
    \ENDWHILE
    \STATE Keep each query pattern $\mathbf{q}_i\in\mathcal{U}$ that has the smallest size of set difference $|\mathbf{q}\setminus\mathbf{q}_i|$, and remove the others.
    \STATE Rank all patterns in $\mathcal{U}$ by $\alpha^*_{\mathbf{q}_i}$.
\BlankLine
\RETURN The ranked list $\mathcal{U}$.
\end{algorithmic}
\end{algorithm}

\newcommand{\cellwidth}{13mm}
\begin{table}[t]
\caption{Statistics of the DBNs. \#Trans. is the number of transactions. \#Items is the number of items stored in all vertex databases. \nop{ \#Items (total) is the total number of items stored in all vertex databases, and \#Items (unique) is the number of unique items in $S$.}
}
\label{Table:dss}
\vspace{1mm}
\centering
\begin{tabular}{|p{14mm}<{\centering}|p{\cellwidth}<{\centering}|p{\cellwidth}<{\centering}|p{\cellwidth}<{\centering}|p{\cellwidth}<{\centering}|}
\hline
Data sets 		& \#Vertices 			& \#Edges 			& \#Trans. 			& \#Items 				\\ \hline
BK			& $5.1{\times}10^4$		& $2.1{\times}10^5$		& $1.2{\times}10^6$		& $1.7{\times}10^6$		\\ \hline
GW			&  $1.1{\times}10^5$		& $9.5{\times}10^5$		& $2.0{\times}10^6$		& $3.5{\times}10^6$		\\ \hline
AMINER		& $1.1{\times}10^6$		& $2.6{\times}10^6$		& $3.1{\times}10^6$		& $9.2{\times}10^6$		\\ \hline
\end{tabular}
\end{table}

\begin{figure}[t]
\centering
\includegraphics[width=85mm]{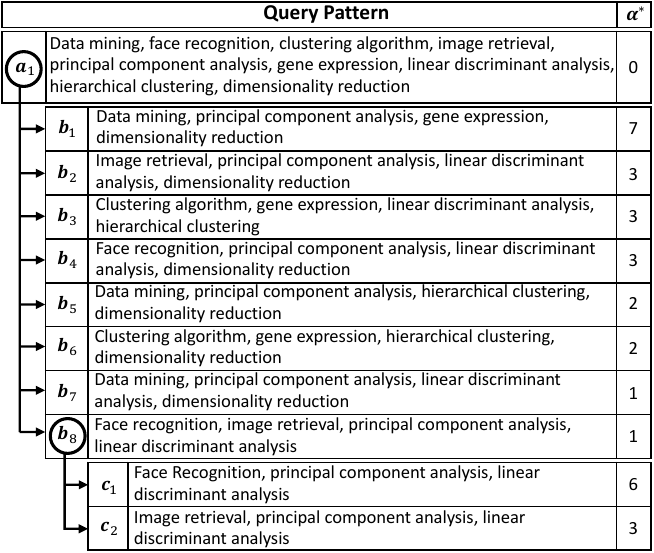}

\vspace{-3mm}
\caption{The query patterns for the theme communities in Figure~\ref{Fig:case_study_b} and Figure~\ref{Fig:case_study_c}. $\mathbf{a}_1$ is the user provided query pattern. $\{\mathbf{b}_1, \ldots, \mathbf{b}_8\}$ are the recommended query patterns based on $\mathbf{a}_1$. $\{\mathbf{c}_1, \mathbf{c}_2\}$ are the top-1 and top-2 recommended query patterns for $\mathbf{b}_8$, respectively. $\alpha^*$ is the maximum cohesiveness of all theme communities induced by a query pattern.}
\label{Fig:query_patterns}
\end{figure}

\begin{figure}[t]
\centering
\includegraphics[width=85mm]{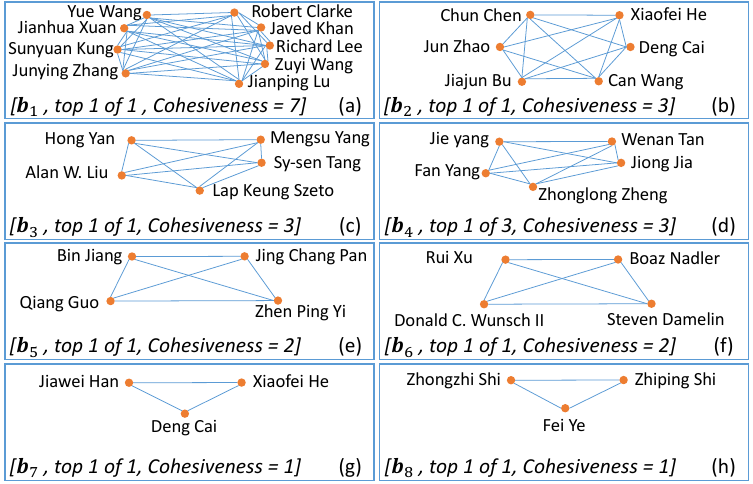}

\vspace{-3mm}
\caption{The top-1 theme communities retrieved by the query patterns $\{\mathbf{b}_1, \ldots, \mathbf{b}_8\}$ in Figure~\ref{Fig:query_patterns}.}
\label{Fig:case_study_b}
\end{figure}

\begin{figure}[t]
\centering
\includegraphics[width=85mm]{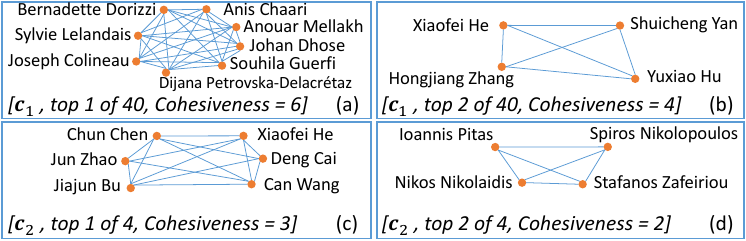}

\vspace{-3mm}
\caption{The top-1 and top-2 theme communities retrieved by the query patterns $\mathbf{c}_1$ and $\mathbf{c}_2$ in Figure~\ref{Fig:query_patterns}.}
\label{Fig:case_study_c}
\end{figure}

The statistics of all data sets are given in Table~\ref{Table:dss}. 

\nop{
\begin{table}[t]
\caption{Statistics of the DBNs. \#Trans. is the number of transactions. \#Items is the number of items stored in all vertex databases. \nop{ \#Items (total) is the total number of items stored in all vertex databases, and \#Items (unique) is the number of unique items in $S$.}}
\centering
\label{Table:dss}
\begin{tabular}{|c|c|c|c|c|}
\hline
  & BK & GW & AMINER & SYN \\ \hline
\#Vertices            &  $5.1{\times}10^4$  &  $1.1{\times}10^5$  &  $1.1{\times}10^6$  &  $1.0{\times}10^6$     \\ \hline
\#Edges              &  $2.1{\times}10^5$  &  $9.5{\times}10^5$  &  $2.6{\times}10^6$  &  $1.0{\times}10^7$     \\ \hline
\#Trans.   &  $1.2{\times}10^6$  &  $2.0{\times}10^6$  &  $3.1{\times}10^6$  &  $6.1{\times}10^6$     \\ \hline
\#Items                &  $1.7{\times}10^6$  &  $3.5{\times}10^6$  &  $9.2{\times}10^6$  &  $1.3{\times}10^8$     \\ \hline
\end{tabular}
\end{table}
}

\subsection{A Case Study}
\label{Sec:cs}
In this subsection, we demonstrate the query recommendation and query answering process of the TC-Tree that indexes all the theme communities in the DBN of AMINER.
Each theme community retrieved by a query pattern represents a group of co-working scholars who share the same research interest characterized by the set of keywords contained in the query pattern.

We start our case study with a user who wants to find theme communities that apply the classic methods, such as ``hierarchical clustering'' (HC), ``principal component analysis'' (PCA), ``linear discriminant analysis'' (LDA) and ``dimensionality reduction'' (DR), in the research areas of ``data mining'' (DM), ``face recognition'' (FR), ``image retrieval'' (IR) and ``gene expression'' (GE). 

Since the user may have little prior knowledge about the theme communities in the DBN, it is difficult for him to design an appropriate query that successfully retrieves valid theme communities.
Therefore, he simply put all his interested key words together and starts querying the TC-Tree using the query pattern $\mathbf{a}_1$ in Figure~\ref{Fig:query_patterns}. 

Unfortunately, no theme community in the DBN of AMINER is induced by $\mathbf{a}_1$, therefore querying the TC-Tree by $\mathbf{a}_1$ returns no theme community.
In this case, we use the query recommendation method in Algorithm~\ref{Alg:qr} to recommend to the user a ranked list of query patterns $\{\mathbf{b}_1, \ldots, \mathbf{b}_8\}$ in Figure~\ref{Fig:query_patterns}. 
Each of $\{\mathbf{b}_1, \ldots, \mathbf{b}_8\}$ is a meaningful combination of the user-interested keywords that describes an interesting research direction. For example, $\mathbf{b}_1$ characterizes the research direction that applies PCA and DR in the area of DM and GE, and $\mathbf{b}_2$ is the research direction that applies PCA, LDA and DR in the area of IR. 


Figure~\ref{Fig:case_study_b} shows the top-1 theme communities retrieved by the recommended query patterns in $\{\mathbf{b}_1, \ldots, \mathbf{b}_8\}$. Most of the top-1 theme communities are cliques, because a clique is a special case of $(\mathbf{p}, \alpha)$-truss, and the dense connection between the vertices of a clique produces a high cohesiveness.
Interestingly, the theme community in Figure~\ref{Fig:case_study_b}(b) is not a clique. This demonstrates the flexibility of the proposed $(\mathbf{p}, \alpha)$-truss in detecting communities that are not cliques but still have large cohesiveness. 

We can also see that the theme communities in Figures~\ref{Fig:case_study_b}(b) and~\ref{Fig:case_study_b}(g) both contain ``Xiaofei He'' and ``Deng Cai''. 
This shows that the proposed theme community finding allows arbitrary overlap between communities of different themes.

As shown in Figure~\ref{Fig:case_study_b}(h), the query pattern $\mathbf{b}_8$ retrieves only one theme community that is small in both size and cohesiveness. This is because FR and IR are separate sub-disciplines in the research field of computer vision, and it is rare for a paper to focus on both of them at the same time.

If the user is not happy with the small theme community retrieved by $\mathbf{b}_8$, we can further apply Algorithm~\ref{Alg:qr} to recommend to the user the query patterns $\mathbf{c}_1$ and $\mathbf{c}_2$ in Figure~\ref{Fig:query_patterns}. Obviously, $\mathbf{c}_1$ describes the sub-discipline that applies PCA and LDA in FR, and $\mathbf{c}_2$ characterizes the sub-discipline that uses PCA and LDA in IR.

Figure~\ref{Fig:case_study_c} shows the theme communities retrieved by $\mathbf{c}_1$ and $\mathbf{c}_2$. Since $\mathbf{c}_1$ and $\mathbf{c}_2$ characterize the two sub-disciplines more precisely than $\mathbf{b}_8$, $\mathbf{c}_1$ and $\mathbf{c}_2$ retrieve more theme communities than $\mathbf{b}_8$, and the cohesiveness of the theme communities retrieved by $\mathbf{c}_1$ and $\mathbf{c}_2$ are much larger than the one retrieved by $\mathbf{b}_8$.

In summary, the proposed theme community finding discovers meaningful theme communities and allows arbitrary overlap between theme communities. The proposed query recommendation method effectively recommends meaningful query patterns, and the proposed TC-Tree makes it possible to efficiently retrieve ranked lists theme communities from large DBNs.

\subsection{Effect of Parameters}
\label{Sec:eop}
In this subsection, we analyze the effect of the cohesion threshold $\alpha$ and the frequency threshold $\epsilon$.
The settings of parameters are $\alpha \in \{0.0, 0.1, 0.2, 0.3, 0.5, 1.0, 1.5, 2.0\}$ and $\epsilon\in\{0.1, 0.2, 0.3\}$. 

We do not evaluate the performance of TCS for $\epsilon=0.0$ and $\epsilon > 0.3$, because TCS cannot stop in reasonable time when $\epsilon=0.0$, and it loses too much accuracy when $\epsilon>0.3$.
Since TCS with $\epsilon\in\{0.1, 0.2, 0.3\}$ still run too slow on the original DBNs of BK, GW and AMINER, we use smaller DBNs that are sampled from the original DBNs by doing a breadth first search from a random seed vertex.
From each of BK and GW, we sample one DBN with 10,000 edges. 
For AMINER, we sample a DBN of 5,000 edges.
\nop{\todo{Do you conduct the sampling multiple times and report the average results or just try one sample?}}

Figures~\ref{Fig:effect_of_parameters}(a), \ref{Fig:effect_of_parameters}(c) and \ref{Fig:effect_of_parameters}(e) show the \textbf{number of theme communities}, denoted by ``\#TCs'', found by TCFA, TCFI and TCS from BK, GW and AMINER, respectively. 
As it is shown, TCFA and TCFI detect the same number of theme communities for all values of $\alpha$ on all DBNs. This is because both TCFA and TCFI produce the exact result in finding all theme communities in a DBN.
However, TCS does not always produce the exact result. 
The reason is that vertices with small pattern frequencies can still form a good theme community with large edge cohesion if they form a densely connected subgraph.
Such theme communities may be missed if the patterns with low frequencies are dropped by the pre-filtering step of TCS.

Figures~\ref{Fig:effect_of_parameters}(b), \ref{Fig:effect_of_parameters}(d) and~\ref{Fig:effect_of_parameters}(f) show the time cost of TCFA, TCFI and TCS on BK, GW and AMINER, respectively.
The time cost of both TCFA and TCFI decreases when $\alpha$ increases, because increasing $\alpha$ reduces the size of $\mathcal{P}^{k-1}$, which further reduces the size of $\mathcal{M}^k$.

When $\alpha$ is small, the time cost of TCFI is much lower than the time cost of TCFA.
This is due to the following two effects of applying the graph intersection property in Proposition~\ref{Lem:gip}.
First, most maximal $(\mathbf{p},\alpha)$-trusses are small local subgraphs that do not intersect with each other, thus many unqualified patterns in $\mathcal{M}^k$ are pruned by TCFI using the graph intersection property.
Second, for each call of MTD, TCFA computes the maximal $(\mathbf{p},\alpha)$-truss in the large theme network induced from the entire DBN, however, TCFI operates on the small theme network induced from the intersection of two maximal $(\mathbf{p},\alpha)$-trusses.

When $\alpha$ is large, the time cost of TCFI and TCFA becomes comparable.
This is because increasing $\alpha$ reduces the size of $\mathcal{M}^k$ and the size of the maximal $(\mathbf{p}, \alpha)$-truss for each pattern in $\mathcal{M}^k$. Therefore, the effectiveness of applying the graph intersection property is reduced.

Take the sampled DBN of AMINER as an example. When $\alpha=0$, TCFA calls MTD 622,852 times and TCFI calls MTD 152,396 times. TCFI effectively prunes 75.5\% of the candidate patterns used by TCFA. 
Moreover, as shown in Figure~\ref{Fig:effect_of_parameters}(f), TCFI is nearly 3 orders of magnitude faster than TCFA when $\alpha=0$. This is because TCFI always operates on small theme networks.
We can also see that, when $\alpha\geq 1$, the time cost of TCFI and TCFA becomes comparable. 
This is because, when $\alpha\geq 1$, AMINER contains only one maximal $(\mathbf{p},\alpha)$-truss, thus TCFI cannot prune any candidate pattern in $\mathcal{M}^k$ or induce any smaller theme network by applying the graph intersection property.

The time cost of TCS is not affected by $\alpha$, because it is largely dominated by the size of $\mathcal{P}$ (see Section~\ref{sec:tcs}), which is irrelevant to $\alpha$. 
Increasing $\epsilon$ reduces the size of $\mathcal{P}$ and improves the efficiency of TCS. However, since the size of $\mathcal{P}$ is still too large, TCS is orders of magnitude slower than TCFI on all DBNs.

In summary, TCFI produces the best detection results of theme communities and achieves the best efficiency performance for all values of $\alpha$ on all DBNs. 

\nop{
\newcommand{\parawidthExpthree}{43mm}
\begin{figure*}[t]
\centering
\subfigure[Query by alpha (BK)]{\includegraphics[width=\parawidthExpthree]{Figs/EXP3_IndexingPerformances_qba_BK.pdf}}
\subfigure[Query by alpha (GW)]{\includegraphics[width=\parawidthExpthree]{Figs/EXP3_IndexingPerformances_qba_GW.pdf}}
\subfigure[Query by alpha (AMINER)]{\includegraphics[width=\parawidthExpthree]{Figs/EXP3_IndexingPerformances_qba_DBLP.pdf}}
\subfigure[Query by alpha (SYN)]{\includegraphics[width=\parawidthExpthree]{Figs/EXP3_IndexingPerformances_qba_SYN.pdf}}
\subfigure[Query by pattern (BK)]{\includegraphics[width=\parawidthExpthree]{Figs/EXP3_IndexingPerformances_qbp_BK.pdf}}
\subfigure[Query by pattern (GW)]{\includegraphics[width=\parawidthExpthree]{Figs/EXP3_IndexingPerformances_qbp_GW.pdf}}
\subfigure[Query by pattern (AMINER)]{\includegraphics[width=\parawidthExpthree]{Figs/EXP3_IndexingPerformances_qbp_DBLP.pdf}}
\subfigure[Query by pattern (SYN)]{\includegraphics[width=\parawidthExpthree]{Figs/EXP3_IndexingPerformances_qbp_SYN.pdf}}
\caption{Querying TC-Tree. (a)-(d) show the QBA performance. (e)-(h) show the QBP performance.}
\label{Fig:query_performances}
\end{figure*}
}

\subsection{Scalability of Theme Community Finding}
\label{Sec:eotcf}
In this subsection, we analyze the runtime of all methods with respect to the size of the DBNs. 
For each DBN, we generate a series of DBNs with different sizes by sampling the original DBN using the sampling method introduced in Section~\ref{Sec:eop}.
Since TCS and TCFA run too slow on large DBNs, we stop reporting the performance of TCS and TCFA when they cost more than one day. 
The performance of TCFI is evaluated on all sizes of DBNs including the original ones. To evaluate the worst case performance of all methods, we set $\alpha=0$.

\newcommand{\parawidthExpone}{41.5mm}
\begin{figure}[t]
\centering
\subfigure{\includegraphics[width=\parawidthExpone]{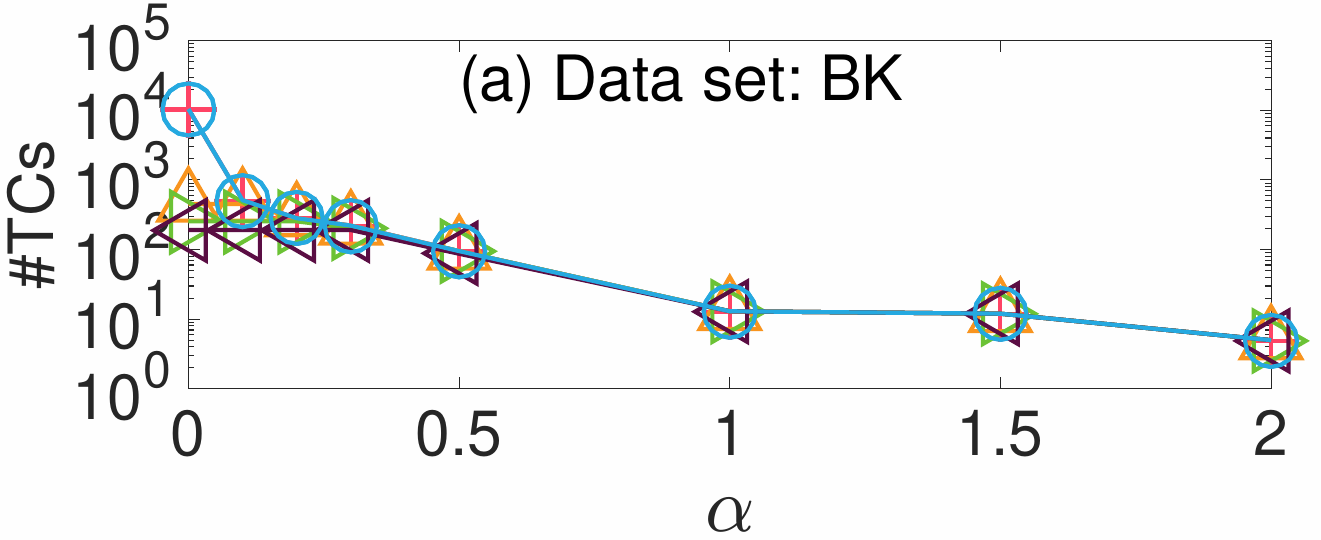}}
\subfigure{\includegraphics[width=\parawidthExpone]{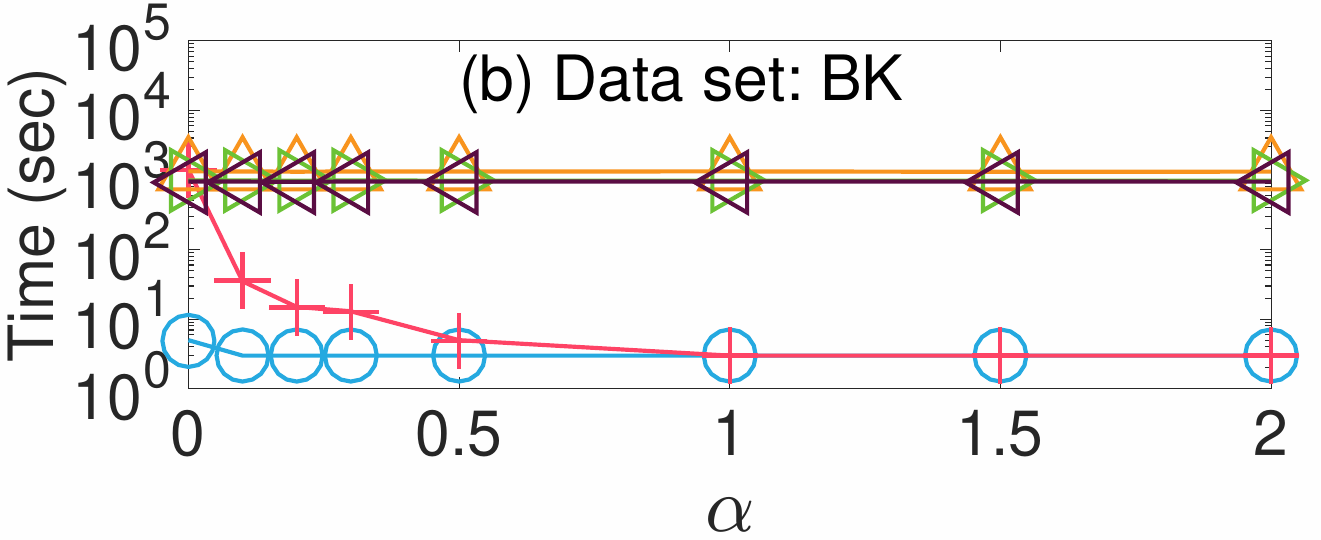}}

\vspace{-3mm}
\subfigure{\includegraphics[width=\parawidthExpone]{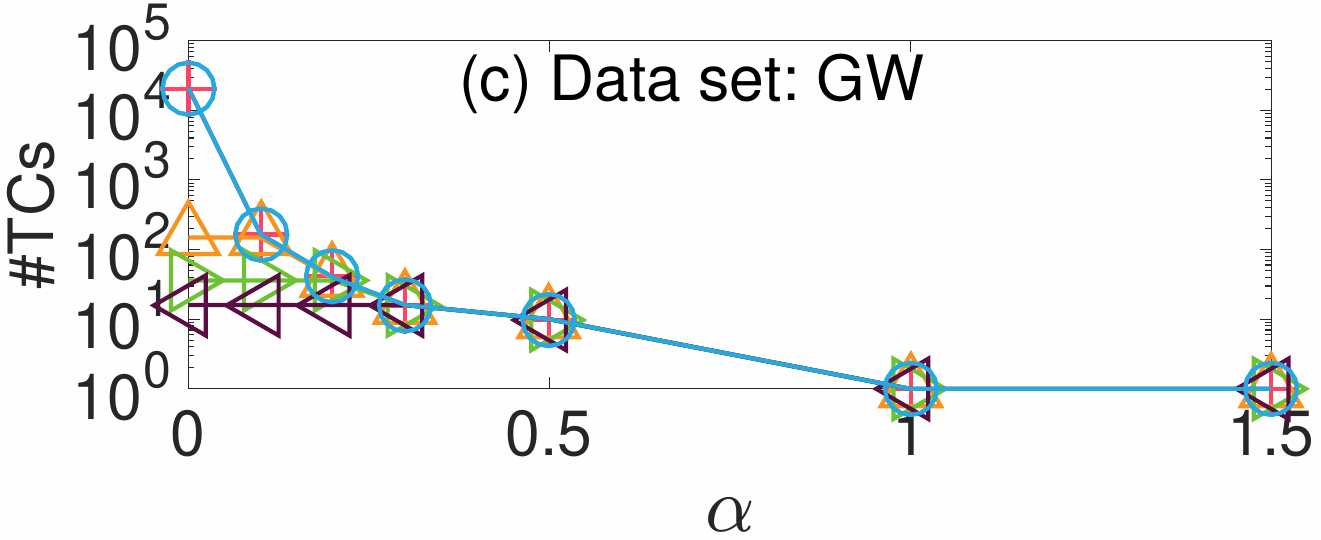}}
\subfigure{\includegraphics[width=\parawidthExpone]{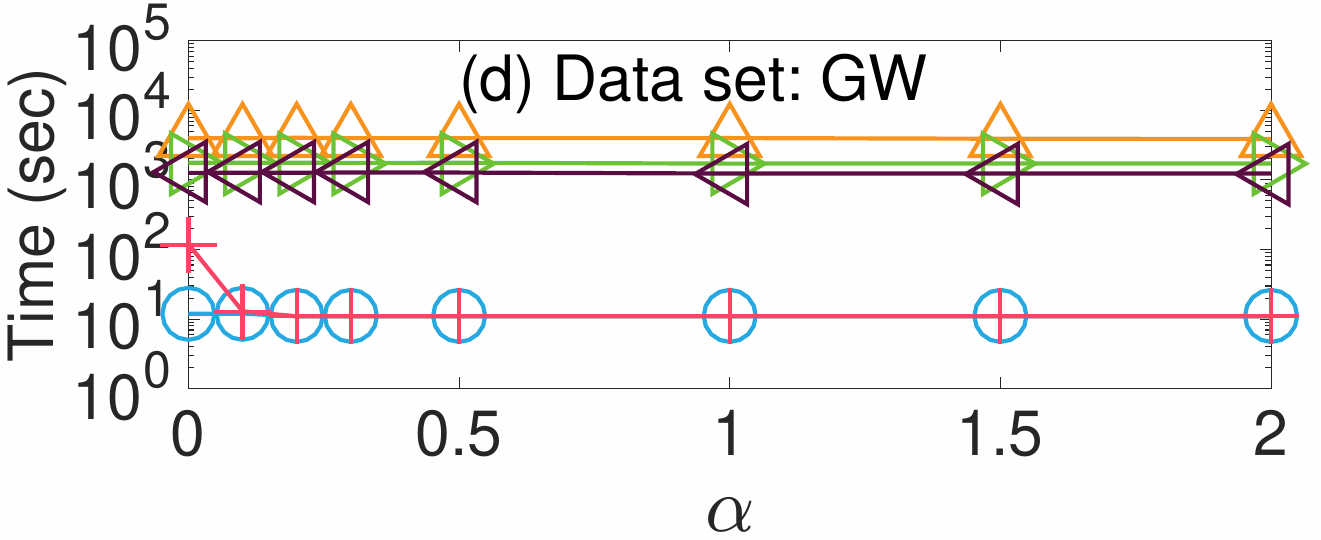}}

\vspace{-3mm}
\subfigure{\includegraphics[width=\parawidthExpone]{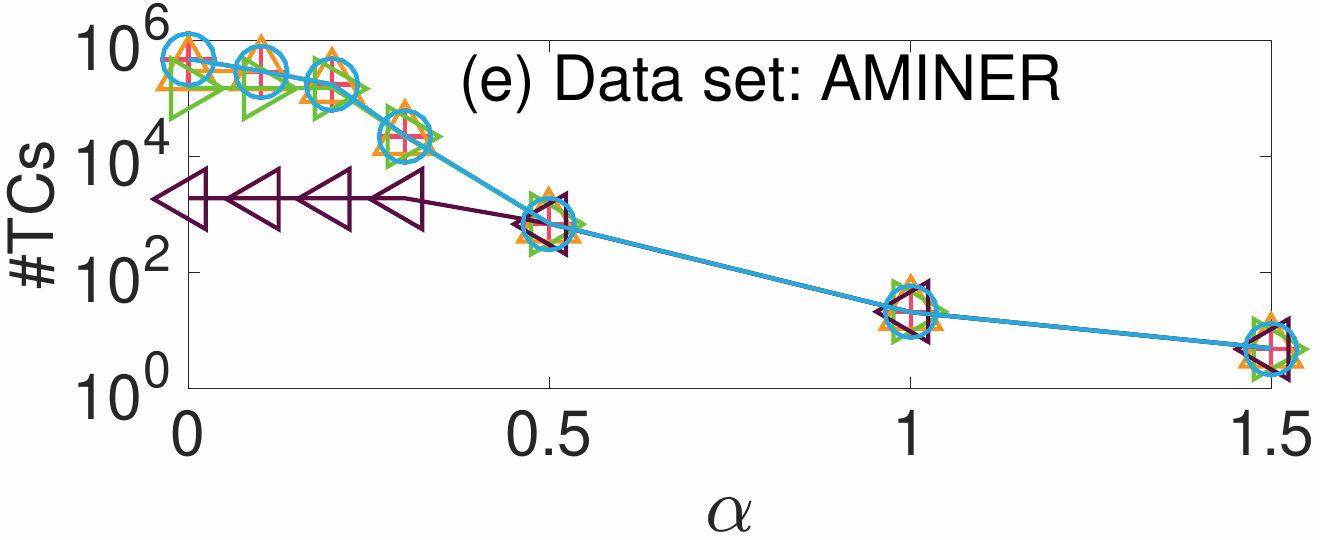}}
\subfigure{\includegraphics[width=\parawidthExpone]{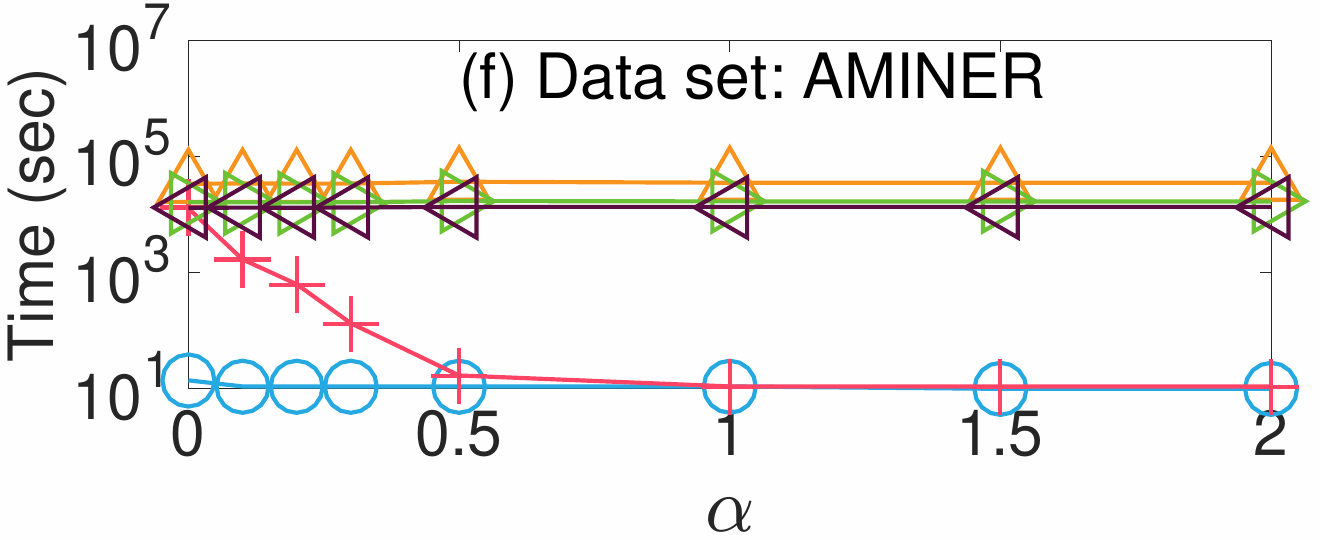}}

\vspace{-3mm}

\subfigure{\includegraphics[height=2.8mm]{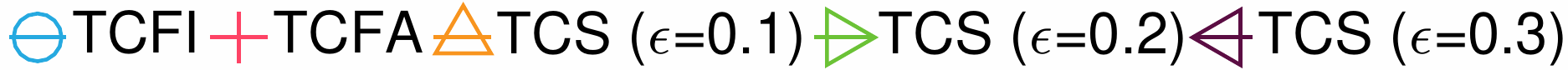}}

\vspace{-3mm}
\caption{The effects of parameters $\alpha$ and $\epsilon$. In (c) and (e), \#TCs is zero when $\alpha=2.0$.}
\label{Fig:effect_of_parameters}
\end{figure}

\newcommand{\parawidthExptwo}{41.5mm}
\begin{figure}[t]
\centering
\subfigure{\includegraphics[width=\parawidthExptwo]{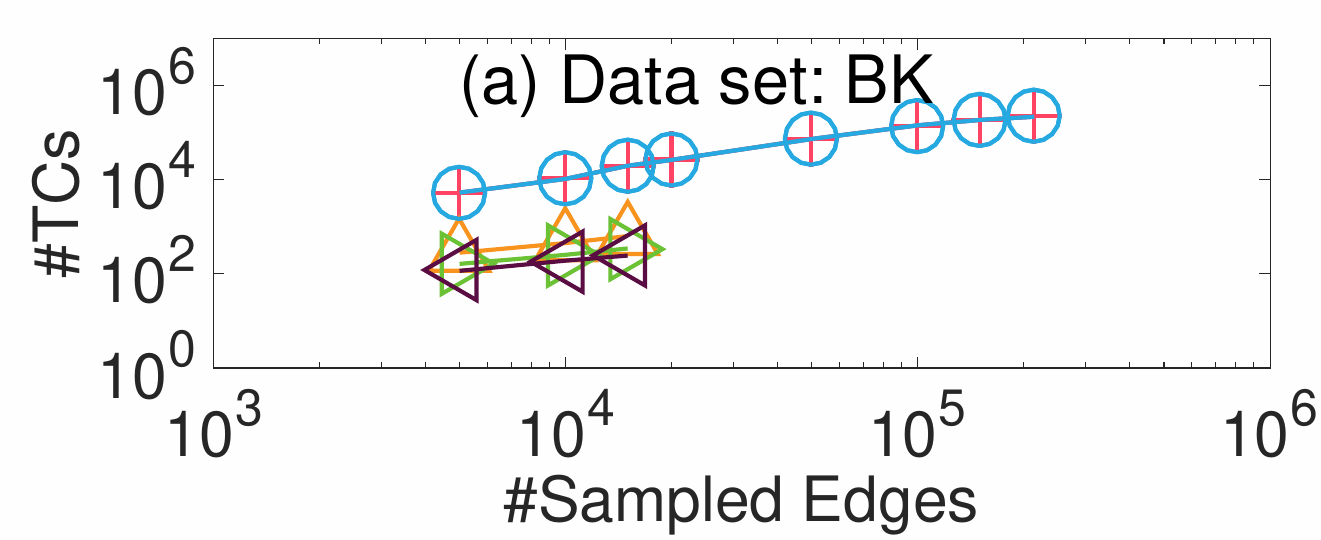}}
\subfigure{\includegraphics[width=\parawidthExptwo]{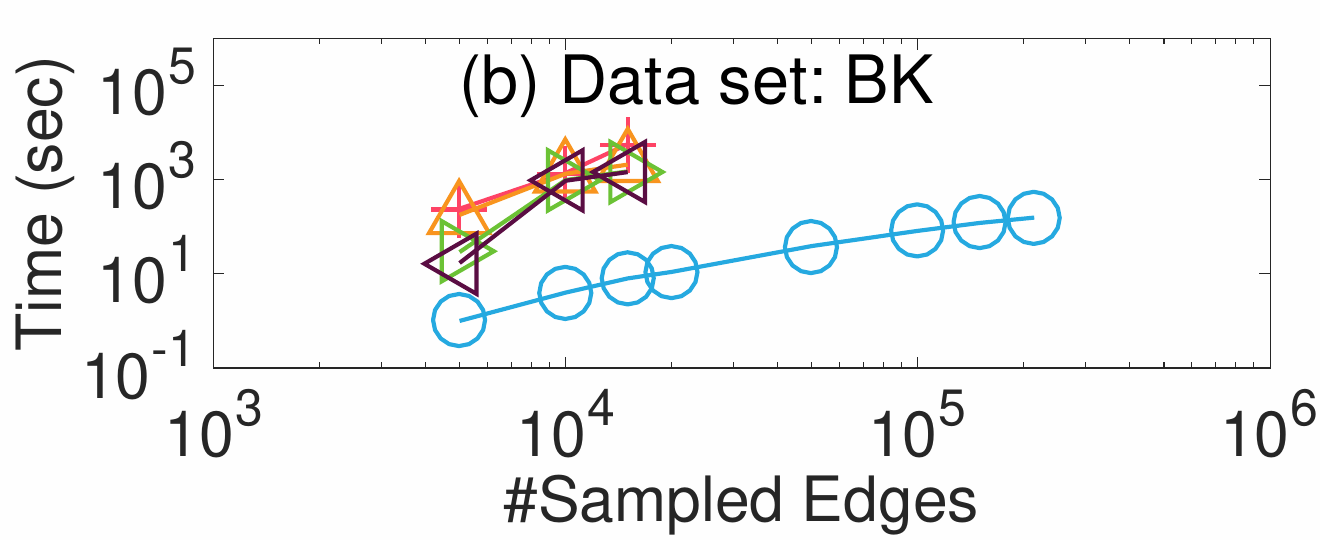}}

\vspace{-3mm}
\subfigure{\includegraphics[width=\parawidthExptwo]{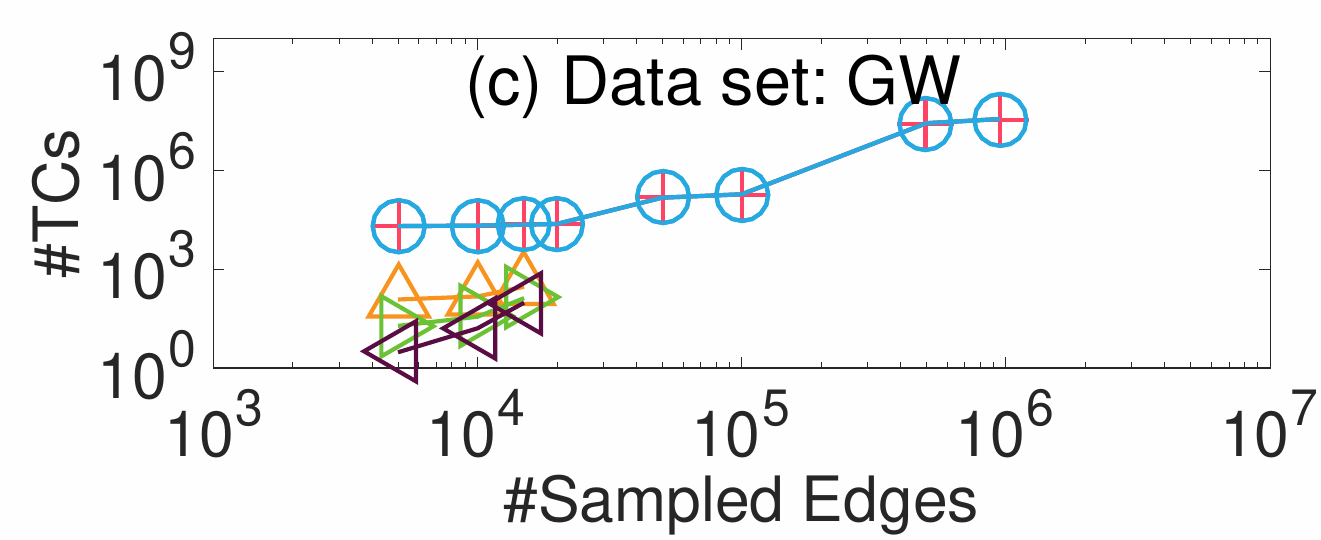}}
\subfigure{\includegraphics[width=\parawidthExptwo]{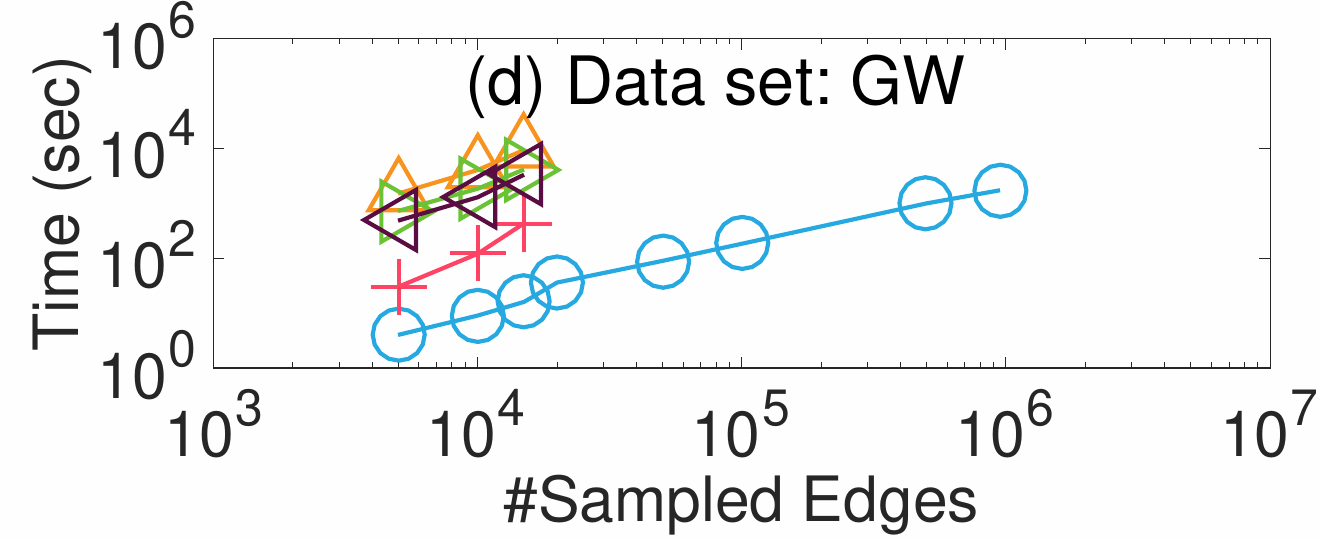}}

\vspace{-3mm}
\subfigure{\includegraphics[width=\parawidthExptwo]{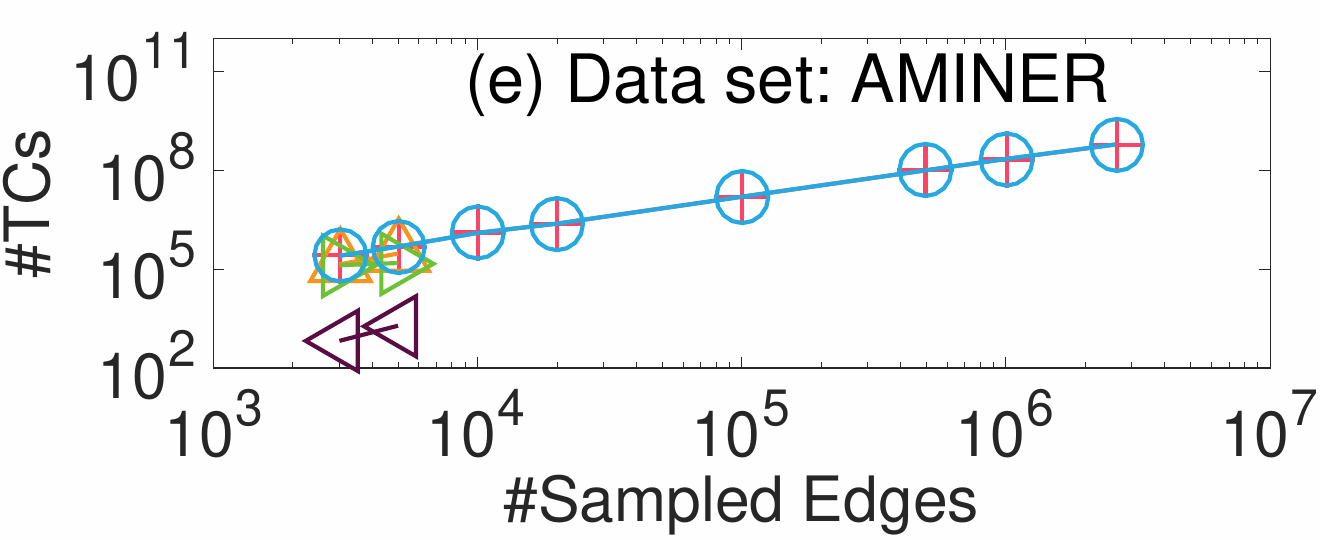}}
\subfigure{\includegraphics[width=\parawidthExptwo]{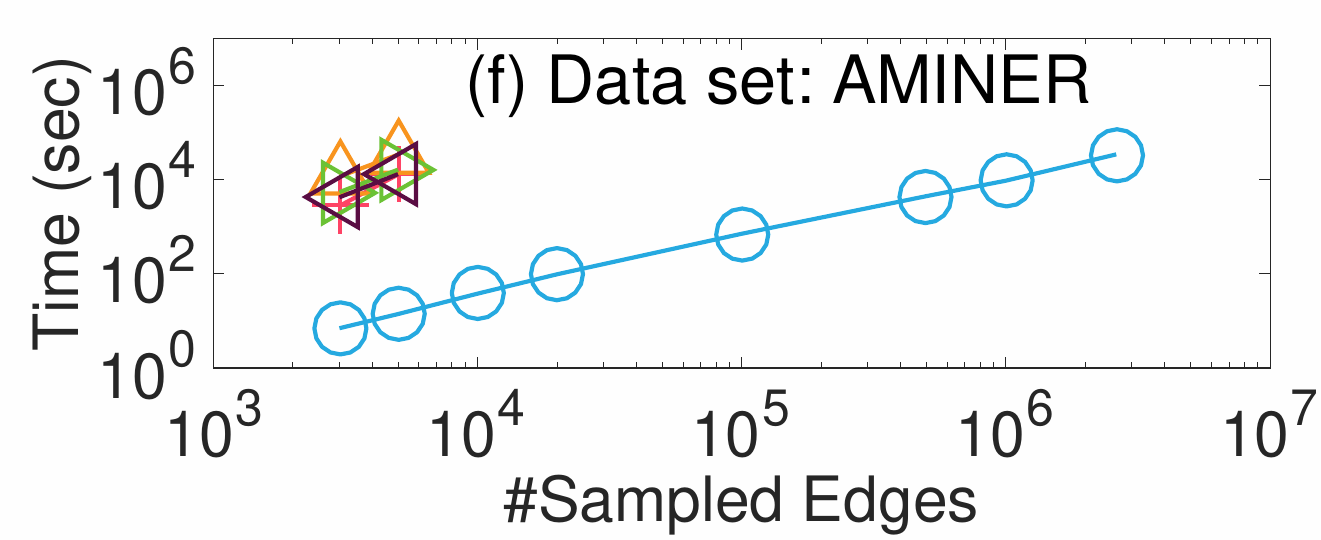}}

\vspace{-3mm}

\subfigure{\includegraphics[height=2.8mm]{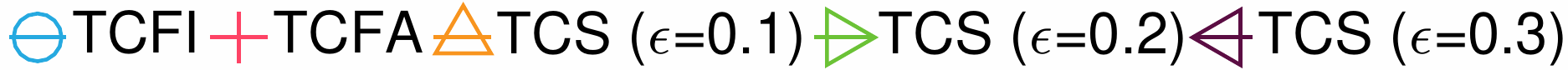}}

\vspace{-3mm}
\caption{The \#TCs and the time cost of TCS, TCFA and TCFI on different sizes of networks.}
\label{Fig:scalability_results}
\end{figure}

Figures~\ref{Fig:scalability_results}(a),~\ref{Fig:scalability_results}(c) and~\ref{Fig:scalability_results}(e) show the number of theme communities found by TCFA, TCFI and TCS from BK, GW and AMINER, respectively. 
When the number of sampled edges increases, the numbers of theme communities reported by all methods increase. Increasing the size of the DBN increases the number of maximal $(\mathbf{p},\alpha)$-trusses, which further increases the number of theme communities.
Both TCFI and TCFA consistently produce the exact results. Due to the accuracy loss caused by pre-filtering the patterns with low frequencies, TCS cannot produce the same results as TCFI and TCFA.

Figures~\ref{Fig:scalability_results}(b),~\ref{Fig:scalability_results}(d) and \ref{Fig:scalability_results}(f) show the performance of time cost.
When the number of sampled edges increases, the time cost of all methods increases, because increasing the size of the DBN increases the number of theme communities. 
The time cost of TCS and TCFA grows much faster than that of TCFI.
This is because TCS generates a large number of unqualified candidate patterns by enumerating the patterns of all vertex databases, and TCFA also generates many unqualified candidate patterns by taking the pairwise unions of the patterns of the detected maximal $(\mathbf{p},\alpha)$-trusses.
By applying the graph intersection property, TCFI efficiently generates a substantially smaller number of candidate patterns by the pairwise unions of the patterns of two intersecting maximal $(\mathbf{p},\alpha)$-trusses, and only runs MTD on the small intersection of two maximal $(\mathbf{p},\alpha)$-trusses.
This significantly reduces the time cost.
As a result, TCFI achieves the best scalability and is more than two orders of magnitude faster than TCS and TCFA on large DBNs.

\begin{table}[t]
\caption{The performance of indexing of TC-Tree on the full data sets of BK, GW and AMINER.}
\label{Table:iptct}
\vspace{0.5mm}
\centering
\begin{tabular}{| c | c | c | c | }
\hline
  	    & Indexing Time           & Memory            &   \#Nodes	  		\\ \hline
BK        &  215 \;\;\;\, seconds   &  0.35 \;\,GB        &   18,581  		         \\ \hline
GW       &  1,812 \;\;seconds    &  2.66 \;\,GB         &  11,750,761 	                 \\ \hline
AMINER    &  41,958 seconds      &  22.84 GB          &   122,337,700  	         \\ \hline
\end{tabular}
\end{table}

\newcommand{\parawidthExpfour}{41.5mm}
\begin{figure}[t]
\centering
\subfigure{\includegraphics[width=\parawidthExpfour]{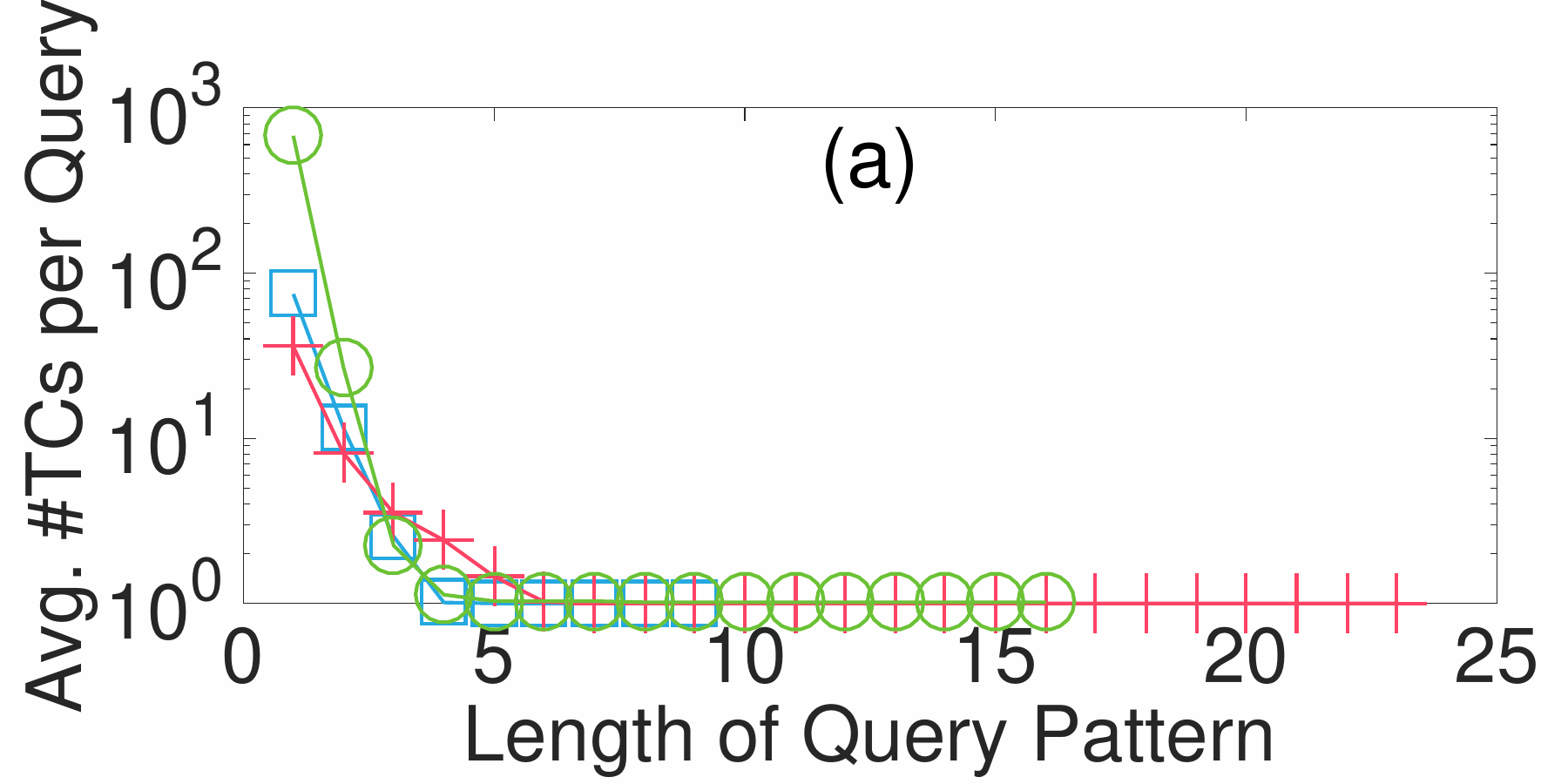}}
\subfigure{\includegraphics[width=\parawidthExpfour]{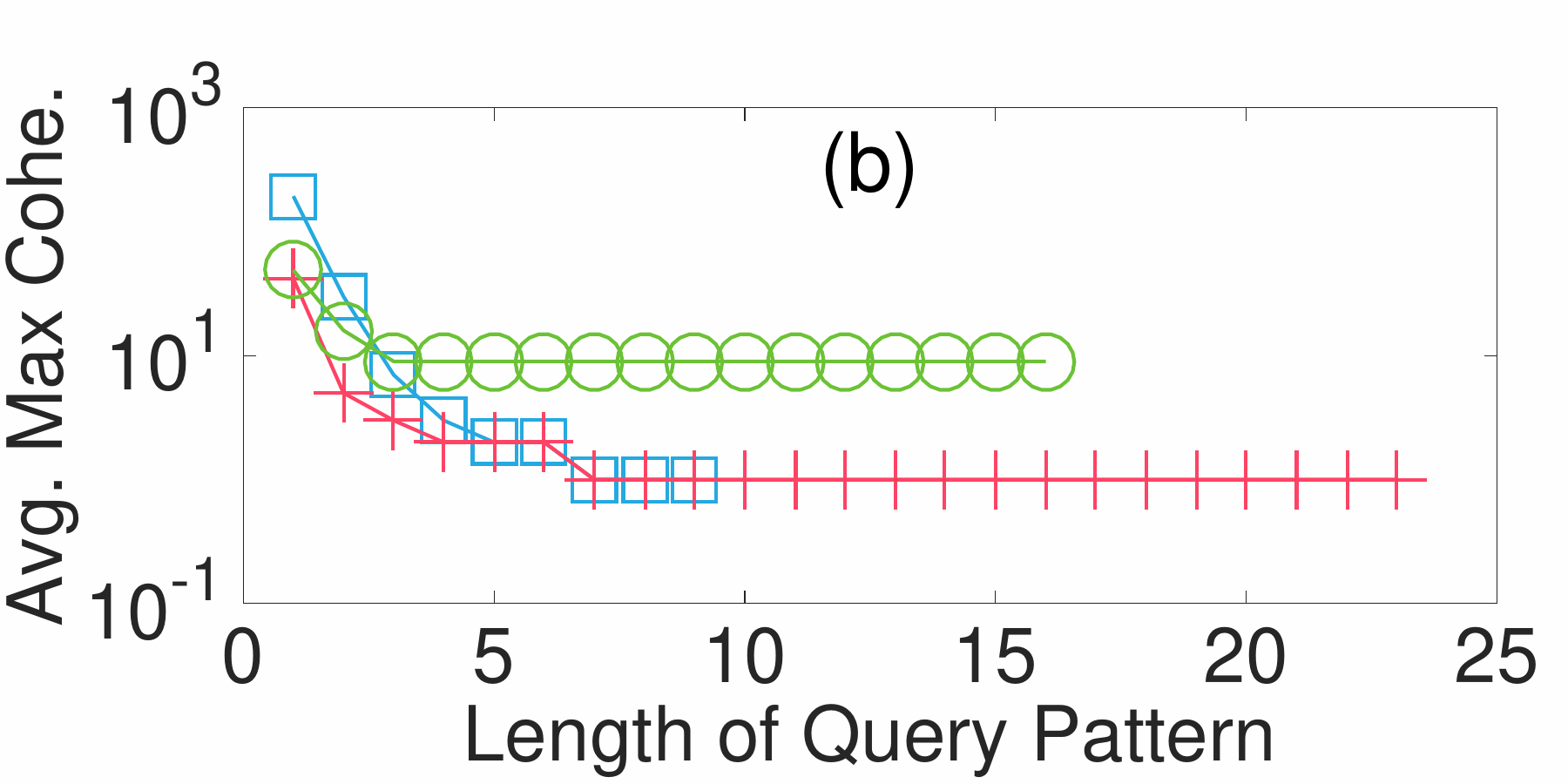}}

\vspace{-3mm}

\subfigure{\includegraphics[width=\parawidthExpfour]{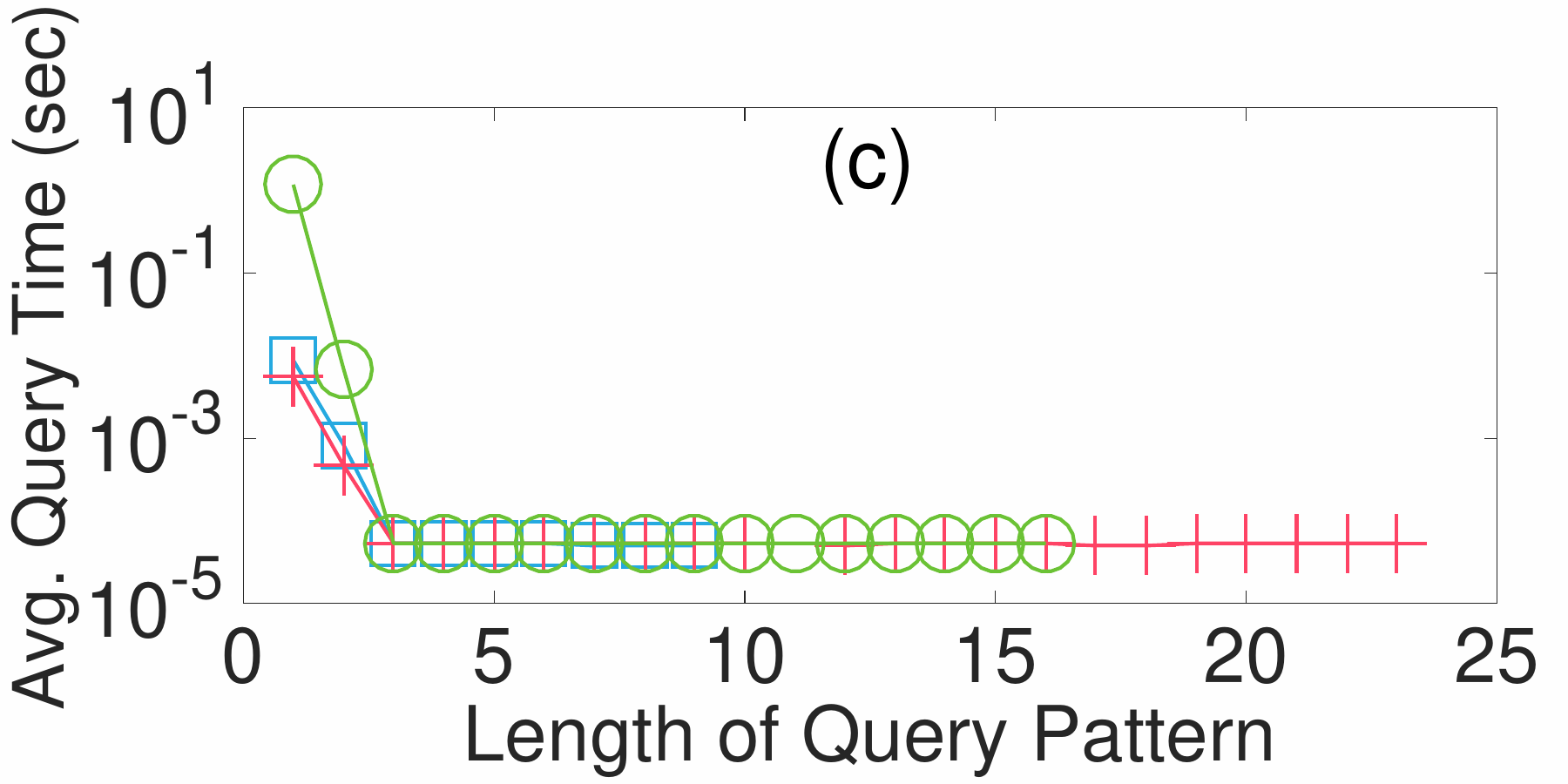}}
\subfigure{\includegraphics[width=\parawidthExpfour]{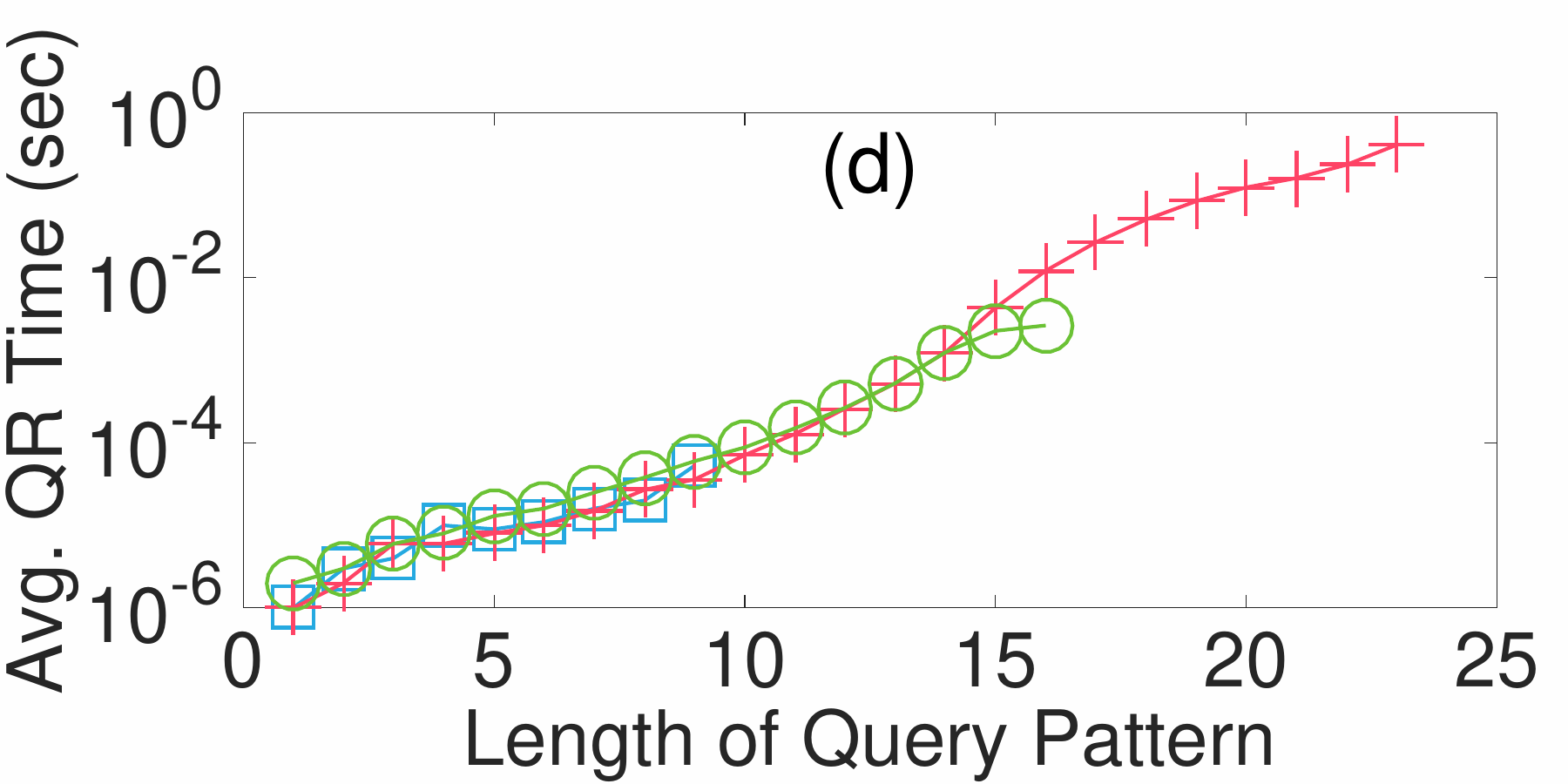}}

\vspace{-3mm}

\subfigure{\includegraphics[height=2.3mm]{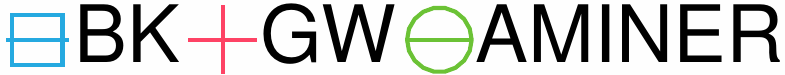}}

\vspace{-3mm}
\caption{The effect of the length of query pattern. ``Avg.'', ``Cohe.'' and ``QR'' are abbreviations for ``Average'', ``Cohesiveness'' and ``Query Recommendation'', respectively. }
\label{Fig:patlen_effect}
\end{figure}

\subsection{Performance of Theme Community Tree}
\label{Sec:eotci}
In this subsection, we analyze the performance of Theme Community Tree (TC-Tree), including the indexing efficiency of TC-Tree, the effect of the length of query pattern and the distribution of retrieved theme communities.

The indexing efficiency of TC-Tree in all DBNs is shown in Table~\ref{Table:iptct}. ``Indexing Time'' is the cost to build a TC-Tree; ``Memory'' is the peak memory usage when building a TC-Tree; ``\#Nodes'' is the number of nodes in a TC-Tree.

Building a TC-Tree is efficient in both Indexing Time and Memory.
For the largest DBN AMINER, TC-Tree scales up well and indexes more than 120 million nodes.

To analyze the effect of the length of query pattern on the querying performance, we generate a set of query patterns $\mathcal{B}$ by randomly sampling 1,000 nodes from the TC-Tree and using the patterns of the sampled nodes as query patterns. 
Denote by $\mathcal{B}^l=\{\mathbf{q} \mid \mathbf{q}\in\mathcal{B}, |\mathbf{q}|=l\}$ the set of length-$l$ query patterns in $\mathcal{B}$, and by $L$ the maximum length of all query patterns in $\mathcal{B}$, we have $\mathcal{B} = \mathcal{B}^1\cup \cdots \cup \mathcal{B}^L$.
We use the query patterns in $\mathcal{B}$ to query the TC-Tree and analyze the following querying performance.

The ``Avg.\ \#TCs per Query'' in Figure~\ref{Fig:patlen_effect}(a) is the average number of theme communities retrieved by the query patterns in each set of $\mathcal{B}^l, l\in\{1,\ldots,L\}$.
In Figure~\ref{Fig:patlen_effect}(a), the Avg.\ \#TCs per Query decreases quickly when the length of query pattern increases, this verifies our analysis in Section~\ref{Sec:tcfp} that a longer pattern is less likely to induce a theme community.
For query patterns with length larger than 1, the average number of retrieved theme communities is less than 100.  It is not a heavy burden for users to analyze such a small number of theme communities.

In Figure~\ref{Fig:patlen_effect}(b), the ``Avg.\ Max Cohe.'' is the average cohesiveness of the top-1 theme communities retrieved by the query patterns in each set of $\mathcal{B}^l, l\in\{1,\ldots,L\}$. The Avg.\ Max Cohe.\ decreases when the length of query pattern increases, because a longer pattern has a lower frequency, which limits the cohesiveness of the corresponding theme communities.

Figure~\ref{Fig:patlen_effect}(c) shows the average query time for the query patterns in each set of $\mathcal{B}^l, l\in\{1,\ldots,L\}$. 
For all DBNs, the average query time is less than 1 second when the length of pattern is 1, and decreases quickly to less than $0.1$ second when the length of query pattern increases.
The query time is dominated by the time to compute the theme communities and their cohesiveness using the decomposed maximal $(\mathbf{p}, \alpha)$-truss stored in a node of the TC-Tree.
When the length of query pattern increases, the size of the maximal $(\mathbf{p}, \alpha)$-truss reduces, thus the query time decreases.

Figure~\ref{Fig:patlen_effect}(d) shows the average query recommendation time for the query patterns in each set of $\mathcal{B}^l, l\in\{1,\ldots,L\}$. 
When the length of query pattern increases, the average query recommendation time increases, because a longer query pattern requires Algorithm~\ref{Alg:qr} to visit more TC-Tree nodes.
The query recommendation is highly efficient, and the time cost is less than 1 second on all DBNs.

\newcommand{\parawidthExpthree}{41.5mm}
\begin{figure}[t]
\vspace{-2mm}
\centering
\subfigure{\includegraphics[width=\parawidthExpthree]{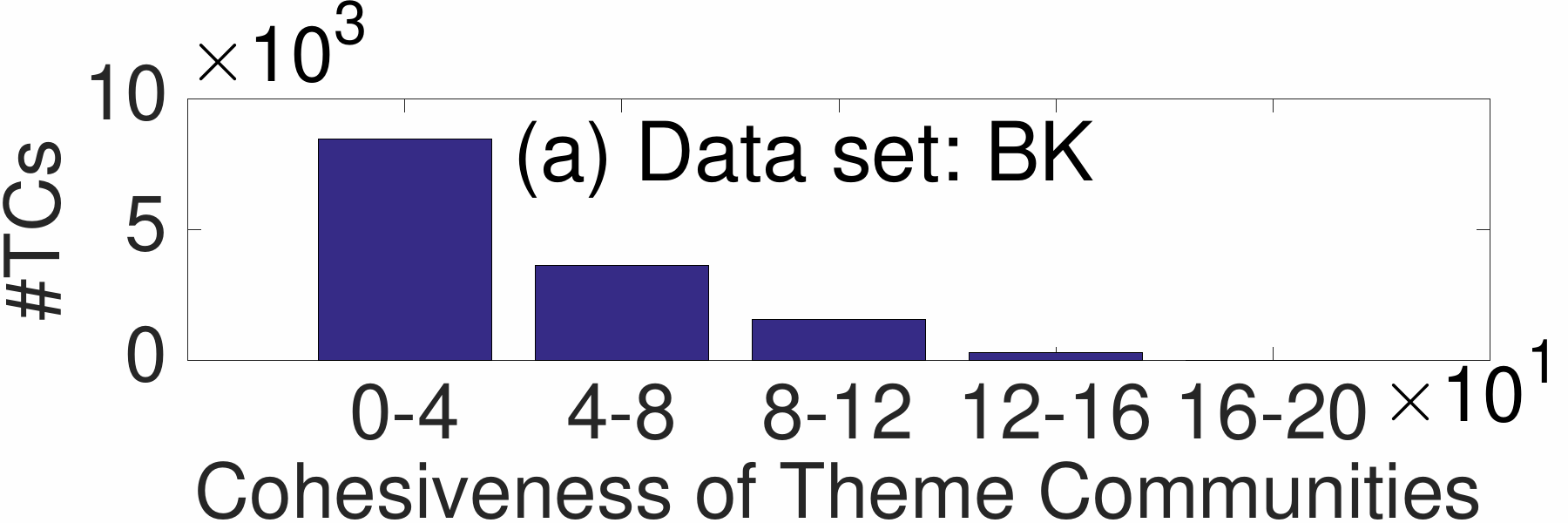}}
\subfigure{\includegraphics[width=\parawidthExpthree]{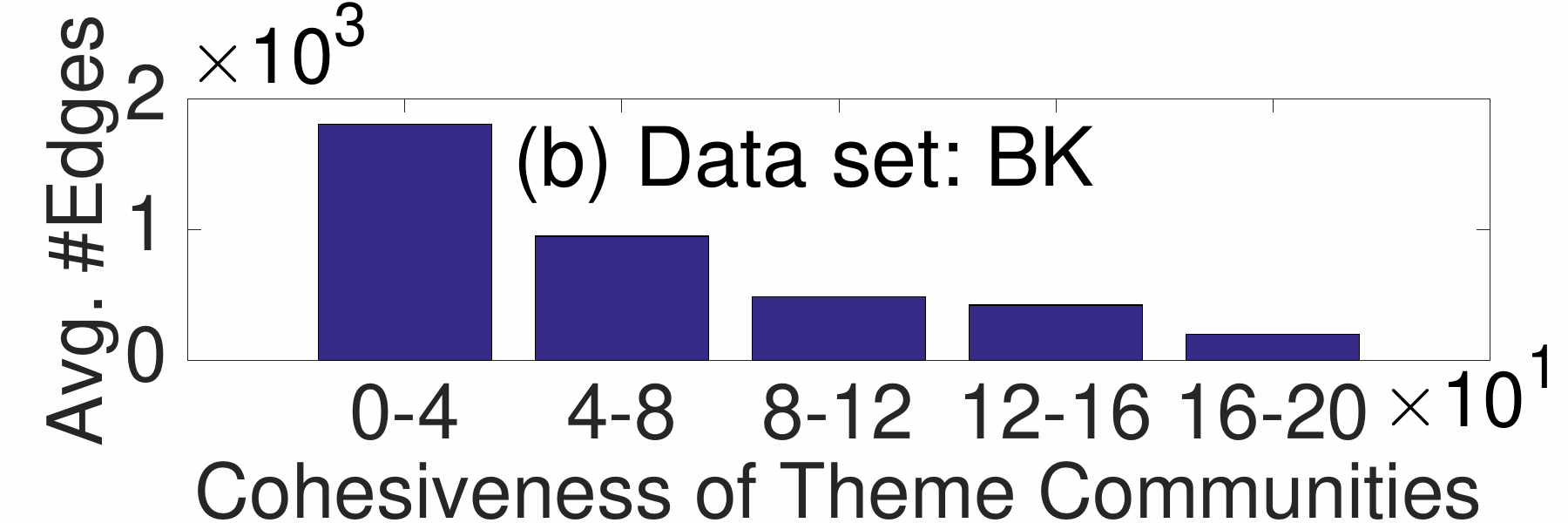}}

\vspace{-3mm}
\subfigure{\includegraphics[width=\parawidthExpthree]{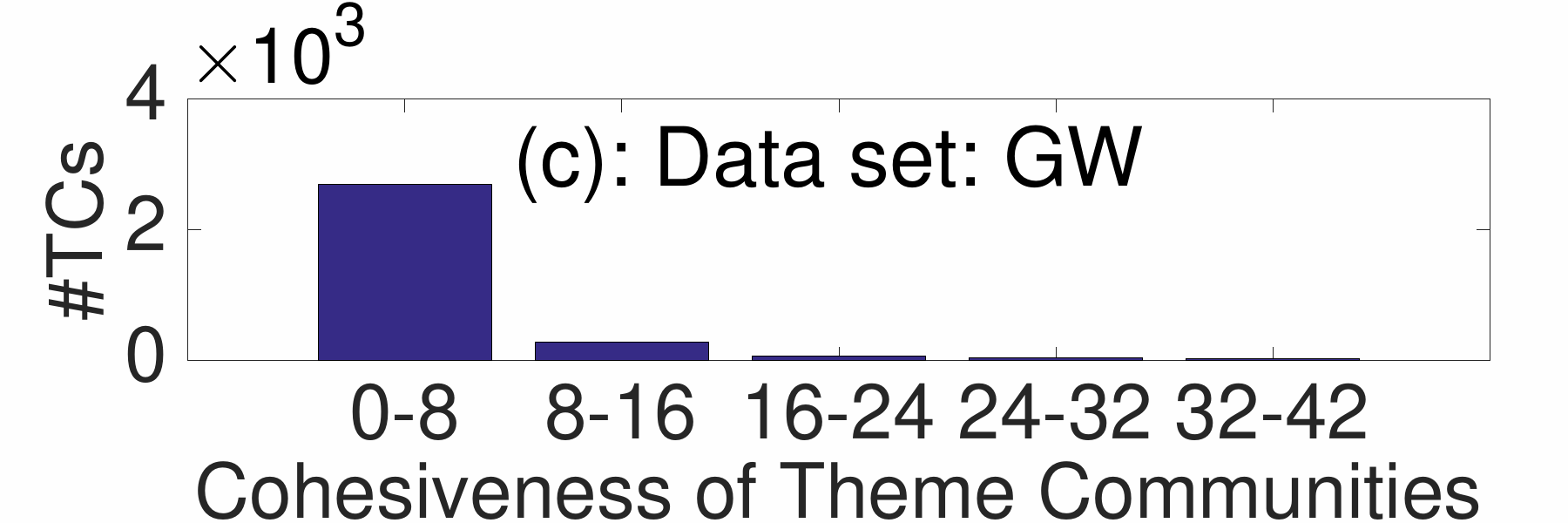}}
\subfigure{\includegraphics[width=\parawidthExpthree]{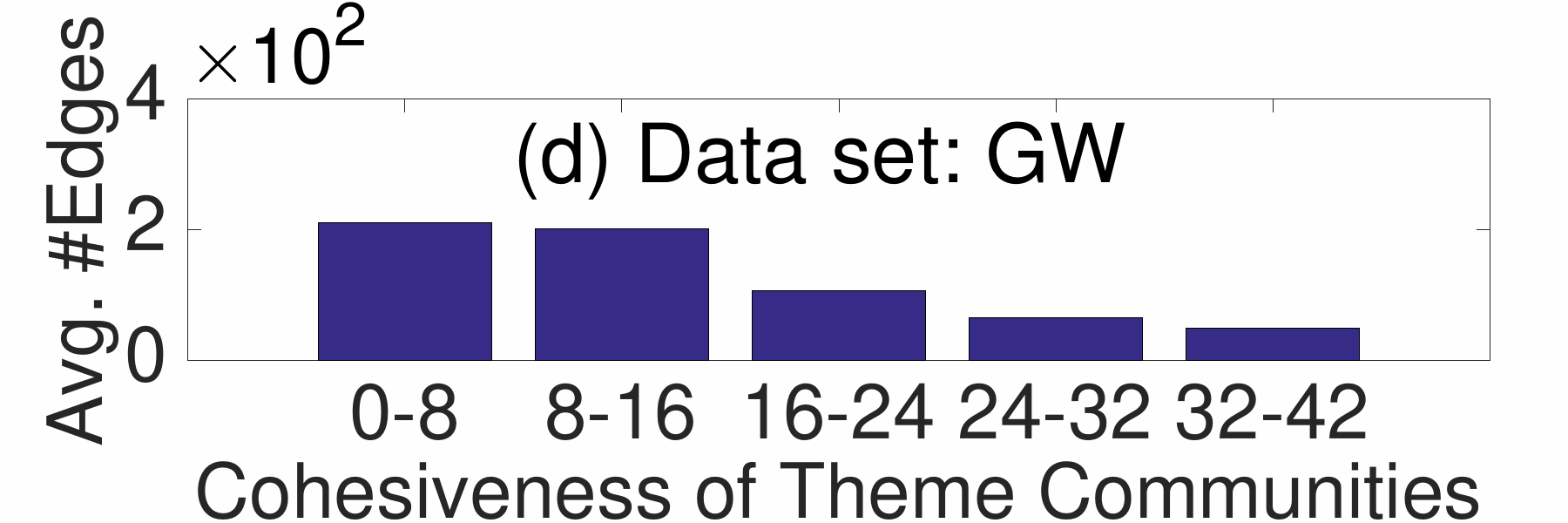}}

\vspace{-3mm}
\subfigure{\includegraphics[width=\parawidthExpthree]{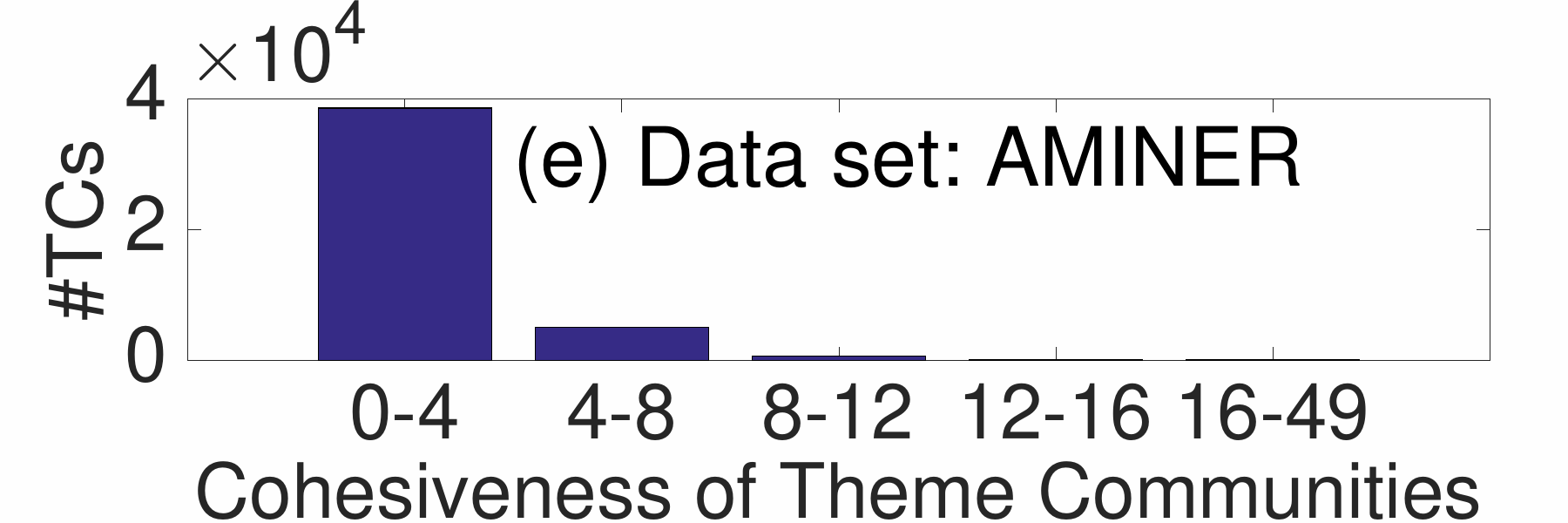}}
\subfigure{\includegraphics[width=\parawidthExpthree]{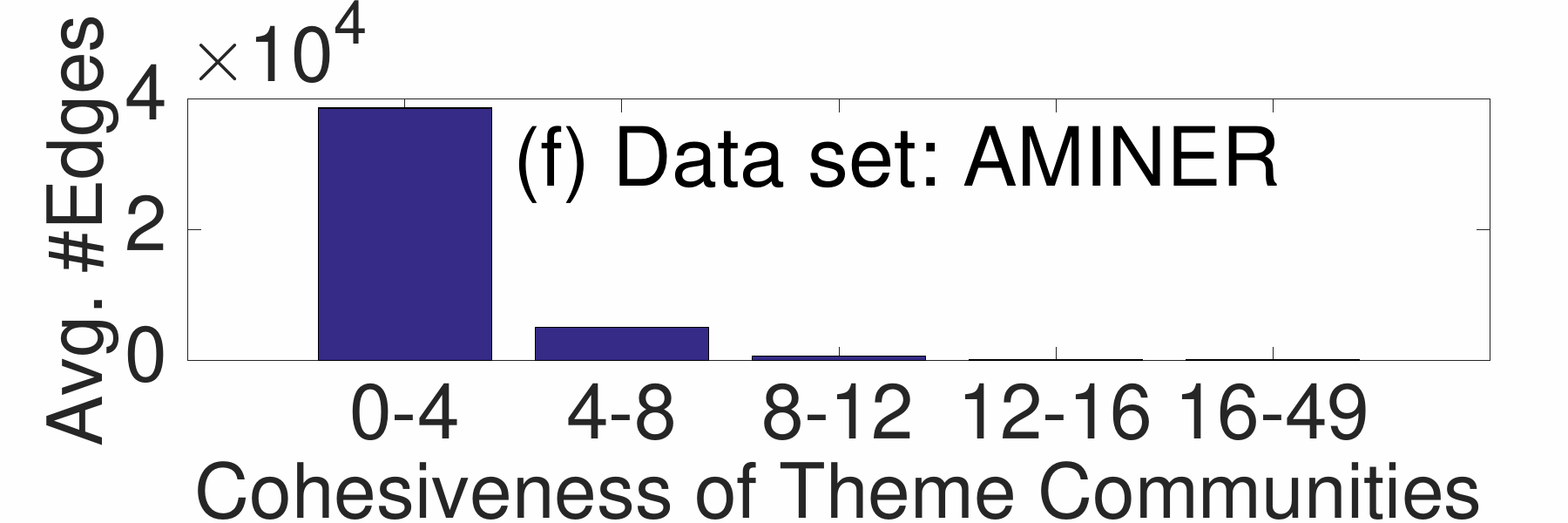}}

\vspace{-2mm}
\caption{The distribution of retrieved theme communities. Each bin on the x-axis represents an interval of the cohesiveness of theme communities. ``\#TCs'' is the number of theme communities in a bin. ``Avg. \#Edges'' is the average number of edges of all theme communities in a bin.}
\label{Fig:hists}
\end{figure}

Next, we query the TC-Tree by the query patterns in $\mathcal{B}$, and show the distribution of all retrieved theme communities in Figure~\ref{Fig:hists}. A large proportion of the retrieved theme communities have very small cohesiveness, and contain a larger number of edges. 
These theme communities are trivially induced by the patterns that are contained in a large number of vertex databases with low pattern frequencies, and they have low ranks in the retrieved ranked list of theme communities due to their small cohesiveness.

\nop{
Therefore, these theme communities have very low ranks in the retrieved rank list of theme communities due to their small cohesiveness, and they are less interesting than the relatively smaller theme communities that have large cohesiveness.
}

There are a lot of theme communities that have large cohesiveness and contain a small number of edges.
Those theme communities are usually induced by the patterns that frequently appear in the vertex databases of a small number of strongly connected vertices.
As demonstrated by the case study in Section~\ref{Sec:cs}, these theme communities are ranked high in the retrieved ranked list and often reveal interesting patterns and communities in DBNs.

\subsection{Comparison with Two Vertex Attributed Network Methods}
\label{sec:vanc}
In this subsection, we evaluate the theme community detection performance of TCFI and two community detection methods for Vertex Attributed Networks (VAN), such as CESNA~\cite{yang2013community} and SCI~\cite{wang2016semantic}. 
We set $\alpha=0$ for TCFI, and use the default settings of CESNA and SCI, respectively. 
Since SCI cannot efficiently process large networks, we sample one DBN with 10,000 edges from each of BK, GW and AMINER using the sampling method in Section~\ref{Sec:eop}.
Both CESNA and SCI cannot directly process a DBN, thus we extend them by first converting a DBN into a VAN, then applying them to detect communities from the VAN.

We adopt the following three types of \textbf{extensions}.

\textbf{Extension A}: we convert each vertex database into a set of items by taking the union of all transactions in it.
The pattern of a detected community is a set of items, which is used as the theme of the community.

\textbf{Extension B}: we regard a transaction as a set-valued item, and treat each vertex database as a set of set-valued items.
The pattern of a detected community is a set of transactions, where each transaction is used as a theme of the community.

\textbf{Extension C}: we convert a DBN in the same way as Extension B. Instead of using each transaction as a theme, we treat the union of the set of transactions as the theme of the detected community.

Denote by $D$ a detected community with a theme $\mathbf{p}$. 
We measure the quality of $D$ by the following metrics.

\textbf{Average Edge Strength (AES)}: 
denote by $\mathcal{E}$ and $\mathcal{V}$ the sets of edges and vertices of $D$, respectively. 
The \textbf{edge strength} of an edge $e_{ij}\in \mathcal{E}$ is measured by
$ES_{ij}=f_i(\mathbf{p})\times f_j(\mathbf{p})$. 
The AES of $D$ is 
\begin{equation}\nonumber
	AES = \frac{2\sum_{e_{ij}\in\mathcal{E}} ES_{ij}}{|\mathcal{V}|(|\mathcal{V}|-1)},
\end{equation}
which is the ratio between the sum of edge strength of all edges in $\mathcal{E}$ and the number of edges of a clique containing a number of $|\mathcal{V}|$ vertices. A large AES means the vertices in $\mathcal{V}$ have a high frequency of $\mathbf{p}$, and are densely connected by the edges in $\mathcal{E}$.
Obviously, a larger AES indicates a higher quality of the theme community.

\textbf{Cohesiveness (COHE)}: as defined in Definition~\ref{Def:cohe_theme_community}, the cohesiveness of $D$ is the minimum cohesion of its edges.
If COHE is 0, then either at least one vertex of $D$ does not contain $\mathbf{p}$ in its vertex database, or at least one edge of $D$ is not contained in any triangle. Thus, we say $D$ is an \textbf{invalid} theme community if its COHE is 0, and say $D$ is \textbf{valid} if its COHE is positive.
Obviously, a larger COHE indicates a higher quality of $D$.

\begin{figure}[t]
\centering
\subfigure{\includegraphics[width=\parawidthExpthree]{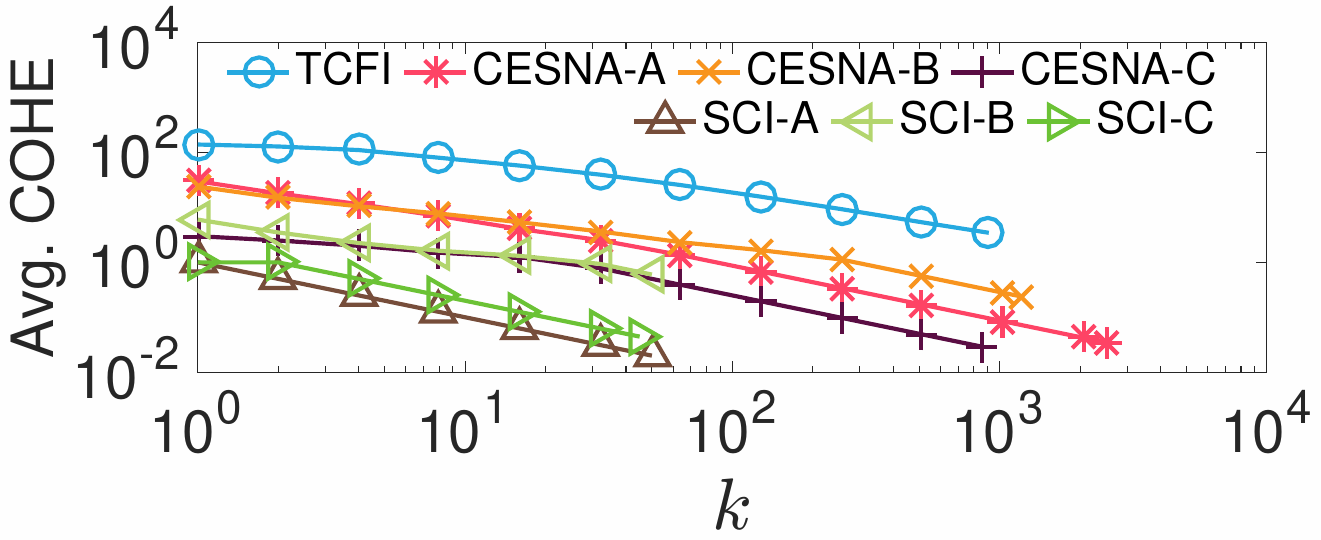}}
\subfigure{\includegraphics[width=\parawidthExpthree]{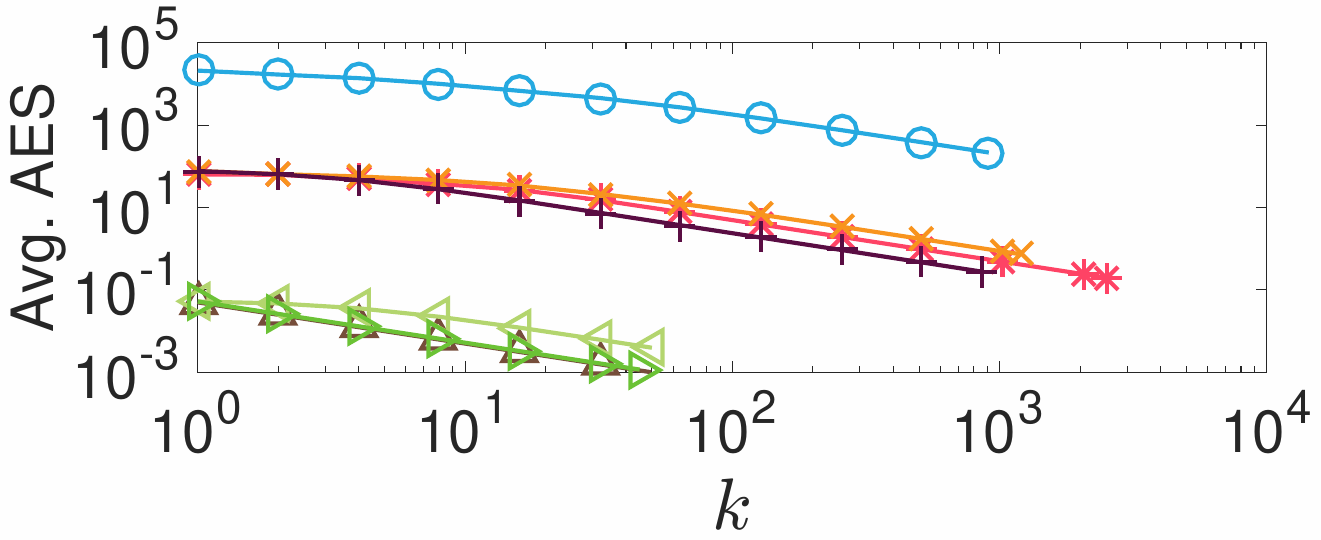}}

\vspace{-3mm}
\subfigure{\includegraphics[width=\parawidthExpthree]{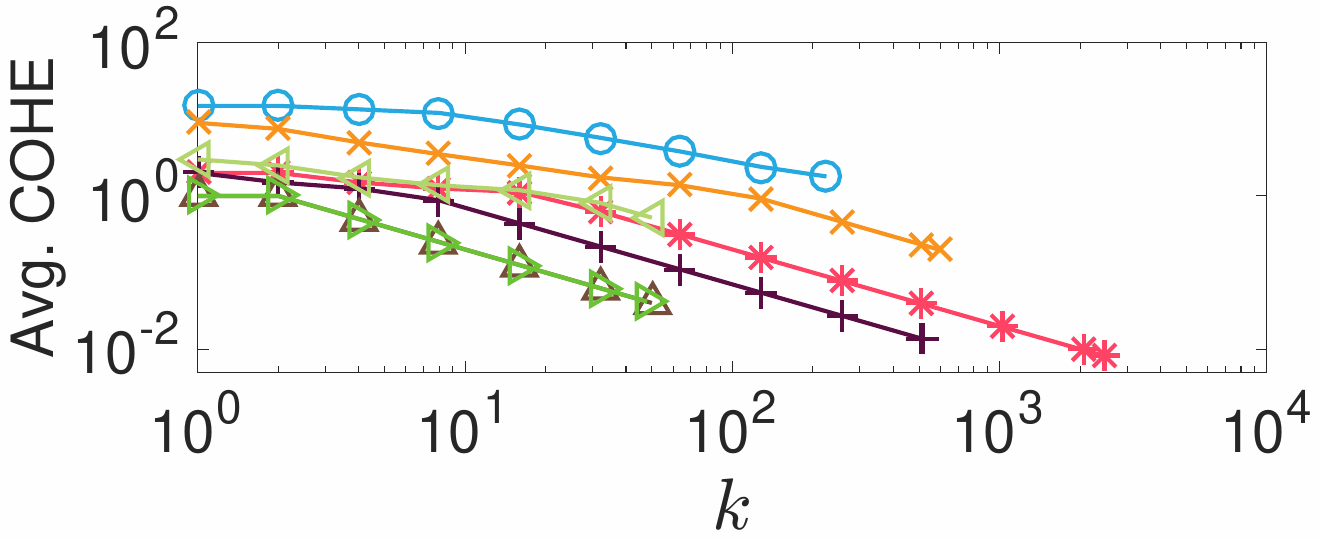}}
\subfigure{\includegraphics[width=\parawidthExpthree]{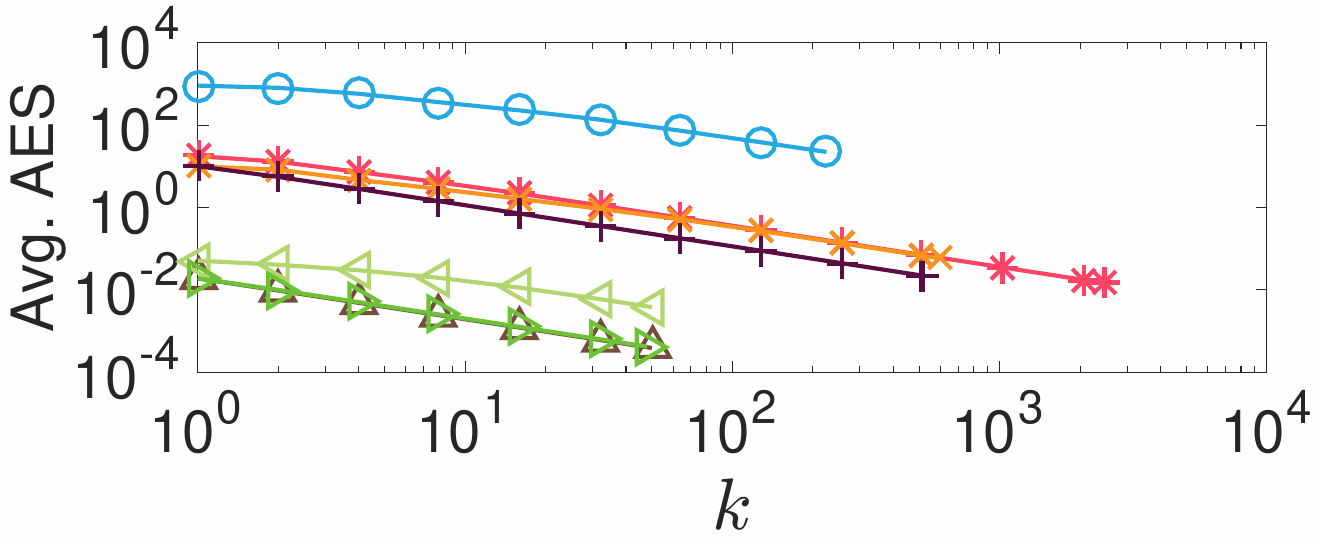}}

\vspace{-3mm}
\subfigure{\includegraphics[width=\parawidthExpthree]{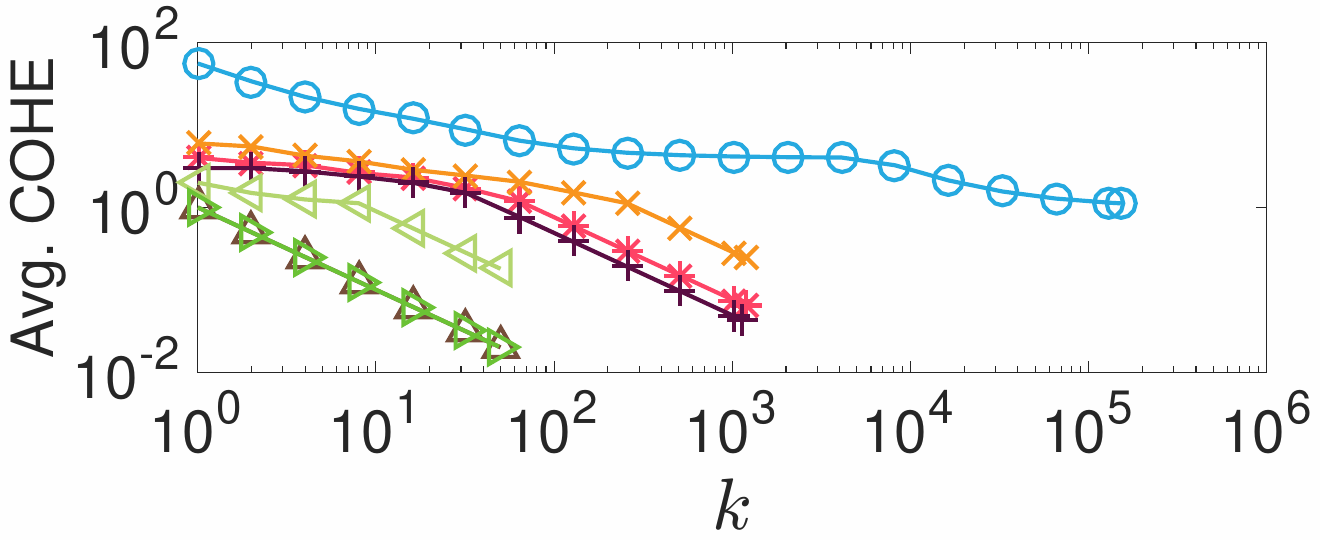}}
\subfigure{\includegraphics[width=\parawidthExpthree]{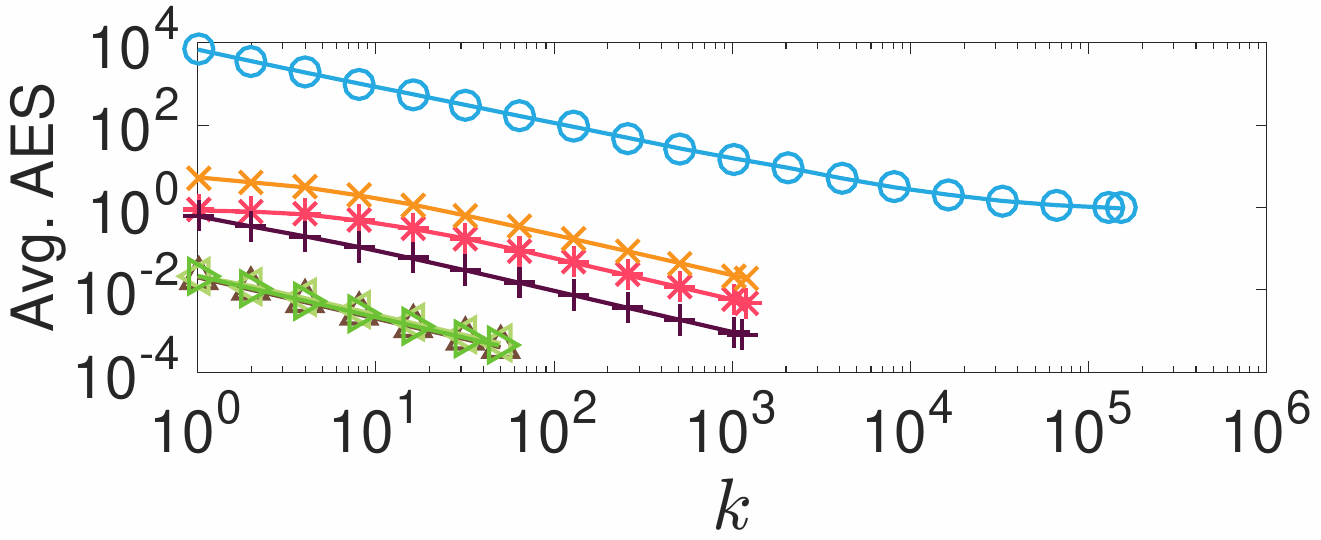}}

\vspace{-3mm}

\subfigure{\includegraphics[height=2.9mm]{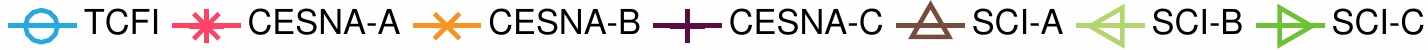}}

\vspace{-3mm}
\caption{The Average (Avg.) COHE and Avg. AES of top-$k$ detected communities.}
\label{Fig:cohe_aes}
\end{figure}

\begin{figure}[t]
\centering
\subfigure{\includegraphics[width=75mm]{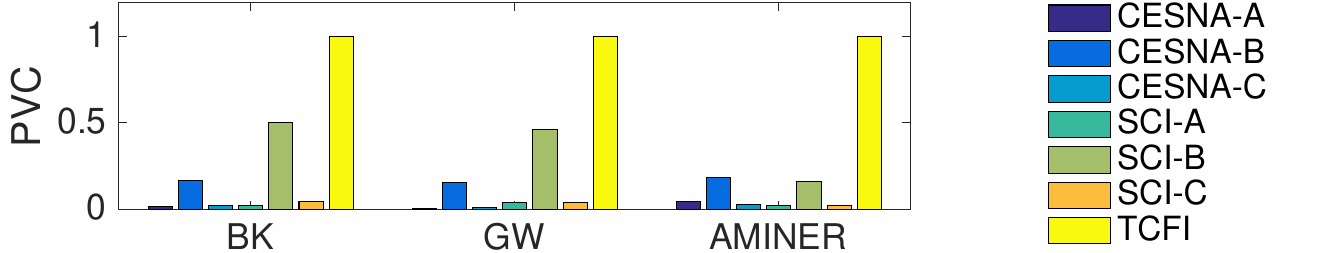}}

\vspace{-3mm}
\caption{The PVC performance.}
\label{Fig:pvc}
\end{figure}

Both AES and COHE are computed based on the absolute frequency of the theme $\mathbf{p}$.

Figure~\ref{Fig:cohe_aes} shows the Average (Avg.) COHE and the Avg. AES of the top-$k$ detected theme communities with the largest COHE and AES, respectively. Both CESNA and SCI do not achieve a good performance of Avg. COHE or Avg. AES for all types of extensions, because all the extensions lose the valuable information about item co-occurrences and pattern frequencies when converting a vertex database into a set of items.
Since TCFI effectively explores the item co-occurrence and pattern frequency of DBNs, it achieves the best performance on all data sets.

Figure~\ref{Fig:pvc} shows the \textbf{Proportion of Valid Community (PVC)} of all compared methods. Here, PVC is the proportion of valid theme communities in the set of all detected communities.
The PVC of CESNA and SCI is low for Extensions A and C, because both of them take the union of transactions, which generates a lot of new patterns that do not exist in the DBN.
These new patterns induce many detected communities, which are invalid in the DBN.
The extension B does not take the union of transactions, thus it achieves a higher PVC than Extensions A and C. However, since CESNA and SCI do not strictly require every vertex in the same community to have exactly the same pattern, nor do they require every edge of a community to be contained in one or more triangles, many communities detected by CESNA-B and SCI-B are invalid, thus their PVC is still much lower than TCFI.
The PVC of TCFI is always 1.0, which demonstrates the superior performance of TCFI in accurately finding valid theme communities.

In sum, a DBN naturally models the rich and valuable vertex information, such as item co-occurrence and pattern frequency, which is way beyond the limited descriptive power of VANs.
As a result, it is difficult to straightforwardly extend VAN-based methods, such as CESNA and SCI, to effectively enumerate the theme communities in DBNs.

\nop{
Denote by $\mathcal{D}=\{D_1, \ldots, D_m\}$ a set of detected theme communities. We measure the quality of $\mathcal{D}$ as follows.
}

\nop{
\textbf{Average COHE and Average AES}: 

\textbf{Proportion of Valid Communities (PVC)}: 
For any $D_i \in \mathcal{D}$, if the COHE of $D_i$ is 0, then $D_i$ does not contain any triangle, or at least one vertex of $D_i$ does not contain theme $\mathbf{p}$ in its vertex database. Therefore, we say $D_i$ is \textbf{invalid} if its COHE is 0, and say $D_i$ is \textbf{valid} if it has a positive COHE.
}
\nop{
\begin{equation}
	AES = \frac{2\sum_{e_{ij}\in\mathcal{E}} ES_{ij}}{|\mathcal{V}|(|\mathcal{V}|-1)}
\end{equation}
}

\nop{
,
that is, the product of the absolute frequencies of theme $\mathbf{p}$ on vertices $v_i$ and $v_j$. 
}

\nop{
Using the above converting methods, we generate the following three types of extensions for CESNA and SCI.

\textbf{Type-A}: Convert a DBN into a VAN 
}

\section{Conclusions and Future Work}
\label{sec:con}

In this paper, we tackle the novel problem of finding theme communities from DBNs. 
We first introduce the novel concept of DBN, which is a natural abstraction of many real world networks.
Then, we propose TCFI and TC-Tree that efficiently discover and index hundreds of millions of theme communities in large DBNs.
As demonstrated by extensive experiments, TCFI and TC-Tree are highly efficient and scalable.
As future work, we will extend TCFI and TC-Tree to find theme communities from edge DBNs, where each edge is associated with a database that describes relationships between vertices.

\nop{
\begin{table}[t]
\caption{The sets of keywords for Figure~\ref{Fig:case_study}.}
\centering
\label{Table:listofpat}
\begin{tabular}{|p{2mm} p{78mm}|}
\hline
$p_1:$ &  data mining, sequential pattern \\ \hline
$p_2:$ & data mining, sequencial pattern, intrusion detection  \\ \hline
$p_3:$ & data mining, search space, complete set, pattern mining \\ \hline
$p_4:$ & data mining, sensitive information, privacy protection \\ \hline
$p_5:$ & principal component analysis, linear discriminant analysis, dimensionality reduction, component analysis \\ \hline
$p_6:$ & Image retrieval, image database, relevance feedback\\ \hline
\end{tabular}
\end{table}
}

\bibliographystyle{abbrv}
\bibliography{Reference} 

\begin{thebibliography}{10}

\bibitem{agrawal1994fast}
R.~Agrawal and R.~Srikant.
\newblock Fast algorithms for mining association rules.
\newblock {\em PVLDB}, 1215:487--499, 1994.

\bibitem{Aminer_data}
AMINER.
\newblock \url{https://aminer.org/citation}, 2010.

\bibitem{assal2010multi}
A.~F.~A. Assal, R.~M.~M. Lim, and M.-H.~L. Lee.
\newblock Multi-user on-line real-time virtual social networks based upon
  communities of interest for entertainment, information or e-commerce
  purposes, Dec. 2010.
\newblock US Patent 7,853,881.

\bibitem{balasubramanyan2011block}
R.~Balasubramanyan and W.~W. Cohen.
\newblock Block-{LDA}: Jointly modeling entity-annotated text and entity-entity
  links.
\newblock In {\em Proceedings of the SIAM International Conference on Data
  Mining}, volume~11, pages 450--461, 2011.

\bibitem{berlingerio2013abacus}
M.~Berlingerio, F.~Pinelli, and F.~Calabrese.
\newblock {ABACUS}: frequent pattern mining-based community discovery in
  multidimensional networks.
\newblock {\em Data Mining and Knowledge Discovery}, 27(3):294--320, 2013.

\bibitem{BKGW_data}
E.~Cho, S.~A. Myers, and J.~Leskovec.
\newblock Friendship and mobility: user movement in location-based social
  networks.
\newblock In {\em Proceedings of the ACM SIGKDD International Conference on
  Knowledge Discovery and Data Mining}, pages 1082--1090, 2011.

\bibitem{chu2015alid}
L.~Chu, S.~Wang, S.~Liu, Q.~Huang, and J.~Pei.
\newblock {ALID}: Scalable dominant cluster detection.
\newblock {\em PVLDB}, 8(8):826--837, 2015.

\bibitem{cohen2008barycentric}
J.~Cohen.
\newblock Barycentric graph clustering.
\newblock {\em Oregon Health Science University}, 2008.

\bibitem{cohen2008trusses}
J.~Cohen.
\newblock Trusses: Cohesive subgraphs for social network analysis.
\newblock {\em National Security Agency Technical Report}, page~16, 2008.

\bibitem{cohen2009graph}
J.~Cohen.
\newblock Graph twiddling in a mapreduce world.
\newblock {\em Computing in Science \& Engineering}, 11(4):29--41, 2009.

\bibitem{combe2012combining}
D.~Combe, C.~Largeron, E.~Egyed-Zsigmond, and M.~G{\'e}ry.
\newblock Combining relations and text in scientific network clustering.
\newblock In {\em International Conference on Advances in Social Networks
  Analysis and Mining}, pages 1248--1253, 2012.

\bibitem{cruz2011semantic}
J.~D. Cruz, C.~Bothorel, and F.~Poulet.
\newblock Semantic clustering of social networks using points of view.
\newblock In {\em CORIA}, pages 175--182, 2011.

\bibitem{dang2012community}
T.~Dang and E.~Viennet.
\newblock Community detection based on structural and attribute similarities.
\newblock In {\em International Conference on Digital Society}, pages 7--12,
  2012.

\bibitem{gunnemann2011db}
S.~G{\"u}nnemann, B.~Boden, and T.~Seidl.
\newblock {DB-CSC}: a density-based approach for subspace clustering in graphs
  with feature vectors.
\newblock In {\em Joint European Conference on Machine Learning and Knowledge
  Discovery in Databases}, pages 565--580, 2011.

\bibitem{gunopulos2003discovering}
D.~Gunopulos, R.~Khardon, H.~Mannila, S.~Saluja, H.~Toivonen, and R.~S. Sharma.
\newblock Discovering all most specific sentences.
\newblock {\em ACM Transactions on Database Systems}, 28(2):140--174, 2003.

\bibitem{han2000mining}
J.~Han, J.~Pei, and Y.~Yin.
\newblock Mining frequent patterns without candidate generation.
\newblock In {\em Proceedings of the ACM SIGMOD International Conference on
  Management of Data}, volume~29, pages 1--12, 2000.

\bibitem{hu2013utilizing}
A.~L. Hu and K.~C. Chan.
\newblock Utilizing both topological and attribute information for protein
  complex identification in ppi networks.
\newblock {\em IEEE Transactions on Computational Biology and Bioinformatics},
  10(3):780--792, 2013.

\bibitem{huang2014querying}
X.~Huang, H.~Cheng, L.~Qin, W.~Tian, and J.~X. Yu.
\newblock Querying k-truss community in large and dynamic graphs.
\newblock In {\em Proceedings of the ACM SIGMOD International Conference on
  Management of Data}, pages 1311--1322, 2014.

\bibitem{huang2017attribute}
X.~Huang and L.~V. Lakshmanan.
\newblock Attribute-driven community search.
\newblock {\em PVLDB}, 10(9):949--960, 2017.

\bibitem{huang2015method}
Y.~Huang, Y.~Tang, C.~Li, Z.~Wu, and H.~Dong.
\newblock A method for latent-friendship recommendation based on community
  detection in social network.
\newblock In {\em Web Information System and Application Conference}, pages
  3--8. IEEE, 2015.

\bibitem{huang2013commerce}
Z.~Huang and M.~Benyoucef.
\newblock From e-commerce to social commerce: A close look at design features.
\newblock {\em Electronic Commerce Research and Applications}, 12(4):246--259,
  2013.

\bibitem{luce1950connectivity}
R.~D. Luce.
\newblock Connectivity and generalized cliques in sociometric group structure.
\newblock {\em Psychometrika}, 15(2):169--190, 1950.

\bibitem{mislove2010you}
A.~Mislove, B.~Viswanath, K.~P. Gummadi, and P.~Druschel.
\newblock You are who you know: inferring user profiles in online social
  networks.
\newblock In {\em International Conference on Web Search and Data Mining},
  pages 251--260. ACM, 2010.

\bibitem{moosavi2016community}
S.~A. Moosavi, M.~Jalali, N.~Misaghian, S.~Shamshirband, and M.~H. Anisi.
\newblock Community detection in social networks using user frequent pattern
  mining.
\newblock {\em Knowledge and Information Systems}, pages 1--28, 2016.

\bibitem{moser2009mining}
F.~Moser, R.~Colak, A.~Rafiey, and M.~Ester.
\newblock Mining cohesive patterns from graphs with feature vectors.
\newblock In {\em Proceedings of the SIAM International Conference on Data
  Mining}, volume~9, pages 593--604, 2009.

\bibitem{ng2002spectral}
A.~Y. Ng, M.~I. Jordan, and Y.~Weiss.
\newblock On spectral clustering: Analysis and an algorithm.
\newblock In {\em Advances in Neural Information Processing Systems}, pages
  849--856, 2002.

\bibitem{pavan2007dominant}
M.~Pavan and M.~Pelillo.
\newblock Dominant sets and pairwise clustering.
\newblock {\em IEEE Transactions on Pattern Analysis and Machine Intelligence},
  29(1):167--172, 2007.

\bibitem{prado2013mining}
A.~Prado, M.~Plantevit, C.~Robardet, and J.-F. Boulicaut.
\newblock Mining graph topological patterns: Finding covariations among vertex
  descriptors.
\newblock {\em IEEE Transactions on Knowledge and Data Engineering},
  25(9):2090--2104, 2013.

\bibitem{rymon1992search}
R.~Rymon.
\newblock Search through systematic set enumeration.
\newblock {\em Technical Reports (CIS)}, page 297, 1992.

\bibitem{seidman1983network}
S.~B. Seidman.
\newblock Network structure and minimum degree.
\newblock {\em Social Networks}, 5(3):269--287, 1983.

\bibitem{shi2000normalized}
J.~Shi and J.~Malik.
\newblock Normalized cuts and image segmentation.
\newblock {\em IEEE Transactions on Pattern Analysis and Machine Intelligence},
  22(8):888--905, 2000.

\bibitem{steinhaeuser2008community}
K.~Steinhaeuser and N.~V. Chawla.
\newblock Community detection in a large real-world social network.
\newblock In {\em Social computing, behavioral modeling, and prediction}, pages
  168--175. Springer, 2008.

\bibitem{wang2012truss}
J.~Wang and J.~Cheng.
\newblock Truss decomposition in massive networks.
\newblock {\em PVLDB}, 5(9):812--823, 2012.

\bibitem{wang2016semantic}
X.~Wang, D.~Jin, X.~Cao, L.~Yang, and W.~Zhang.
\newblock Semantic community identification in large attribute networks.
\newblock In {\em the AAAI Conference on Artificial Intelligence}, 2016.

\bibitem{huang2016truss}
H.~Xin, L.~Wei, and L.~V.~S. Lakshmanan.
\newblock Truss decomposition of probabilistic graphs: Semantics and
  algorithms.
\newblock In {\em Proceedings of the ACM SIGMOD International Conference on
  Management of Data}, pages 77--90, 2016.

\bibitem{yang2015defining}
J.~Yang and J.~Leskovec.
\newblock Defining and evaluating network communities based on ground-truth.
\newblock {\em Knowledge and Information Systems}, 42(1):181--213, 2015.

\bibitem{yang2013community}
J.~Yang, J.~McAuley, and J.~Leskovec.
\newblock Community detection in networks with node attributes.
\newblock In {\em International Conference on Data Mining}, pages 1151--1156,
  2013.

\bibitem{zhou2009graph}
Y.~Zhou, H.~Cheng, and J.~X. Yu.
\newblock Graph clustering based on structural/attribute similarities.
\newblock {\em PVLDB}, 2(1):718--729, 2009.

\end{thebibliography}



\balance

\end{document}